\documentclass[letterpaper, 11pt]{amsart}

\usepackage[margin=1.1in]{geometry}
\usepackage{color}
\usepackage{graphicx}
\usepackage{tikz-cd}

\usepackage{caption}

\usepackage{hyperref}

\usepackage{amsmath}

\usepackage{amsthm}

\usepackage{enumerate}

\newtheoremstyle{definition}
{3pt} % Space above
{3pt} % Space below
{} % Body font
{} % Indent amount
{\bfseries} % Theorem head font
{.} % Punctuation after theorem head
{.5em} % Space after theorem head
{} % Theorem head spec (can be left empty, meaning `normal')

\newtheorem{thm}{Theorem}[section]
\newtheorem{theorem}[thm]{Theorem}
\newtheorem{lemma}[thm]{Lemma}
\newtheorem{prop}[thm]{Proposition}
\newtheorem{proposition}[thm]{Proposition}
\newtheorem{cor}[thm]{Corollary}
\newtheorem{defn}[thm]{Definition}

\newtheorem{conjecture}[thm]{Conjecture}

\newtheorem{remark}[thm]{Remark}

\theoremstyle{definition}

\newcommand{\Ai}{\mathrm{Ai}}
\newcommand{\Airy}{\text{Airy}}

\renewcommand{\Re}{\text{Re}}

\renewcommand{\i}{\textbf{i}}

\newcommand{\E}{\mathbb{E}}
\newcommand{\R}{\mathbb{R}}
\newcommand{\eps}{\varepsilon}

\newcommand*{\defeq}{\mathrel{\vcenter{\baselineskip0.5ex \lineskiplimit0pt
                     \hbox{\scriptsize.}\hbox{\scriptsize.}}}%
                     =}

\newcommand{\ur}{
	\begin{tikzpicture}
		[scale=.7,very thick]
		\def\dd{.3}
		\draw (0,0)--++(0,\dd)--++(\dd,0);
\end{tikzpicture}\hspace{.8pt}}

\newcommand{\corner}[1][]{\begin{tikzpicture}[baseline=(current bounding box.center)]
\draw[thick] (0.5,.8)--(.8,.8);
\draw[thick] (1,1)--(1,1.3);
\node at (1,.8) {$\mathrm{O}$} ;
\node at (1,1.5) {$\mathrm{H}$} ;
\node at (.3,.8) {$\mathrm{H}$} ;
\end{tikzpicture} }

\newcommand{\other}[1][]{\begin{tikzpicture}[baseline=(current bounding box.center)]
\draw[thick] (1.5,.8)--(1.2,.8);
\draw[thick] (1,1)--(1,1.3);
\node at (1,.8) {$\mathrm{O}$} ;
\node at (1,1.5) {$\mathrm{H}$} ;
\node at (1.7,.8) {$\mathrm{H}$} ;
\end{tikzpicture} }

\newcommand{\straight}[1][]{
\begin{tikzpicture}[baseline={([yshift={-\ht\strutbox *1.15}]current bounding box.north)}]
\draw[thick] (0.5,1)--(.8,1);
\draw[thick] (1.2,1)--(1.5,1);
\node at (1,1) {$\mathrm{O}$} ;
\node at (1.7,1) {$\mathrm{H}$} ;
\node at (.3,1) {$\mathrm{H}$} ;
\end{tikzpicture} }

\newcommand{\upright}[1][]{\begin{tikzpicture}[baseline=(current bounding box.center)]
\draw[thick] (1.5,.8)--(1.2,.8);
\draw[thick] (1,.6)--(1,.3);
\node at (1,.8) {$\mathrm{O}$} ;
\node at (1,.1) {$\mathrm{H}$} ;
\node at (1.7,.8) {$\mathrm{H}$} ;
\end{tikzpicture} }

\newcommand{\up}[1][]{\begin{tikzpicture}[baseline=(current bounding box.center)]
\draw[thick] (1,1)--(1,1.3);
\draw[thick] (1,.6)--(1,.3);
\node at (1,.8) {$\mathrm{O}$} ;
\node at (1,.1) {$\mathrm{H}$} ;
\node at (1,1.5) {$\mathrm{H}$} ;
\end{tikzpicture} }

\thanks{}

\title{six-vertex Model and Random Matrix Distributions}

\date{\today}

\author{Vadim Gorin}
\address[Vadim Gorin]{University of California, Berkeley}
\email{vadicgor@gmail.com}

\author{Matthew Nicoletti}
\address[Matthew Nicoletti]{Massachusetts Institute of Technology}
\email{mnicolet@mit.edu}

\thanks{We are grateful to Alice Guionnet and Karol Kozlowski, who organized the ENS Lyon school ``The multiple
facets of the six-vertex model'': lecturing there led us to writing this text. We thank Alexei Borodin for helpful discussions. We thank the anonymous referee for their comments. The work of V.G.\ was partially supported by the NSF grant DMS - 2152588.}

\begin{document}
\maketitle

\begin{abstract}
We survey the connections between the six-vertex (square ice) model of 2d statistical mechanics and random matrix theory.  We highlight the same universal probability distributions appearing on both sides, and also indicate related open questions and conjectures. We present full proofs of two asymptotic theorems for the six-vertex model: in the first one  the Gaussian Unitary Ensemble and GUE-corners process appear; the second one leads to the Tracy-Widom distribution $F_2$. While both results are not new, we found shorter transparent proofs for this text. On our way we introduce the key tools in the study of the six-vertex model, including the Yang-Baxter equation and the Izergin-Korepin formula.
\end{abstract}

%In this article, we study the \emph{universal} probability distributions which describe statistics of both random six-vertex configurations and random matrices in the large scale limit. Using techniques from integrable probability, we provide full proofs of the emergence of these distributions in the six-vertex model in special cases. In particular, we prove that the \emph{GUE corners process} appears near the tangency point with domain wall boundary conditions at $\Delta = 0$, and we also show that the one point height fluctuations in the stochastic six-vertex model with step initial data converge to the \emph{Tracy--Widom GUE distribution} $F_2$.

\section{Introduction}

\sloppy

This text is about connections between two seemingly unrelated subareas of the probability theory and we start by introducing these two areas.

\subsection{Random configurations of square ice}

The six-vertex model is one of the celebrated lattice models of statistical mechanics. Its study was initiated by Pauling in 1935~\cite{pauling1935structure} with the goal of computing numeric characteristics (e.g., the residual entropy) of the real-world ice: the six-vertex model can be seen as a model for ``two-dimensional ice'' because its state space can be defined as allowed configurations of water molecules arranged in a finite piece of the square lattice.

On the infinite lattice~$\mathbb{Z}^2$ the definition of the model begins with the placement of atoms~``$\mathrm{O}$" at the vertices and atoms~``$\mathrm{H}$" at the edges connecting adjacent vertices of the grid. A configuration of the model is a matching of these atoms into molecules of water,~$\mathrm{H}_2 \mathrm{O}$, such that each atom is matched and each~$\mathrm{O}$ is matched to two neighboring~$\mathrm{H}$'s out of four. Note that there are ${4 \choose 2} = 6$ ways to choose two neighbors out of the four, hence, there are six possible local configurations of a molecule at a vertex, leading to the name ``six-vertex model". We assign six weights~$ a_1,a_2,b_1,b_2,c_1,c_2$ to the six allowed local configurations.

\begin{figure}[h]
\centering
\includegraphics[scale=.8]{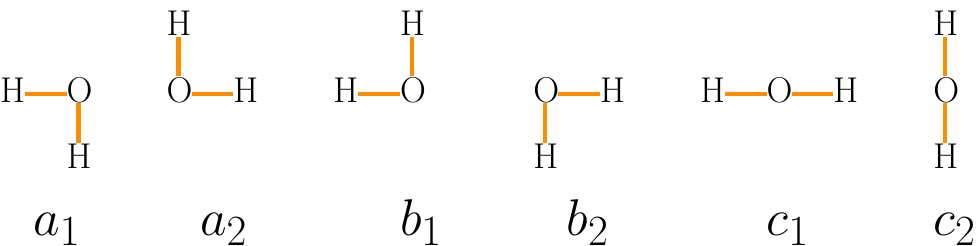}
\caption{The six local configurations and their weights.}
\label{fig:six_vs_weights}
\end{figure}

These six configurations and their weights are shown in Figure~\ref{fig:six_vs_weights}. Focusing on a finite subdomain of the infinite grid, the configuration space consists of choices of local configurations at each lattice site of this subdomain which are consistent with each other. See Figure~\ref{fig:six_v_configs} for an example. There are finitely many configurations in a finite domain and we would like to assign probabilities to them. For a configuration or state $S$, we set
\begin{equation}\label{eqn:bm}
\text{Prob}(S) = \frac{1}{Z} \prod_{(i,j)} \text{weight}(i,j; S),
\end{equation}
where the product goes over all vertices in our domain and the weight is the local weight of the configuration $S$ at site~$(i, j)$, so in particular ``weight" is one of the six numbers $ a_1,a_2,b_1,b_2,c_1,c_2$. The normalization constant $Z$ in \eqref{eqn:bm} is called the \emph{partition function}; it is defined in such a way as to guarantee that the probabilities sum up to $1$.

There are thousands of papers in mathematical and
theoretical physics literature devoted to
the study of the six-vertex model. The
early influential results concerned the
computation of the partition function $Z$
for the model on the torus going back to
the seminal work of Lieb in 1967~\cite{Lieb1967SixVertex}; see also
the extensive survey of Lieb and Wu~\cite{ferroelectric_models}.
 Many other developments in the study
 of the six-vertex model are covered
 in the classical 1982 textbook by Baxter~\cite{baxter2007exactly}, and we
 also refer to \cite{reshetikhin2010lectures} for a more recent review. Remarkably, in the last ten years the (planar) square ice was claimed to be observed experimentally, see the Nature article~\cite{algara2015square} and discussion in \cite{zhou2015observation}.

\medskip

The central mathematical question we are concerned with is: If a domain is very large, how does the random configuration typically look like?

\begin{figure}[t]
\centering
\includegraphics[width=0.42\linewidth]{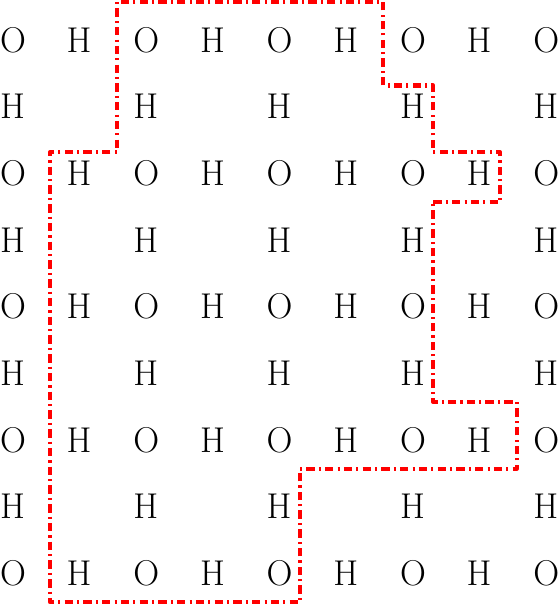}\hfill
\includegraphics[width=0.42\linewidth]{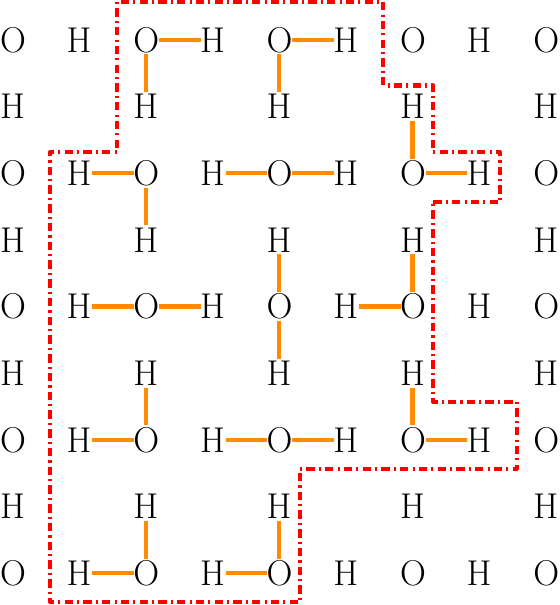}
\caption{A domain for the six-vertex model, and a consistent choice of local configurations, i.e. a state of the six-vertex model in this domain.}
\label{fig:six_v_configs}
\end{figure}

\subsection{Random matrices}

Let us consider an $N \times N$ self-adjoint random matrix
 $$A =
\begin{pmatrix}
a_{1 1} & \cdots & a_{1 N} \\
\vdots & & \vdots \\
a_{N 1} & \cdots & a_{N N}
\end{pmatrix},
$$
meaning that we specify (in arbitrary way) the distribution of matrix elements, subject to the constraint of being self-adjoint. The matrix $A$ has $N$ random (real) eigenvalues $\lambda_1\leq \lambda_2\leq \dots\leq \lambda_N$. The study of the distribution of these eigenvalues and their asymptotic properties as $N\to\infty$ is one of the central mathematical topics of the random matrix theory. This is a huge area and there are many surveys and textbooks on the subject, including~\cite{MehtaRMT, ZSRMTbook, PSRMTbook, AndersonGuionnetZeitouniBook, Forrester-LogGas}.

One of the classical choices for~$a_{i j}$ is to set~$A = \frac{1}{2}\left(X + X^*\right)$, where~$X$ is a random matrix whose entries are i.i.d.\ complex Gaussian random variables~$\mathcal{N}(0,1) + \i \mathcal{N}(0, 1)$, with independent real and imaginary parts. In particular, in this case~$a_{i j}$ are all independent, except for the constraint~$a_{i j} = \overline{a_{j i}}$. This ensemble is called the \emph{Gaussian Unitary Ensemble} (GUE), and it has underlying structure which allows for an exact calculation of its eigenvalue distribution.

The algebraic structure underlying the model has allowed for a very precise analysis of many properties of this eigenvalue distribution in the limit as~$N \rightarrow \infty$. One asymptotic result is the distributional convergence of the largest eigenvalue~$\lambda_N$:
\begin{equation}\label{eqn:TWF2}
 \frac{\lambda_N - 2 \sqrt{N}}{N^{-\frac{1}{6}}} \stackrel{d}{\rightarrow} \xi,
\end{equation}
where~$\xi$ is distributed according to the~\emph{Tracy-Widom GUE distribution}, the distribution function of which we will denote by~$F_2$. The function~$F_2$ can be represented as a Fredholm determinant of a certain operator with kernel known as the~\emph{Airy kernel}, or in terms of solutions to a certain Painlev\'e differential equation, as was shown by Tracy and Widom~\cite{tracy1994level}. For a precise definition of~$F_2$, see the Appendix~\ref{app:pt_proc}.

Another important distribution often arising in random matrix theory is the \emph{Tracy-Widom GOE distribution}~$F_1$. This is the limiting distribution function of the appropriately rescaled largest eigenvalue of a matrix sampled from the~\emph{Gaussian Orthogonal ensemble}, and it also has expressions in terms of Fredholm determinants or a Painlev\'e transcendant~\cite{tracy1996orthogonal}. In this case, the random matrix is real symmetric: We let~$A = \frac{1}{2}\left( Y + Y^* \right)$, where~$Y$ has independent~$\mathcal{N}(0,2)$ entries. Using this matrix, the distribution function~$F_1$ is defined by the same limit transition~\eqref{eqn:TWF2}.

The phenomenon known as \emph{random matrix universality} predicts (and in many cases there are rigorous proofs) that the $F_2$ and $F_1$ distributions govern asymptotics of largest eigenvalues for many other ensembles of random complex Hermitian and real symmetric matrices, respectively.

\bigskip

The central message of this article is that these two topics, 2d lattice models in statistical mechanics on one hand and random matrix theory on the other hand, are very related.

\subsection{Examples: Domain wall boundary conditions}
\label{subsec:ex}

\begin{figure}[t]
\begin{center}
\includegraphics[width=0.5\linewidth]{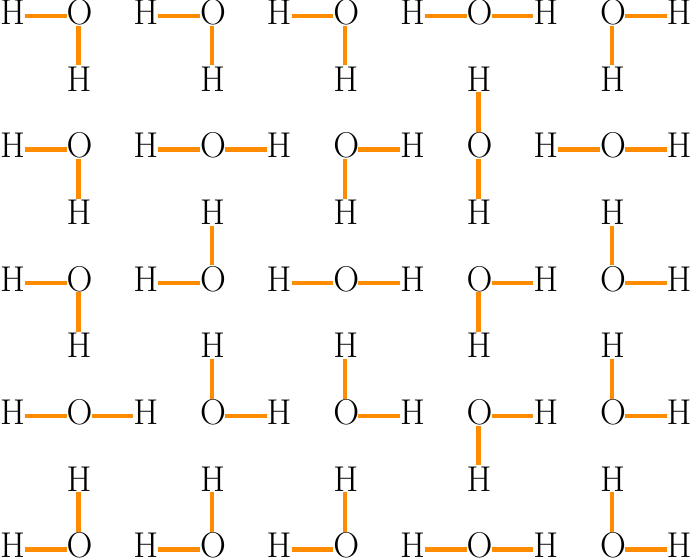}
\caption{Domain wall boundary conditions for the six-vertex model, with~\mbox{$N=5$}.\label{fig:six_v_dwbc}}
\end{center}
\end{figure}

The asymptotic behavior of physical
and probabilistic quantities in
the six-vertex model are very sensitive to boundary
conditions.
One natural choice of boundary conditions (perhaps, the simplest one) is~\emph{domain wall boundary conditions}, which corresponds to choosing an~$N \times N$ square in the lattice; see Figure~\ref{fig:six_v_dwbc}. There are two possible orientations of the square and we choose the one, in which the left and right boundaries of the square consist of $\mathrm{H}$ atoms and the top and bottom boundaries consist of alternating $\mathrm{H}$ and $\mathrm{O}$ atoms. In particular, there are $N^2$ atoms $\mathrm{O}$ and $2 N^2$ atoms $\mathrm{H}$ inside the square.

We leave the proof of the following combinatorial lemma as an exercise for the reader.
\begin{lemma}
\label{Ex_1} For any configuration of the six-vertex model with domain wall boundary conditions, on level $k$ from the bottom, the number~$m$ of horizontal molecules \straight satisfies $1 \leq m \leq k$.
\end{lemma}
As an example, for~$k=1$ there is exactly one horizontal molecule of that form. Indeed, near the bottom left along the boundary, the molecule is one type of corner~\corner, and near the bottom right it is the other type,~\other. The change between the types must happen somewhere, and this gives the position of the unique molecule~\straight in the first line.

\begin{figure}
\centering
\includegraphics[scale=1.2]{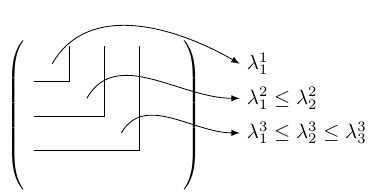}
\caption{An illustration of the eigenvalues~$\{\lambda_i^k\}_{1 \leq i \leq k \leq K}$ of the GUE corners process, with~$K=3$.}
\label{fig:corners}
\end{figure}

We now state several theorems describing the asymptotics of a random configuration of the six-vertex model with  domain wall boundary conditions. In order to do so, we must first define the~\emph{GUE corners process}. Define $X$ as an $N \times N$ matrix of independent $\mathcal{N}(0,1) + \i \mathcal{N}(0,1)$ elements. Then let $M = \frac{1}{2}(X + X^*)$ be a GUE random matrix. Define $\lambda_{i}^k$ to be the~$i^{th}$ eigenvalue of top left $k \times k$ corner of $M$; see Figure~\ref{fig:corners} for an illustration. The joint distribution~$\{\lambda_{i}^k\}_{1 \leq i \leq k \leq K}$ is a marginal of the GUE corners process. The size~$N$ of this matrix does not have to be large; it only needs to be larger than $K$. For varying~$K$, these distributions are consistent, so these finite dimensional marginals define a random collection of random variables $\{\lambda_{i}^k\}_{i \leq k}$, which is called the GUE corners process.

\begin{remark} It is a nice exercise in linear algebra for the reader to show that in the setting of the previous paragraph,~$\lambda_i^{k+1} \leq \lambda_i^k \leq \lambda_{i+1}^{k+1}$ almost surely. Hint: Use the variational characterization of the eigenvalues of a Hermitian matrix, also known as the Courant-Fischer-Weyl min-max principle.
\end{remark}

\begin{theorem}[Johansson-Nordenstam~\cite{johansson2006eigenvalues}] \label{thm:JNeig}
Suppose $a_1 = a_2 = b_1 = b_2 =1$, $c_1 c_2 = 2$. As $N \rightarrow \infty$, there are exactly $k$ horizontal molecules~\straight in row $k$ with probability tending to $1$. Let $x^k_1 < x^k_2 < \cdots < x^k_k$ be their coordinates. Then in the sense of convergence of finite-dimensional distributions, we have
$$\lim_{N \rightarrow \infty} \left\{ \frac{x_i^k - \gamma_1 N}{\gamma_2 \sqrt{N}} \right\}_{1 \leq i \leq k} \stackrel{d}{=} \{ \lambda_{i}^k \}_{1 \leq i \leq k},$$
with $\gamma_1=\gamma_2=\frac{1}{2}$ and $\{ \lambda_{i}^k \} \stackrel{d}{=} \text{GUE corners process}$.
\end{theorem}
The choice of weights in the theorem above is an instance of the six-vertex model at the~\emph{free fermionic point}, which means that~$\Delta \defeq \frac{a_1 a_2 + b_1 b_2 - c_1 c_2}{2 \sqrt{a_1 a_2 b_1 b_2}}$, an important parameter of the weights, satisfies~$\Delta = 0$ (see Baxter's book~\cite[Chapter 8]{baxter2007exactly} for a discussion of the parameter~$\Delta$). At the free fermionic point, the model has the structure of a~\emph{determinantal point process}; see Appendix~\ref{app:pt_proc} for a brief definition.

Rigorous results about the model away from the free fermion point are significantly more difficult to obtain. The next theorem is one of the few examples of such a result.\footnote{For a joint extension of Theorems \ref{thm:JNeig} and \ref{thm:GP}, see upcoming \cite{Gorin_Liechty_2023}.}
\begin{theorem}[Gorin~\cite{gorin2014alternating}, Gorin-Panova~\cite{GorinPanova2012_full}] \label{thm:GP}
The conclusion of Theorem~\ref{thm:JNeig} remains true for $a_1 = a_2 = b_1 = b_2 = c_1 = c_2 = 1$ (the uniform measure). As $N \rightarrow \infty$, there are exactly $k$ horizontal molecules~$\straight$ in row $k$,~$x_1^k < \cdots < x_k^k$, with probability tending to $1$. Furthermore,
$$\lim_{N \rightarrow \infty}  \left \{ \frac{x_i^k - \gamma_1 N}{\gamma_2 \sqrt{N}}  \right \}_{1 \leq i \leq k} \stackrel{d}{=} \{ \lambda_{i}^k \}_{1 \leq i \leq k},$$
with $\gamma_1= \frac{1}{2}$, $\gamma_2=\sqrt{\frac{3}{8}}$ and $\{ \lambda_{i}^k \} \stackrel{d}{=} \text{GUE corners process}$.
\end{theorem}

\begin{figure}[t]
\begin{center}
\includegraphics[width=0.5\linewidth]{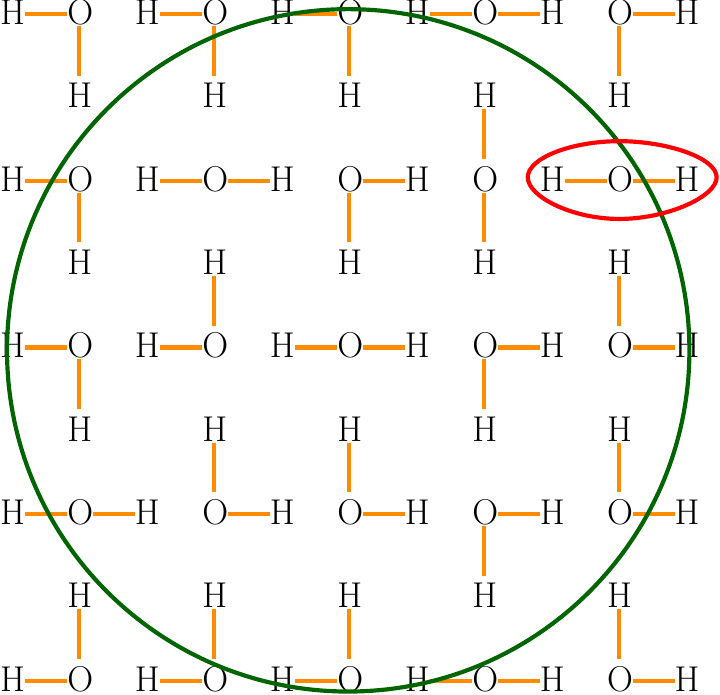}
\caption{The setting of Theorem~\ref{thm:circle}.\label{fig:circle}}
\end{center}
\end{figure}

Other related theorems about asymptotic behavior of random six-vertex configurations with domain wall boundary conditions include the following:

\begin{theorem}[Johansson~\cite{Johansson2005arctic}]\label{thm:circle}
Suppose~$a_1 = a_2 = \cdots = b_2 =1$, $c_1 c_2 = 2$.
\begin{enumerate}
\item As $ N  \rightarrow \infty$, positions of~\straight and~\up vertices stay inside the circle inscribed in the square (cf.\ Figure \ref{fig:circle}) with overwhelming probability\footnote{In more details, for each $\epsilon>0$, the probability of the event ``all vertices of these two types are inside the $(\epsilon N)$--neighborhood of the inscribed circle'' tends to $1$ as $N\to\infty$.}.
\item Let~$\frac{1}{2} < \alpha < 1$, and~$k = \lfloor  \alpha N  \rfloor$. If~$x_{\text{max}}^k$ denotes the horizontal coordinate of the rightmost vertex on the line~$y = k$ which is not of type~$b_2$,~\upright, then we have
\begin{align*}
\lim_{N\to\infty}\frac{x_{\text{max}}^k - \gamma_3(\alpha) N}{\gamma_4(\alpha) N^{\frac{1}{3}}} \stackrel{d}{=} \xi_{\text{GUE}},
\end{align*}
\end{enumerate}
where~$\xi_{\text{GUE}}$ is distributed according to~$F_2$, the Tracy-Widom GUE distribution. In the expression above,~
$\gamma_3(\alpha) =  \frac{1}{2}  +  \sqrt{\alpha - \alpha^2}$ and~
$\gamma_4(\alpha) = \frac{1}{2}\left(  \frac{(1-2 \alpha)^2}{\sqrt{(1-\alpha) \alpha}} \right)^{\frac{1}{3}}$.

\end{theorem}

The previous theorem implies that for large~$N$, outside of the inscribed circle and in a neighborhood of the top right corner of the square, one only sees type~$b_2$ molecules with high probability. By various symmetries of the model, in the bottom right, bottom left, and top left corners, the molecules of types~$a_2$,~$b_1$,~$a_1$, respectively, are dominant. Furthermore, in these other three corners, the border of the corresponding cluster of molecules along a fixed horizontal slice will have a similar asymptotic behavior, but with different constants~$\gamma_3$ and~$\gamma_4$.

\begin{theorem}[Johansson~\cite{Johansson2005arctic}]\label{thm:Ltriangle}
Suppose $a_1 = \cdots = b_2 = 1$, $c_1 c_2 = 2$. Cut off a $L \times L$ triangle from the top right of the square, as in Figure \ref{fig:Ltriangle}. Define $\xi_N$ as the largest $L$ such that there are only type~$b_2$ molecules,~\upright, inside the triangle. Then there are explicit constants~$\gamma_5, \gamma_6$, such that~$\xi_N$ satisfies
$$\lim_{N\to\infty}\frac{\xi_N  - \gamma_5 N}{\gamma_6 N^{\frac{1}{3}}} \stackrel{d}{=} \xi_{\text{GOE}}, $$
where~$\xi_{\text{GOE}}$ is distributed according to~$F_1$, the Tracy-Widom GOE distribution.
\end{theorem}

\begin{figure}[t]
\centering
\includegraphics[width=0.5\linewidth]{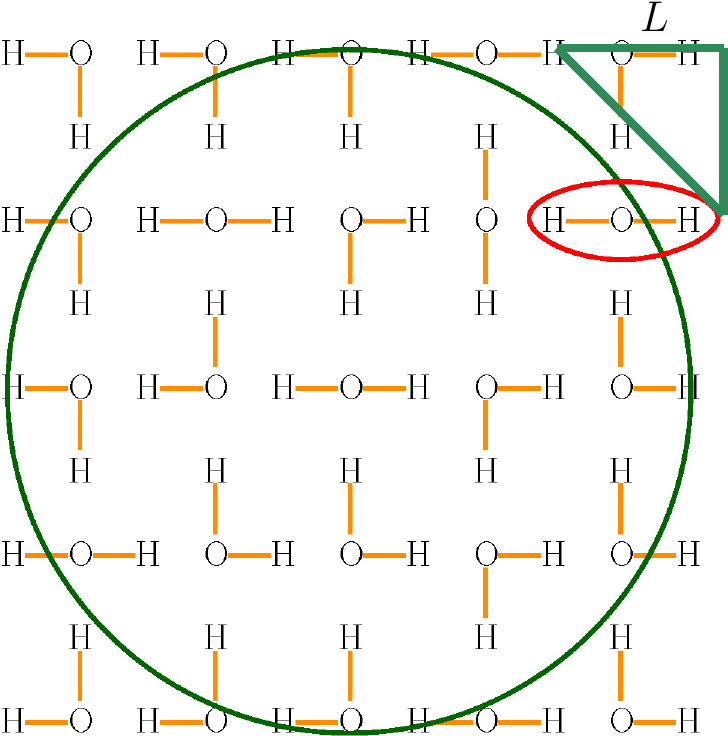}
\caption{The setting of Theorem~\ref{thm:Ltriangle}.}
\label{fig:Ltriangle}
\end{figure}

In both Theorems~\ref{thm:circle} and \ref{thm:Ltriangle}, Johansson actually proves the result for random domino tilings of the ``Aztec diamond'' domain. However, there is a local, weight preserving, many-to-one mapping of domino tilings of the Aztec diamond to six-vertex configurations with domain wall boundary conditions if~$\Delta = 0$ as in the two theorems above. In particular, the mapping is bijective up to certain local modifications of dominos, which in particular may only change the positions of observables being analyzed by at most~$O(1)$. Thus, the results translate immediately to the~$\Delta = 0$ six-vertex model setting. See~\cite{ferrari2006domino} for more details on this correspondence. We refer to Figures~\ref{fig:circle} and~\ref{fig:Ltriangle} for illustrations of the setting of Theorem~\ref{thm:circle} and~\ref{thm:Ltriangle}, respectively.

A very recent development in the study of the six-vertex model with domain wall boundary conditions, away from~$\Delta = 0$, is the following result.

\begin{theorem}[Ayyer-Chhita-Johansson~\cite{ayyer2021goe}]\label{thm:GOE_unif}
Suppose that $a_1 = \cdots = b_2 = c_1 = c_2 = 1$ (the uniform measure). Then the same limit in distribution as in Theorem~\ref{thm:Ltriangle} holds, but with $\gamma_5, \gamma_6$ replaced by new constants.
\end{theorem}

The following conjecture extends all of these results.
\begin{conjecture}\label{conj:GUE}
Theorems \ref{thm:JNeig}--\ref{thm:GOE_unif} extend (perhaps with different values of constants) to all $a_1,a_2,b_1,b_2,c_1,c_2$ satisfying $\Delta  < 1$, and to much more general boundary conditions.
 \end{conjecture}
One way to make Conjecture~\ref{conj:GUE}
more precise is to consider the six-vertex model
in a large polygon with axis-parallel sides
(rather than in a square).
Then we expect the appearance of a curve,
similar to the circle of
Figures~\ref{fig:circle} and~\ref{fig:Ltriangle},
such that asymptotically all~$c_1$ and~$c_2$
 molecules are inside of the curve.
 (See~\cite{colomo2010arctic, colomo2010arctic2, Colomo_S_2016, aggarwal2020arctic}
 for examples of such curves).
  On various parts of the curve, the fluctuations of
  extremal molecules of the appropriate types
  should be described by versions of
  Theorems~\ref{thm:JNeig}--\ref{thm:GOE_unif}. For example,
  heuristic calculations and numerical investigations of \cite{praehofer2023domain}
  strongly support the conjecture that with
  uniform weights ($a_1 = a_2 = b_1 = b_2 = c_1 = c_2 = 1$)
  and domain wall boundary conditions,
  the scale of fluctuations of the extremal molecules
  (away from tangency points)
  is $N^{\frac{1}{3}}$, and that these
  fluctuations follow the Tracy-Widom GUE distribution $F_2$; i.e.
  an analog of Theorem~\ref{thm:circle} should hold in this setup.

\bigskip

What if $\Delta > 1$?  This will be discussed as well, later in the notes. The random matrix distributions can appear in another form in this situation, as we outline later in Section \ref{Section_height_function_assy}. As a teaser, we refer to Figure \ref{Figure_teaser} in that section for simulations and pointers to the asymptotic objects
appearing in the six-vertex model with various parameters and boundary conditions.

\subsection{Plan of the next sections}

In the rest of this note, we provide proofs of some of the results connecting the six-vertex model to random matrix distributions.

First, in Sections \ref{Section_IK} and \ref{Section_GUE} we present an argument for a slight generalization of Theorem \ref{thm:JNeig}. Our approach is quite different from the original proof of Johansson and Nordenstam: we rely on the \emph{Izergin-Korepin determinant} --- a celebrated evaluation of the inhomogeneous partition function for the six-vertex model with domain wall boundary conditions. Section \ref{Section_IK} proves this determinantal formula, introducing on our way a central identity in the study of the six-vertex model --- the Yang--Baxter equation. Section \ref{Section_GUE} paves a way from the Izergin-Korepin determinant to the GUE--corners process.

Next, in Sections \ref{Section_Stochastic_6v} and \ref{Section_Schur} we focus on another asymptotic theorem, presenting and proving the result of \cite{BCG6V}, which connects the fluctuations of the height function of the stochastic six-vertex model (corresponding to $\Delta>1$ case) to the Tracy-Widom GUE distribution $F_2$. Again, we do not follow the arguments of \cite{BCG6V}, but give an alternative proof, which is based on a recently discovered generalization of the Izergin-Korepin determinant. Section \ref{Section_Stochastic_6v} proves this generalization, then states the asymptotic theorem of \cite{BCG6V} and discusses the relevant background. Section \ref{Section_Schur} proves this asymptotic theorem by utilizing the determinant, connecting it to the \emph{Schur measures}, and analyzing the latter using difference operators and contour integrals (this particular approach to the asymptotics of the Schur measures is different from the standard approach in the literature). Finally, Appendix \ref{app:pt_proc} collects auxiliary definitions and facts about the Tracy--Widom distribution and its multidimensional generalization known as the Airy point process.

Our ultimate goal is three-fold: we give self-contained proofs of two theorems connecting the six-vertex model to random matrices; we present new proofs of these theorems, which arguably are as short as one can hope for; on our path we highlight various key tools used in the study of the six-vertex model.

\section{Izergin-Korepin Determinant}
\label{Section_IK}
\subsection{Vertex weights and Gibbs property}

Starting from this section, we use an equivalent representation shown in Figure~\ref{fig:corresp} for the six local configurations of the six-vertex model. Thus, if we fix a sub-graph of~$\mathbb{Z}^2$, a \emph{state} of the six-vertex model is a collection of up-right paths in (the sub-graph of)~$\mathbb{Z}^2$ which can meet at a corner but never cross, see Figures \ref{fig:gibbs_height} and \ref{fig:dwbc}.

If we have a finite sub-graph of~$\mathbb{Z}^2$, then \emph{boundary conditions} are a choice of entrance and exit locations for paths oriented to travel north and east. In general, when specifying the particular model we are studying, we will specify:  the sub-graph of~$\mathbb{Z}^2$ on which our states live, the~$6$ weights~$a_1,a_2,b_1,b_2,c_1,c_2$, and the boundary conditions.  Given these specifications, we study the Boltzmann measure on states~$S$ satisfying the boundary conditions and defined by~\eqref{eqn:bm}.

\begin{figure}[t]
\centering
\includegraphics[scale=1.3]{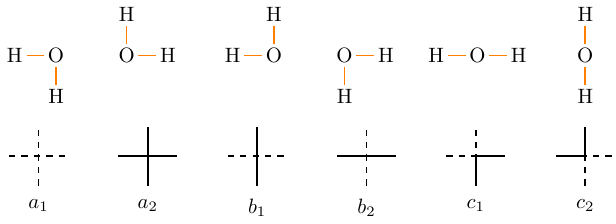}
\caption{The correspondence between the two representations of the six local configurations.\label{fig:corresp}}
\end{figure}

Each Boltzmann measure satisfies the~\emph{Gibbs property}; given a fixed boundary condition for the paths on a sub-graph of the given domain, the conditional distribution of the state~$S$ inside the sub-graph given these boundary conditions is itself given by the Boltzmann measure on this sub-graph.

\begin{figure}[t]
\includegraphics[width=0.9\linewidth]{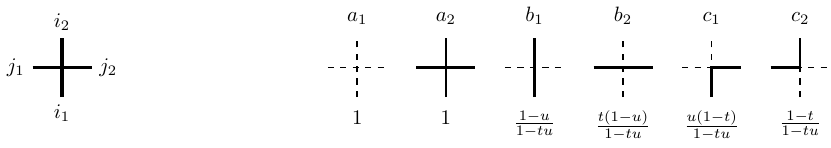}
\caption{The six local configurations, with general weights shown above the vertices and weights parameterized by~$u, t$ shown below. If the spectral parameter is~$u$, the weight of the vertex shown on the left above is denoted by $w_u(i_1, j_1; i_2, j_2)$.}
\label{fig:vertices}
\end{figure}

We claim that the entire family of probability measures on configurations in a fixed domain, using certain \emph{gauge transformations}, which are operations on the weights preserving the Boltzmann measure, can be parameterized by two parameters $u, t$ (rather than the~$6$ parameters~$a_1,a_2,b_1,b_2,c_1,c_2$) as follows. We denote the presence of a path on a vertical or horizontal edge by a $\{0, 1\}$--valued indicator $i$ or $j$, and denote the corresponding vertex weight by $w_u(i_1,j_1;i_2, j_2)$; see Figure \ref{fig:vertices}. We define the weights by

\begin{align}\label{eqn:weights}
		&w_u(0,0;0,0)= a_1 = 1,&
		&w_u(1,1;1,1)= a_2 = 1,&
        \\         \notag
		&w_u(1,0;1,0)= b_1 =  \frac{1 - u}{1 - t u},&
		&w_u(0,1;0,1)  = b_2 =  \frac{t (1 - u) }{1 - t u},&
		\\
		&w_u(1,0;0,1) = c_1 = \frac{u(1 -t)}{1 - t u},&
		&w_u(0,1;1,0)  = c_2 =  \frac{1 - t}{1 - t u}.& \notag
\end{align}

Above $t$ is the \emph{quantum parameter}, which is fixed, and $u$ is the \emph{spectral parameter}, which we will ultimately allow to vary from vertex to vertex, although thus far it has been fixed.

There are other common parameterizations of vertex weights (see, e.g., \cite{baxter2007exactly}), and by changing variables, each of these can be obtained from our parameterization above. The parameterization above has the benefit of giving unified formulas for vertex weights in all regimes (see the discussion after Lemma~\ref{lem:Gauge} for more on the different regimes, which correspond to different possible values of~$\Delta$). Furthermore, the weights~\eqref{eqn:weights} possess three special features which we will exploit:
\begin{itemize}
\item After multiplying by~$1- t u$ they are polynomials in~$t $ and~$u$.
\item The stochasticity condition holds:~$b_1 + c_1 = b_2 + c_2 = 1$.
\item An algebraic relation called the \emph{Yang--Baxter equation} holds, as detailed in Theorem~\ref{thm:YBE}.
\end{itemize}

\begin{figure}[t]
\centering
\includegraphics[width=0.35\linewidth]{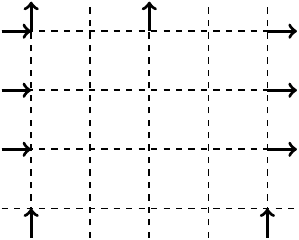} \qquad \qquad \qquad
\includegraphics[width=0.35\linewidth]{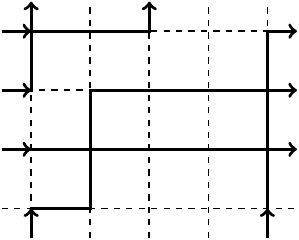}
\caption{An example of boundary conditions for the six-vertex model in a~$4 \times 5$ rectangle. Entrance and exit locations of paths are marked by arrows. On the right is one possible configuration with these boundary conditions.}
\label{fig:gibbs_height}
\end{figure}

The following lemma explains that up to equivalence under gauge transformations, there is only a two parameter family of weight functions. Therefore, there is no loss in generality in working only with weights \eqref{eqn:weights}.

\begin{lemma}\label{lem:Gauge}
If we fix a sub-graph of~$\mathbb{Z}^2$, boundary conditions for the six-vertex model, and six weights~$a_1, a_2, b_1, b_2, c_1, c_2$ which are the same for all vertices~$(i, j)$, then the Boltzmann measure
\begin{equation}\label{eqn:bmeas}
\mu(S) = \frac{1}{Z} \prod_{i, j} \text{wt}(i, j; S)
\end{equation}
only depends on the two parameters $\frac{a_1 a_2}{b_1 b_2}$, $\frac{a_1 a_2}{c_1 c_2}$.
\end{lemma}
\begin{proof}
Given a particular configuration $S$, we call the number of $a_1$ vertices $N_1$, the number of $a_2$ vertices $N_2$, the number of $b_1$ vertices $N_3$, the number of~$b_2$ vertices~$N_4$, the number of~$c_1$ vertices~$N_5$, and the number of~$c_2$ vertices~$N_6$. Observe that if we have a linear combination of the~$N_i$ whose value is independent of the configuration, we obtain a one parameter family of transformations of the weights that preserve the Boltzmann measure. We have the following conserved quantities and corresponding gauge transformations:
\begin{enumerate}[(a)]
\item Since the number of vertices in the domain is constant,~$\sum_{i=1}^6 N_i = \text{Const}$, so multiplying all six weighs by a constant preserves the measure~\eqref{eqn:bmeas}.

\item Since all paths have fixed entry and exit points, the number of times a particular path turns up minus the number of times it turns right is prescribed by the boundary conditions to be one of the three numbers~$\{-1, 0, 1\}$. This ultimately implies that~$N_5 - N_6 = \text{Const}$. So $(c_1, c_2) \rightarrow (\alpha c_1, \frac{1}{\alpha} c_2)$ preserves the measure~\eqref{eqn:bmeas}.

\item The total number of edges occupied by a path is constant, which implies that
$$2 N_2 + N_3 + N_4 + N_5 + N_6 = \text{Const}.$$
This together with (a) above implies that $N_1 - N_2$ is constant. So $(a_1, a_2) \rightarrow (\alpha a_1, \frac{1}{\alpha} a_2)$ preserves the measure~\eqref{eqn:bmeas}.

\item The total number of vertical edges occupied by paths is fixed;~$N_2 + N_3 + \frac{1}{2}N_5 + \frac{1}{2} N_6 = \text{Const}$. Subtracting $\frac{1}{2}$ of (c) from this, we get that~$N_3 - N_4 = \text{Const}$. So $(b_1, b_2) \rightarrow (\alpha b_1, \frac{1}{\alpha} b_2)$ preserves the measure~\eqref{eqn:bmeas}.

\end{enumerate}
The claim follows from (a)--(d): using these four transformation we can get any six-tuple $a_1,a_2,b_1,b_2,c_1,c_2$ from any other six-tuple with the same values of $\frac{a_1 a_2}{b_1 b_2}$ and $\frac{a_1 a_2}{c_1 c_2}$.
\end{proof}

Lemma~\ref{lem:Gauge} implies that any positive weights~$a_1,a_2, b_1, b_2, c_1, c_2 > 0$ can be brought to the form~\eqref{eqn:weights} via transformations preserving the measure~\eqref{eqn:bmeas}. Note that it may be necessary to take~$u$ and~$t$ to be complex numbers for that. In fact, the formulas~\eqref{eqn:weights} give a positive measure~\eqref{eqn:bmeas} in two situations:
\begin{enumerate}[1)]
\item If both~$t$ and~$u$ are positive reals.
\item If both~$t$ and~$u$ are complex numbers on the unit circle: $u,t\in\mathbb C$ and $|t| = |u| = 1$.
\end{enumerate}

Indeed, in each of these cases we have~$b_1 b_2 > 0$ and~$ c_1 c_2 > 0$, which guarantees positivity of the measure. Plugging~\eqref{eqn:weights} into~$ \Delta =\frac{a_1 a_2 + b_1 b_2 - c_1 c_2}{2 \sqrt{a_1 a_2 b_1 b_2}}$, one computes~$\Delta = \pm \frac{1}{2}(t^{\frac{1}{2}} + t^{-\frac{1}{2}})$, where the sign~$\pm$ arises from the square root in the denominator, and the correct choice is determined by requiring that~$\sqrt{a_1 a_2 b_1 b_2} > 0$. Hence, the first case corresponds to~$|\Delta| \geq 1$ and the second one to~$|\Delta| \leq 1$. In particular,~$\Delta = 0$ when~$t = -1$.

\begin{figure}[t]
\centering
\includegraphics[width=0.35\linewidth]{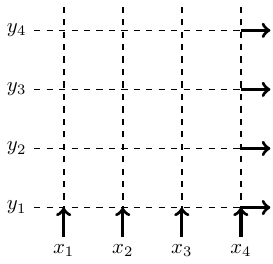}
\qquad \qquad \qquad
\includegraphics[width=0.35\linewidth]{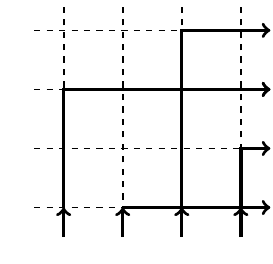}
\caption{An illustration of domain wall boundary conditions for~$N = 4$, and an example of a configuration.}\label{fig:dwbc}
\end{figure}

\subsection{Izergin-Korepin formula and Yang--Baxter equation}
\label{sec:dwbc}

The $N \times N$ square~$[1, N] \times [1, N] \cap \mathbb{Z}^2$ with paths entering at every site of the bottom boundary and exiting at every site of the right boundary is called the \emph{domain wall boundary conditions} (DWBC), see Figure~\ref{fig:dwbc} for an example.

A special feature of DWBC is that in this situation the partition function, which is the normalization constant~$Z$ in~\eqref{eqn:bmeas}, admits an explicit formula. In fact (crucially) the computation also works for inhomogeneous weights obtained by allowing~$u$ in~\eqref{eqn:weights} to depend on the vertex.

We attach spectral parameter $x_i$ to column $i$ and $y_j$ to row $j$, $1\le i,j\le N$, which means that the weight $\text{wt}(i, j; S)$ of a vertex $(i, j)$ uses the weight function $w_{x_i y_j}$, i.e. we use~\eqref{eqn:weights} with~$u = x_i y_j$. We denote by $Z_N(y_1,\dots, y_N; x_1,\dots, x_N;t)$ the inhomogeneous partition function for this model:
$$Z_N(y_1,\dots, y_N; x_1,\dots, x_N;t) \defeq \sum_{S} \prod_{i=1}^N \prod_{j=1}^N  \text{wt}(i, j; S) .$$

\begin{theorem}[Izergin-Korepin Determinant \cite{KorepinNorms},~\cite{IKDet}]\label{thm:IK}
For domain wall boundary conditions, the partition function~$Z_N(x_1,\dots, x_N; y_1,\dots, y_N; t)$ is given by the formula
\begin{equation}\label{eqn:IZK}
Z_N(y_1,\dots, y_N; x_1,\dots, x_N;t)  =  \frac{\prod_{i, j=1}^N (1 - x_i y_j)}{\prod_{i < j} (x_i - x_j)(y_i - y_j)} \det \left[ \frac{(1-t) x_i y_j}{(1-x_i y_j)(1- t x_i y_j)} \right]_{i,j=1}^N .
\end{equation}
\end{theorem}

\begin{proof}
We proceed by induction on $N$.
\begin{enumerate}[1.]
\item \textbf{Base case:} For $N = 1$, we directly check the equality, which amounts to observing that for $x = x_1, y = y_1$, we have
$$c_1 =  \frac{(1-t) x y}{1 - t x y} = \frac{(1-x y)(1-t)x y}{ (1- x y) (1- t x y)} = Z_1(x; y; t).$$

\item \textbf{Inductive step:} We denote $\boldsymbol{x} = (x_1,\dots, x_N),  \boldsymbol{y} = (y_1, \dots, y_N)$. First we argue that it suffices to show that the LHS and RHS in~\eqref{eqn:IZK}, $Z_N(y_1,\dots, y_N; x_1,\dots, x_N;t)$ and $R_N(y_1,\dots, y_N; x_1,\dots, x_N;t)$, are rational functions satisfying the following properties:
\begin{enumerate}
\item[(a)] Symmetry in variables~$(x_1, \dots, x_N)$, and in variables~$(y_1,\dots, y_N)$.
\item[(b)] The function $Z_N(\boldsymbol{x} ; \boldsymbol{y}  ; t) \cdot \prod_{i, j} (1 - t x_i y_j)$ (resp. $R_N(\boldsymbol{x} ; \boldsymbol{y}  ; t) \cdot \prod_{i, j} (1 - t x_i y_j)$) is a polynomial in $x_1,\dots,x_N, y_1,\dots,y_N$, such that the degree of each variable is at most $N$.
\item[(c)] If $x_N = \frac{1}{y_N}$, then
\begin{align*}
Z_N(\boldsymbol{x} ; \boldsymbol{y} ; t) &= Z_{N-1}(x_1,\dots, x_{N-1}; y_1, \dots, y_{N-1}; t)   \qquad \text{ and} \\
R_N(\boldsymbol{x} ; \boldsymbol{y} ; t) &= R_{N-1}(x_1,\dots, x_{N-1}; y_1, \dots, y_{N-1}; t) .
\end{align*}
\item[(d)] If $x_i = 0$ for any $i$, then~$Z_N(\boldsymbol{x} ; \boldsymbol{y} ; t) = 0$ and~$R_N(\boldsymbol{x} ; \boldsymbol{y} ; t) = 0$.
\end{enumerate}

\smallskip

Suppose we can show that these four properties are satisfied by the LHS and RHS. Then if $P_N$ represents either the LHS or RHS, we have that \mbox{$x_N \mapsto \prod_{i,j}(1-t x_i y_j) P_N(x_1,\dots,x_N; y_1,\dots,y_N;t)$} is a degree $N$ polynomial in $x_N$, and at the $N+1$ values of $x_N$:
$$0,\quad 1/y_1,\quad 1/y_2,\quad \dots,\quad 1/y_N,$$
its values are determined by properties (c) and (d) and by the induction assumption. A degree~$N$ polynomial is uniquely determined by its values at~$N+1$ points, therefore,~$Z_N = R_N$.

Next, we prove that the functions on the LHS and RHS satisfy the four properties.

\smallskip

\textbf{Proof that the RHS of~\eqref{eqn:IZK} satisfies (a)-(d)}:
\begin{enumerate}
\item[(a)] We argue for symmetry in $(x_1,\dots, x_N)$, as the symmetry in $y$'s is similar. The product $\prod_{i, j}(1 - x_i y_j)$ is clearly symmetric. The function $\prod_{i < j} (x_i - x_j) $ is anti-symmetric in $(x_1,\dots, x_N)$. In the determinant term, the swap $x_i \leftrightarrow x_j$ swaps rows $i$ and $j$ in the matrix, so the determinant will change sign, and thus the determinant is also anti-symmetric. Thus the RHS is symmetric in the $x_i$.
\item[(b)] The function
$$\left( \prod_{i, j=1}^N (1 - x_i y_j) \right) \left(\prod_{i, j=1}^N (1 - t x_i y_j)  \right)  \det \left[ \frac{(1-t) x_i y_j}{(1-x_i y_j)(1- t x_i y_j)} \right]_{i,j=1}^N$$
is a polynomial which is anti-symmetric in $(x_1,\dots, x_N)$, and in $(y_1,\dots, y_N)$. Thus, it is divisible by~$\prod_{i < j} (x_i - x_j) (y_i-  y_j)$. It is easy to see that the degree of
$$\frac{\left( \prod_{i, j=1}^N (1 - x_i y_j) \right) \left(\prod_{i, j=1}^N (1 - t x_i y_j)  \right) }{ \prod_{i < j} (x_i - x_j) (y_i-  y_j)}  \det \left[ \frac{(1-t) x_i y_j}{(1-x_i y_j)(1- t x_i y_j)} \right]_{i,j=1}^N$$
in any variable is equal to $N$: Say, if we send~$x_1 \rightarrow \infty$, it grows proportionally to~$x_1^N$.

\item[(c)] Before sending~$x_N \rightarrow \frac{1}{y_N}$, we multiply the bottom row of the determinant in the RHS of~\eqref{eqn:IZK} by $(1-x_N y_N)$. Then when we set~$x_N = \frac{1}{y_N}$, the bottom row has a single nonzero entry, equal to~$1$, in the bottom right corner. Thus, the determinant reduces to the one corresponding to~$R_{N-1}(x_1,\dots, x_{N-1}; y_1,\dots, y_{N-1}; t)$. The extra pre-factors containing~$x_N$ and~$y_N$ in the numerator and denominator exactly cancel, so the entire expression evaluates to exactly~$R_{N-1}(x_1,\dots, x_{N-1}; y_1,\dots, y_{N-1}; t)$.

\item[(d)] This follows from the fact that a row of the determinant becomes $0$ if $x_i = 0$.

\end{enumerate}

\medskip

\textbf{Proof that the LHS of~\eqref{eqn:IZK} satisfies (a)-(d)}:
\begin{enumerate}
\item[(a)] Symmetry of the partition function $Z_N$ follows from the \emph{Yang--Baxter Equation}. See Theorem \ref{thm:YBE} and Lemma \ref{lem:symmetry} below.

\item[(b)] The function $Z_N(\boldsymbol{x}; \boldsymbol{y} ; t) \cdot \prod_{i, j} (1 - t x_i y_j)$ is a polynomial because the factor ${\prod_{i, j} (1 - t x_i y_j)}$ clears the denominators of vertex weights. The degree in $x_i$ is $N$ because ${(1 - t x_i y_j) \text{wt}(i, j; S)}$ has degree $1$ in $x_i$ for any $(i, j)$ and any configuration $S$.

\item[(c)] This follows from observing (see the vertex weights~\eqref{eqn:weights} and Figure~\ref{fig:vertices}) that if $x_N = 1/y_N$, then the vertex at $(N, N)$ must have configuration $c_1$. This implies that each vertex at~$(i, N)$ for $i < N$ has type~$a_1$, and each vertex $(N, j)$ for $j < N$ has type~$a_2$. Thus, each configuration with nonzero weight is determined by the configuration of paths in the square with top right corner $(N-1, N-1)$, and furthermore the weight of the configuration is equal to the weight of the vertices in this smaller square.

\item[(d)] By symmetry it suffices to prove that~$Z_N = 0$ if~$x_N= 0$. This can be deduced by observing that at least one vertex in the right-most column must have type~$c_1$, and this weight vanishes in the right-most column if $x_N = 0$. \qedhere

\end{enumerate}

\end{enumerate}

\end{proof}

\begin{remark}
The proof above is a relatively straightforward verification. But how would one \emph{derive} such a formula for the partition function? The symmetry relations and recursive formulas for the partition function at special values of~$(x, y)$ provide a set of conditions which must be satisfied by any expression for the partition function. However, it is not obvious how to arrive at the exact formulation of the solution in terms of a determinant. Historically, the purpose of Korepin's original paper~\cite{KorepinNorms} was to compute norms of Bethe Ansatz eigenfunctions, and a similar recurrence is introduced there. The recurrence was solved in terms of a determinant in~\cite{IKDet}. It seems that \cite{IKDet} was an educated guess inspired by frequent appearances of the determinants in theoretical physics literature of that era.
\end{remark}

Now we discuss
the~\emph{Yang--Baxter Equation}
 entering into the proof of the
 symmetry of ${Z_N(x_1,\dots, x_N; y_1, \dots, y_N; t)}$.
 The Yang--Baxter Equation is responsible
  for the \emph{quantum integrability} of
  the six-vertex model. The Yang--Baxter equation first appeared in works of
  McGuire \cite{mcguire1964study}, Brezin--Zinn-Justin \cite{Brezin_Zinn_Justin_1966}, and Yang
  \cite{yang1967some}  in their
  study of a certain quantum mechanical many body problem, and ultimately, Baxter utilized similar relations in his study of the eight vertex model in 1972~ \cite{baxter2007exactly}. A further study of the algebraic structure underlying the methods of the aforementioned works led to the introduction and study of quantum groups. See the surveys~\cite{MR1338602, reshetikhin2010lectures,Lamers_2015} and references therein for a more detailed account of quantum integrability and the Yang--Baxter equation.

The Yang--Baxter Equation is a collection of equalities of partition functions for domains containing three vertices, one for each boundary condition, which can be stated graphically as shown in Figure~\ref{fig:ybe}.

\begin{figure}[t]
\includegraphics[scale=1]{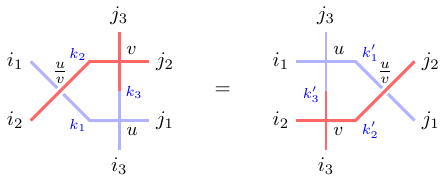}
\caption{Graphical representation of the Yang--Baxter equation.} \label{fig:ybe}
\end{figure}

In Figure~\ref{fig:ybe}, the lattice on each side represents a partition function, with weight functions~$w_u, w_v$ at the vertices labeled in the figure with~$u$,~$v$, respectively. The weights at the ``cross vertex'' are defined using~\eqref{eqn:weights} by
$$X_{u,v}
	\Biggl(
		\scalebox{.5}{\begin{tikzpicture}
		[scale=1.5, baseline=16.5pt]
		\draw[line width=3] (0,0)--++(1,1);
		\draw[line width=3] (0,1)--++(1,-1);
		\node[left] at (0,0) {\LARGE$i_1$};
		\node[left] at (0,1) {\LARGE$j_1$};
		\node[right] at (1,0) {\LARGE$j_2$};
		\node[right] at (1,1) {\LARGE$i_2$};
		\end{tikzpicture}}
	\Biggr) = w_{u/v} (i_1,j_1; i_2,j_2).$$
The labels~$i_l, j_l$, where $l \in \{1,2,3\}$, are an arbitrary fixed boundary condition, and $k_l, k_l'$, where $l \in \{1,2,3\}$, are the indicator variables of internal edges, which we sum over to get the partition function. Note that the two spectral parameters $u$ and $v$ are swapped when going from the left hand side to the right hand side. Using this notation, the Yang--Baxter Equation can be stated as follows.

\begin{figure}[t]
\centering
\includegraphics[width=0.8\linewidth]{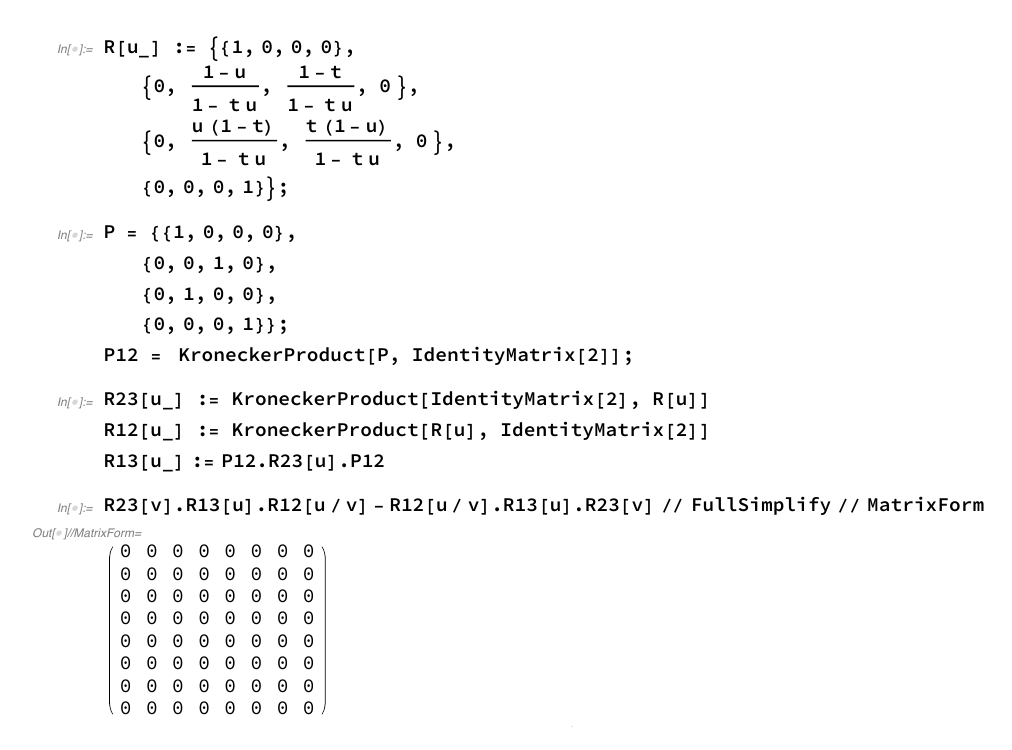}
\caption{Code to check the Yang--Baxter equation in Mathematica.}
\label{fig:YBEcode}
\end{figure}

\begin{theorem}[Yang--Baxter Equation]\label{thm:YBE}
For each $i_1,i_2,i_3, j_1,j_2,j_3 \in \{0,1\}$,
\begin{align}
\notag \sum_{k_1,k_2,k_3} & w_{u/v} (i_2,i_1; k_2,k_1) w_u(i_3, k_1; k_3, j_1) w_v(k_3, k_2; j_3, j_2) \\
&= \sum_{k_1',k_2',k_3'}  w_v(i_3, i_2; k_3', k'_2) w_u(k_3', i_1; j_3, k_1') w_{u/v} (k_2',k_1'; j_2,j_1) .\label{eqn:YBE}
\end{align}
\end{theorem}
\begin{proof}
This is a direct check for all $64 = 2^6$ cases. As an example, we may explicitly check the following equation for~$(i_1,i_2, i_3, j_1, j_2, j_3) = (0, 1,0, 1,0,0)$:

\begin{center}
\begin{tikzpicture}[baseline=(current bounding box.center)]
\draw (0.75,2)--(1.4,2);
\draw (-.4,2)--(0,1.5)--(0.35,2)--(.75,2);
\draw[very thick] (-.4,1)--(0,1.5)--(0.35,1)--(1.5,1);
\draw (1,0.6)--(1,1)--(1.4,1);
\draw (1,0.6)--(1,2.4);
\end{tikzpicture}
$=$
\begin{tikzpicture}[baseline=(current bounding box.center)]
\draw (0,2)--(1,2)--(1.4,1.5)--(1.8,2);
\draw[very thick] (0,1)--(1,1)--(1.4,1.5)--(1.8,1);
\draw (1,1)--(1.4,1.5)--(1.8,1);
\draw (.4,0.6)--(0.4,2.4);
\end{tikzpicture}
$+$
\begin{tikzpicture}[baseline=(current bounding box.center)]
\draw (0.4,0.6)--(0.4,2.4);
\draw (0,2)--(1,2)--(1.4,1.5)--(1.8,2);
\draw[very thick] (0,1)--(0.4,1)--(0.4,2)--(1,2)--(1.8,1);
\draw (0.4,1)--(1,1)--(1.4,1.5);
\end{tikzpicture}
\end{center}

Above, each diagram represents a path configuration whose weight can be computed using the vertex weights shown in Figure~\ref{fig:ybe}. In this case, the Yang--Baxter Equation reads
\begin{equation}
\frac{u/v (1-t)}{1- t u/v} \frac{t (1-u)}{1- t u} =
 \frac{t (1-v)}{1- t v} \frac{u/v (1-t)}{1- t u/v} + \frac{1-t}{1- t v} \frac{u(1-t)}{1- t u} \frac{t (1- u/v)}{1 - t u/v}
\end{equation}
and this identity of rational functions is indeed readily checked to hold.
\end{proof}

\begin{remark}
Another way to write the Yang--Baxter Equation is
\begin{align}\label{eqn:YBE_mat}
R_{2 3}(v) R_{1 3}(u) R_{1 2}(u/v) = R_{1 2}(u/v) R_{1 3}(u) R_{2 3}(v) .
\end{align}
The above is an equality of operators in $(\mathbb{C}^2)^{\otimes 3}$, and~\eqref{eqn:YBE} identifies their~$(2^3)^2$ matrix elements.

The operators are described as follows. Let $e_0, e_1$ be the standard basis of $\mathbb{C}^2$; these basis vectors correspond to ``no path" and ``path", respectively. First we define a family of operators $R(u) : (\mathbb{C}^2)^{\otimes 2} \rightarrow (\mathbb{C}^2)^{\otimes 2}$, which in the basis $\{e_0 \otimes e_0, e_0 \otimes e_1,  e_1 \otimes e_0, e_1 \otimes e_1\}$ act as
\begin{align*}
R(u) =
\begin{pmatrix}
1 & 0 & 0 & 0 \\
0 & \frac{1-u}{1- t u} & \frac{1-t}{1 - t u} & 0 \\
0 & \frac{u (1-t) }{1 - t u} & \frac{t (1-u)}{1 - t u} & 0 \\
0 & 0 & 0 & 1
\end{pmatrix} .
\end{align*}
Then for $l \neq p \in \{1,2,3\}$, define $R_{l p} : (\mathbb{C}^2)^{\otimes 3} \rightarrow (\mathbb{C}^2)^{\otimes 3}$ to act by $R$ in factors $l$ and $p$ of the tensor product and identically in the other factor. Figure~\ref{fig:YBEcode} gives code to check Theorem~\ref{thm:YBE} in R-notations.

\end{remark}

We can now deduce the symmetry claimed in the proof of Theorem~\ref{thm:IK}.
\begin{lemma}\label{lem:symmetry}
Let $Z(x_1,\dots, x_N; y_1, \dots, y_N; t)$ denote the partition function of the six-vertex model with DWBC. Then
$$Z(x_1,\dots, x_N; y_1, \dots, y_N; t) = Z(x_1,\dots, x_N; y_1, \dots, y_{i-1},y_{i+1}, y_i,\dots, y_N; t)$$
and
$$Z(x_1,\dots, x_N; y_1, \dots, y_N; t) = Z(x_1,\dots, x_{i-1}, x_{i+1}, x_i, \dots, x_N; y_1, \dots, y_N; t).$$
\end{lemma}
\begin{proof}
We prove the first equation using a graphical argument. (The second equation is proven in the same way.) First we add a cross at the left boundary of the lattice, between rows $i$ and $i+1$, which does not change the partition function since the ``empty" cross has weight $1$ according to the formulas~\eqref{eqn:weights}. Then we drag the cross past each column sequentially, and at each step by the Yang--Baxter Equation the partition function is preserved. Finally, when the cross arrives at the right boundary we have swapped $y_i$ with $y_{i+1}$, and we have a ``fully occupied" cross which also has weight $1$ by~\eqref{eqn:weights}, so it can be removed. Figure~\ref{fig:graphical_arg} below illustrates this argument for~$N = 4$ and~$i=2$.

\begin{figure}[hbt!]
\centering
\includegraphics[scale=.7]{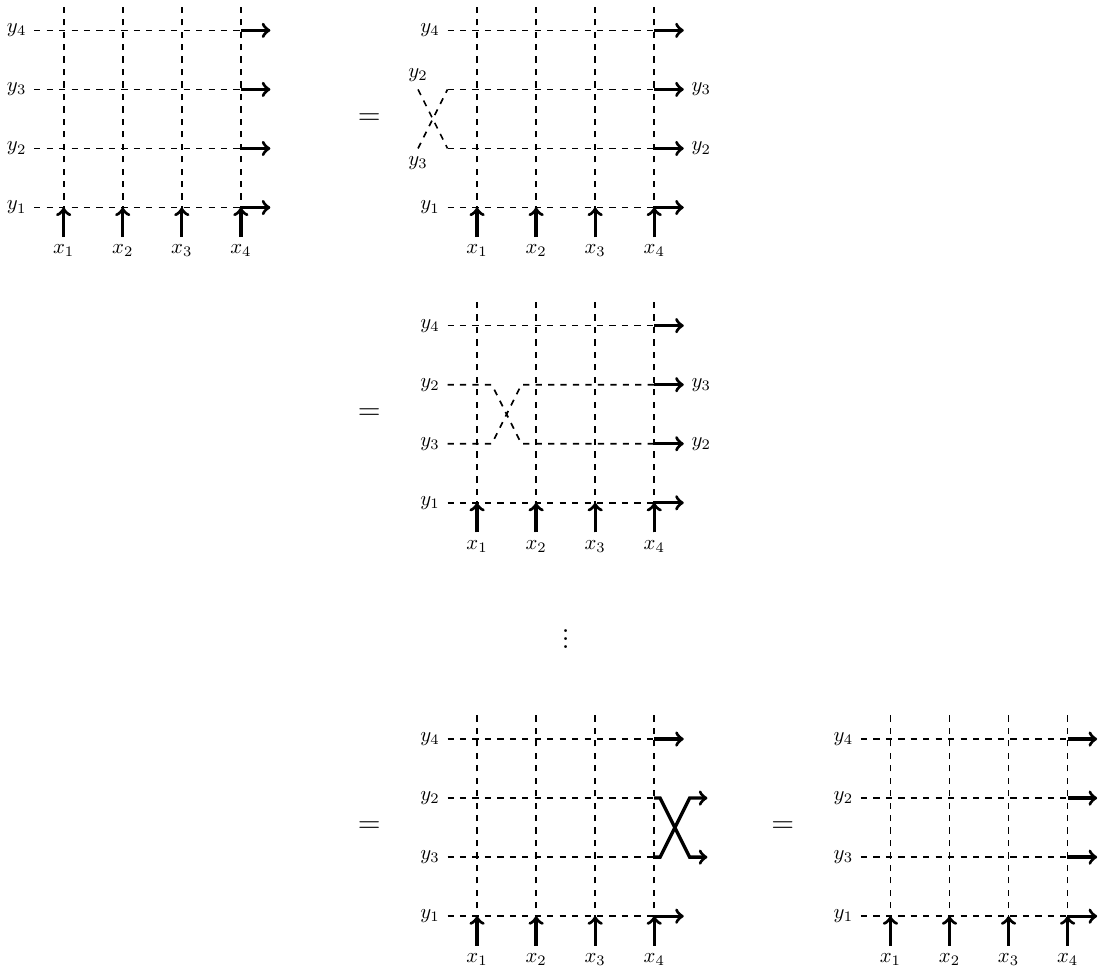}
\caption{Moving a cross from left to right in the proof of Lemma~\ref{lem:symmetry}.}
\label{fig:graphical_arg}
\end{figure}
In Figure~\ref{fig:graphical_arg}, each diagram represents a partition function with given boundary conditions. Furthermore, at each step the cross vertex is equipped with spectral parameter~$\frac{u}{v} = \frac{y_2}{y_3}$, which allows us to employ the Yang--Baxter Equation.

\end{proof}

\section{GUE-Corners Process in 6v model with DWBC at $\Delta = 0$}

\label{Section_GUE}

We use the weights \eqref{eqn:weights} (see also Figure \ref{fig:vertices}) with $t=-1$ throughout this section, implying that $\Delta= \frac{a_1 a_2 + b_1 b_2 - c_1 c_2}{2 \sqrt{a_1 a_2 b_1 b_2}}=0$.
Our goal is to prove that on the first several rows of a random six-vertex
configuration with DWBC on an~$N \times N$ square, the positions of~\ur vertices
are asymptotically described by the GUE corners process, which is the distribution of eigenvalues of ``top left''~$k \times k$ sub-matrices of a GUE random matrix.

The parameter $u\in\mathbb C$ in the weights \eqref{eqn:weights} is chosen to be any complex number satisfying $|u|=1$ with $u\neq -1$. We note that the ratios $\frac{a_1 a_2}{b_1 b_2}$ and $\frac{a_1 a_2}{c_1 c_2}$ are positive real numbers and, therefore, the weights \eqref{eqn:weights} define a probability measure with positive real probabilities (because by Lemma \ref{lem:Gauge} they can be transformed into positive weights). We denote by $\mathbb{P}_N$ the corresponding probability distribution on states for the six-vertex model with domain wall boundary conditions in the $N \times N$ square. Here is the main theorem of this section.

\begin{theorem}\label{thm:GUE_fluct}
Using the weights \eqref{eqn:weights}, suppose $t=-1$, $|u|=1$, $u\ne -1$ and consider a random configuration with DWBC in $N\times N$ square. Let $\xi_i^j$ be the position of the $i^{th}$ vertex of the form \ur (this is a $c_1$ vertex) in the $j^{th}$ row. Also, let~$\{\lambda_i^j\} $ denote the GUE corners process defined in Section~\ref{subsec:ex}. Then as $N \rightarrow \infty$, with $\gamma = \frac{(1-u)(1-u^{-1})}{4}$, we have:
\begin{enumerate}
\item $\displaystyle \mathbb{P}_N\bigl(\# (\text{\ur vertices in $j^{th}$ row }) = j\bigr) \rightarrow 1 $,
\item $\displaystyle  \left\{ \frac{\xi_i^j - \gamma N}{\sqrt{\gamma (1-\gamma)} \sqrt{N}} \right\}_{1 \leq i \leq j} \stackrel{d}{\rightarrow}  \{\lambda_i^j\}_{1 \leq i \leq j} $ in finite-dimensional distributions.
\end{enumerate}
\end{theorem}
Setting~$u = \i$ gives Theorem \ref{thm:JNeig} of Section~\ref{subsec:ex}. The proof of Theorem~\ref{thm:GUE_fluct} is presented in a series of lemmas and occupies the rest of this section. Our analysis crucially uses a simplification of the determinant of Theorem \ref{thm:IK} in the $t=-1$ case. This is based on a computation known as the Cauchy determinant formula.

\begin{lemma}[Cauchy determinant] \label{Lem:Cauchy_det}
We have the following equality of rational functions in variables~$a_1, \dots, a_N, b_1,\dots, b_N$:
$$
\det\left( \frac{1}{a_j-b_i} \right)_{i,j  = 1}^N = \frac{ \prod_{i < j} (a_i - a_j) (b_j - b_i)}{\prod_{i, j} (a_i - b_j)}.
$$
\end{lemma}
\begin{proof}
We proceed by induction. The base case $N=1$ is clear.

For the induction step, we compare two expressions
\begin{equation}\label{eqn:Cdet1}
\prod_{i, j} (a_i - b_j) \det\left( \frac{1}{a_j-b_i} \right)_{i,j  = 1}^N \qquad \text{ and }\qquad  \prod_{i < j} (a_i - a_j) (b_j - b_i) .\end{equation}
Viewing both as polynomials in $a_N$ and fixing generic values for all of the other variables, we note that both
sides are degree $N-1$ polynomials in variable $a_N$. For each side there are $N-1$ zeros at $a_N = a_i$ for $i = 1,\dots, N-1$ (for the left expression this is because of the vanishing of a determinant with two equal rows).
Further, using the induction assumption and multiplying the last row of the matrix under the determinant in the left expression by $(a_N-b_N)$, we conclude that both sides have equal value at $a_N = b_N$. This value is given by
$$\prod_{i < N} (a_i - b_N) (b_N - b_i) \cdot \prod_{i, j=1}^{N-1} (a_i - b_j) \det\left( \frac{1}{a_j-b_i} \right)_{i,j  = 1}^{N-1}.$$
Two degree $(N-1)$ polynomials in $a_N$ with $N$ equal values must coincide, hence, two expressions in \eqref{eqn:Cdet1} are the same.
\end{proof}

\begin{remark}
There are many other proofs of this formula. For example, one can use skew-symmetry of the first expression in~\eqref{eqn:Cdet1} to deduce that it is divisible by the second one.
\end{remark}

The next step is to define auxiliary random variables useful for proving the convergence to the GUE corners process.

From here on we call a vertex of type~$b_1$ or~$b_2$ a ``$b$ vertex", and similarly we call a vertex of type~$c_1$ or~$c_2$ a ``$c$ vertex''. For $j = 1,\dots, N$, define a random variable $\eta_j$ by the property that
\begin{equation}
\label{eq_b_vertices}
\# \text{($b$ vertices in row $j$) }= \gamma N + \sqrt{N}  \eta_j.
\end{equation}
\begin{lemma}\label{lem:eta_conv}
For any fixed $k$ we have the convergence in distribution
$$\lim_{N\to\infty}(\eta_1, \dots, \eta_k)\stackrel{d}{=} (X_1, \dots, X_k)$$
 where~$\{X_i\}$ are independent Gaussians with mean~$0$ and variance~$\gamma (1-\gamma)$.
\end{lemma}

We will first prove Lemma \ref{lem:eta_conv}, and then use this and the Gibbs property satisfied by the interlacing random variables~$\xi_i^j$ to prove Theorem~\ref{thm:GUE_fluct}.

The convergence in distribution in Lemma~\ref{lem:eta_conv} follows from the convergence of two-sided Laplace transforms (see e.g.\ \cite[Chapter 6]{kallenberg_3rd_ed} for general discussion), with the latter summarized in the following statement:

\begin{lemma}
For any~$z_1,\dots, z_k \in \mathbb{C}$, we have the convergence
\begin{equation}
\lim_{N\to\infty} \mathbb{E} \Bigg[ \exp\left( \frac{- 2 \i u}{1-u^2} \sum_{i=1}^k \eta_i z_i   \right)   \Bigg] = \exp\left(\frac{1}{8}\sum_{i=1}^k z_i^2\right) . \label{eqn:lap}
\end{equation}
\label{lem:ind_gauss}
\end{lemma}

\begin{remark}\label{rmk:complete}
The statement of \eqref{eqn:lap} matches the scaling of Lemma~\ref{lem:eta_conv}, since
$$\frac{- 2 \i u}{1-u^2} = \frac{1}{2\sqrt{\gamma (1-\gamma)}} $$
and $\exp\left(\frac{1}{8}\sum_{i=1}^k z_i^2\right)$ is the (two-sided) Laplace transform of the vector of $k$ i.i.d.\ mean $0$ Gaussian random variables of variance $\frac{1}{4}$.
\end{remark}

\begin{proof}[Proof of Lemma \ref{lem:ind_gauss}]
We proceed in steps.
\begin{enumerate}[1.]
\item First, we explicitly compute the Izergin-Korepin determinant when $t = -1$. We claim that
\begin{equation}
Z(x_1,\dots, x_N; y_1,\dots, y_N; -1) = \prod_{i=1}^N (2 x_i y_i) \frac{\prod\limits_{i < j} (x_i + x_j)(y_i + y_j)}{\prod\limits_{i,j=1}^N (1+x_i y_j)}.
\label{eqn:ffdet}
\end{equation}
To see this, note that $(1- x_i y_j)(1-t x_i y_j) = 1-x_i^2 y_j^2 = y_j^2 (y_j^{-2} - x_i^2)$. Hence, we can apply the Cauchy determinant formula of Lemma \ref{Lem:Cauchy_det} to the Izergin-Korepin determinant, to obtain
\begin{equation}
\label{eq_x1}
\prod_{i=1}^N (2 x_i y_i) \frac{\prod\limits_{i, j=1}^N(1-x_i y_j)}{\prod\limits_{i < j} (x_i - x_j)(y_i - y_j)}  \frac{\prod\limits_{i < j} (x_i^2 - x_j^2)(y_j^{-2}-  y_i^{-2} )}{\prod\limits_{j=1}^N y_j^2 \prod\limits_{i,j=1}^N (y_j^{-2} - x_i^2) }.
\end{equation}
\eqref{eqn:ffdet} is an equivalent form of \eqref{eq_x1}.

\item  Next, we compute the (two-sided) Laplace transform $\mathbb{E}[e^{z_1 \eta_1 + \cdots + z_k \eta_k}]$ using Step 1. In the following formula we denote $(\underbrace{u,\dots, u}_{N\text{ times}})$ as $u^N$ and similarly $(\underbrace{1,\dots, 1}_{N\text{ times}})$  as $1^N$. For any $z_1, \dots, z_k \in \mathbb{C}$, we claim that as $N\to\infty$
\begin{align}
&\frac{Z_N(u^N; e^{i z_1/\sqrt{N}}, \dots, e^{i z_k/\sqrt{N}}, 1^{N-k}; -1)}{Z_N(u^N; 1^N; -1)} \notag \\
& =
\exp\left(\sqrt{N} \i  \frac{1-u}{2 (1+u)}(z_1 + \cdots + z_k)\right)\cdot \exp\left( - \frac{(1-u)^2}{(1+u)^2} \left(\frac{z_1^2 +\cdots+ z_k^2}{8}\right)  + o(1)\right)  \label{eqn:PFformula}.
\end{align}
In order to see this, we calculate using the formula \eqref{eqn:ffdet}. More precisely, taking $\log$ of the ratio of partition functions gives
$$N\left( \sum_{i=1}^k \log\left( \frac{1 + u}{1 + u e^{\i z_i/\sqrt{N}}} \right) + \sum_{i=1}^k  \log\left( \frac{1 + e^{\i z_i/\sqrt{N}}}{2} \right) \right) + O\left(\frac{1}{\sqrt{N}}\right).$$
Taylor expanding
$$\log\left( \frac{1 + u}{1 + u e^{\i z /\sqrt{N}}} \right) +   \log\left( \frac{1 + e^{\i z/\sqrt{N}}}{2} \right) $$
gives
$$\frac{\i (1 - u) z }{2 (1 + u)} \frac{1}{\sqrt{N}} - \frac{(1-u)^2 z^2}{8 (1 + u)^2}\frac{1}{ N} + O\left(\frac{1}{(\sqrt{N})^3}\right),$$
and this implies the result \eqref{eqn:PFformula}.

\item Now we can prove our claimed convergence of the Laplace transform of $(\eta_1, \dots, \eta_k)$.  Using $weight$ and $weight'$ to denote the local vertex weights of six-vertex configurations with parameters~$(u^N; 1^{N};-1)$ and
~$(u^N; e^{i z_1/\sqrt{N}}, \dots, e^{i z_k/\sqrt{N}}, 1^{N-k};-1)$, respectively, we have that~\eqref{eqn:PFformula} equals
\begin{align*}
\frac{\sum \prod_{(i,j)} \text{weight}'(i,j; S) }{\sum \prod_{(i,j)} \text{weight}(i,j; S)}
&=  \mathbb{E}\left[\frac{\prod_{(i,j)} \text{weight}'(i,j; S)}{\prod_{(i,j)} \text{weight}(i,j; S)}\right]
\end{align*}
where $\mathbb{E}$ means expectation over random six-vertex configurations~$S$ sampled using $weight$, i.e. using homogenous parameters $(u^N; 1^{N};-1)$. Calculating the ratio
$$\frac{\prod_{(i,j)} \text{weight}'(i,j; S)}{\prod_{(i,j)}\text{weight}(i,j; S)}$$
for any configuration, we get that~\eqref{eqn:PFformula} equals
\begin{align}
 \mathbb{E} \Bigg[ \prod_{j=1}^k & \left( \frac{1- u e^{\i z_j /\sqrt{N}}}{1- u} \cdot \frac{1 + u}{1 + e^{\i z_j / \sqrt{N}} u} \right)^{\text{(\# $b$ vertices in row $j$)} }  \label{eqn:avg}
 \left(1 + O(N^{-1/2})\right)^{\text{(\# $c$ vertices in row $j$)} } \Bigg] . %\notag
\end{align}
By Lemma \ref{Ex_1} of Section \ref{subsec:ex}, the number of $c_1$ and $c_2$ vertices in rows $1$ through $k$ is $O(1)$, so the factor involving~$c$ vertices is irrelevant. Using this, the expression under expectation in \eqref{eqn:avg} is transformed into the product of $k$ factors of the form
\begin{align}
\exp \Bigg( \left(\frac{2 \i u z_j}{(-1 + u) (1 + u)}  \frac{1}{\sqrt{N}} + \frac{
 u (1 + u^2) z_j^2}{(-1 + u)^2 (1 + u)^2 } \frac{1}{N} \right)\left( \gamma N + \sqrt{N}  \eta_j \right) +O\left(\frac{1}{\sqrt{N}}\right) \Bigg) \label{eqn:Lap_exp}.
\end{align}

 We chose $\gamma $ exactly so that the deterministic leading order factor in \eqref{eqn:avg} coming from the leading order term $O(\sqrt{N})$ under exponent in \eqref{eqn:Lap_exp} matches with the prefactor $\exp\left( \sqrt{N}  \i  \frac{1-u}{2 (1+u)}   (z_1 + \cdots + z_k) \right)$ in~\eqref{eqn:PFformula}. Expanding and summing up the sub-leading terms in~\eqref{eqn:Lap_exp}, at the next $O(1)$ order we get
$$\exp\left(\sum_{i=1}^k \frac{2 \i \eta_i u z_i}{-1 + u^2} + \frac{\gamma (u + u^3) z_i^2}{(-1 + u)^2 (1 + u)^2}\right).$$
Comparing with~\eqref{eqn:PFformula}, after moving the deterministic factor in the sub-leading term to the other side and simplifying, we get
\begin{align*}
\E\left[\exp\left(\sum_{i=1}^k \frac{2 \i \eta_i u z_i}{-1 + u^2} \right)\right] &=
\exp\left(\sum_{i=1}^k - \frac{\gamma (u + u^3) z_i^2}{(-1 + u)^2 (1 + u)^2} - \frac{(1-u)^2 z_i^2}{8 (1 + u)^2} + O\left(\frac{1}{\sqrt{N}}\right)\right) \\
&= \exp\left(\frac{1}{8} \sum_{i=1}^k z_i^2\right)\left(1 + O\left(\frac{1}{\sqrt{N}}\right)\right).
\end{align*}

From this, one deduces~\eqref{eqn:lap}, and this concludes the proof of Lemma~\ref{lem:ind_gauss}.\qedhere
\end{enumerate}
\end{proof}

In order to prove Theorem \ref{thm:GUE_fluct} we need to make a bridge between Lemma~\ref{lem:ind_gauss} and convergence of $\{\xi^j_i\}$ to GUE eigenvalues. This is our task for the rest of the section.

\begin{figure}[t]
\centering
\includegraphics[scale=.9]{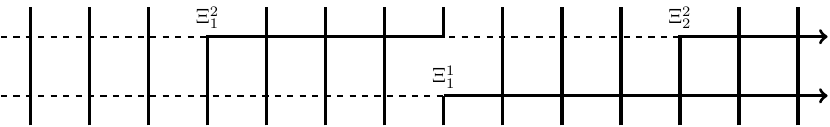}
\caption{\label{fig:dwbc_generic}The first two rows of a generic configuration with domain wall boundary conditions and random variables $\Xi_1^1$, $\Xi_1^2$, $\Xi_2^2 $.}
\end{figure}

Towards this end, we define~$\Xi_i^j$ as the position of the~$i^{th}$ \emph{empty} vertical edge between rows~$j$ and~$j+1$ in the lattice, as in Figure \ref{fig:dwbc_generic}. Due to combinatorial constraints, we always have the interlacing
$$\Xi_1^j \leq \Xi_{1}^{j-1} \leq \Xi_{2}^{j} \leq \cdots \leq  \Xi_{j-1}^{j} \leq \Xi_{j-1}^{j-1} \leq \Xi_{j}^{j}. $$
In what follows we  relate the~$\eta_j$'s to~$\{\Xi_i^j\}$, and show that our result for~$\eta_j$ implies that~$\{\Xi_i^j\}$ converges to the GUE corners process. In parallel, we relate the~$\{\Xi_i^j\}$ to~$\{\xi_i^j\}$.

Analyzing the vertices in the first row of a configuration, we have
$$\Xi_1^1 = (\# \text{b vertices in row $1$}) + 1.$$
Looking at the second row, if~$ \Xi_1^2 \leq \Xi_1^1 < \Xi_{2}^2$, then we have
$$(\Xi_1^2 + \Xi_2^2)-\Xi_1^1 = \text{(\# b vertices in row $2$) } + 2.$$
If, on the other hand,~$ \Xi_1^2 < \Xi_1^1 = \Xi_{2}^2$, then we have
$$(\Xi_1^2 + \Xi_2^2)-\Xi_1^1 =  \Xi_1^2 =    \text{(\# b vertices in row $2$) } .$$

As we will see later, the ``generic'' configurations from a randomly sampled six-vertex state are those where~$ \Xi_1^2 < \Xi_1^1 < \Xi_{2}^2$. See Figure~\ref{fig:dwbc_generic} for an example of such a configuration.

More generally, on the event that
\begin{equation}\label{eqn:event}
\Xi_i^j < \Xi_{i}^{j-1} < \Xi_{i+1}^{j}\qquad \text{ for all } i = 1,\dots, j-1,\quad j=2,\dots, k,
\end{equation}
we have
$$\# (\text{\ur vertices in $j^{th}$ row }) = j \qquad \text{ for all } j = 1,\dots, k$$
and also
$\xi_i^j = \Xi_i^j$ for all~$1 \leq i \leq j \leq k$. Furthermore, on this event we have for each~$j=1,\dots, k$
\begin{equation}\label{eqn:difs}
(\Xi_1^j +\cdots+ \Xi_j^j)-(\Xi_{1}^{j-1} +\cdots+ \Xi_{j-1}^{j-1}) = \text{(\# b vertices in row $j$) } + j .
\end{equation}
Away from the event~\eqref{eqn:event} there are~$O(1)$ corrections to the right hand side above, where the~$O(1)$ constant may depend on~$j$ but not on~$N$.

So recalling the definition of~$\eta_j$ from \eqref{eq_b_vertices}, the discussion above and Lemma~\ref{lem:eta_conv} give us the following.
\begin{lemma}\label{lem:diag_conv}
Denoting $\widetilde{\Xi}_i^j=\widetilde{\Xi}_i^j(N) \defeq \frac{\Xi_i^j - \gamma N}{\sqrt{\gamma (1-\gamma)} \sqrt{N}}$, as~$N \rightarrow \infty$, the random vector
$$\bigl(\widetilde{\Xi}_1^1,\,\, (\widetilde{\Xi}_1^2 + \widetilde{\Xi}_2^2)-\widetilde{\Xi}_1^1,\,\, \dots,\,\, (\widetilde{\Xi}_1^k +\cdots+ \widetilde{\Xi}_k^k)-(\widetilde{\Xi}_{1}^{k-1} +\cdots+\widetilde{\Xi}_{k-1}^{k-1})\bigr)$$
converges in distribution to the vector of $k$ i.i.d. mean~$0$, variance $1$ Gaussian random variables.
\end{lemma}

\begin{remark}
Lemma~\ref{lem:diag_conv} is consistent with the result of Theorem \ref{thm:GUE_fluct} we are proving, as we expect that with overwhelming probability  we have $\xi_i^j = \Xi_i^j$, and on the GUE side the corresponding limiting quantities are $\lambda_1^1 = m_{1 1}$, and $(\lambda^2_1 + \lambda^2_2) - \lambda_1^1 = m_{2 2}$, and $(\lambda^3_1 + \lambda^3_2+\lambda^3_3) - (\lambda_1^2+\lambda_2^2) = m_{3 3}$, and so on, where $m_{i j}$ are the matrix entries of a GUE random matrix. This computation follows from the equivalence of two ways to compute the trace of a matrix: as the sum of diagonal matrix elements or as the sum of eigenvalues. By definition, the diagonal entries are independent Gaussians with unit variance.
\end{remark}

Our next task is to upgrade Lemma \ref{lem:diag_conv} to convergence of~$\{\widetilde{\Xi}_i^j\}$ towards the GUE corners process. This relies on using the uniform Gibbs property or conditional uniformity which we now introduce.

\begin{defn} \label{def:GT_pattern}
We call a set of $k(k+1)/2$ reals~$\{\zeta_i^j\}_{1 \leq i \leq j \leq k}$ a (continuous) \emph{Gelfand-Tsetlin pattern} if
$$ \zeta_1^j \leq \zeta_{1}^{j-1} \leq \zeta_{2}^{j} \leq \cdots \leq  \zeta_{j-1}^{j} \leq \zeta_{j-1}^{j-1} \leq \zeta_{j}^{j}, \quad \text{ for all } j = 1,\dots, k  .$$
The $k$-tuple $(\zeta_1^k,\zeta_2^k,\dots,\zeta_k^k)$ is referred to as the \emph{top row} of the Gelfand-Tsetlin pattern.
\end{defn}

\begin{defn}\label{def:gibbs}
We say a probability measure on Gelfand-Tsetlin patterns $\{\zeta_i^j\}_{1 \leq i \leq j \leq k}$ satisfies the \emph{uniform Gibbs property} if conditional on $\zeta^k=(\zeta_1^k,\zeta_2^k,\dots,\zeta_k^k)$, the distribution of $\{\zeta_i^j\}_{1 \leq i \leq j \leq k}$ is uniform among all Gelfand-Tsetlin patterns with top row $\zeta^{k}$.
\end{defn}

Now we state one more preparatory lemma which we require before we can prove Theorem~\ref{thm:GUE_fluct}.

\begin{lemma}\label{lem:tightness}
The sequence of random  Gelfand-Tsetlin patterns $\{\widetilde{\Xi}^j_i(N)\}_{1 \leq i \leq j \leq k}$ of Lemma \ref{lem:diag_conv} satisfies two properties for each fixed $k=1,2,\dots$:
\begin{itemize}
\item  The sequence~$\{\widetilde{\Xi}^j_i(N)\}_{1 \leq i \leq j \leq k}$ is a tight sequence of random vectors.
\item Any subsequential $N\to\infty$  limit in distribution of $\{\widetilde{\Xi}^j_i(N)\}_{1 \leq i \leq j \leq k}$ satisfies the uniform Gibbs property of Definition \ref{def:gibbs}.
\end{itemize}
\end{lemma}
We sketch a proof of this lemma and refer the reader to \cite{gorin2014alternating} for more details; see in particular the proof of Lemma 7 there.
\begin{proof}[Sketch of the proof of Lemma \ref{lem:tightness}]
First, we deal with the second statement on the asymptotic uniform Gibbs property, which says that conditional on some fixed top row~$\widetilde{\Xi}^k(N)$, as~$N \rightarrow \infty$, the Gelfand-Tsetlin pattern $\{\widetilde{\Xi}^j_i(N)\}_{1 \leq i \leq j \leq k-1}$ becomes uniformly distributed. This follows from the Gibbs property for the six-vertex model. Indeed, suppose we fix and condition on the values of $(\Xi_1^k,\Xi_2^k,\dots,\Xi_k^k)$ equal to some integers $\lambda_1 < \lambda_2 < \cdots < \lambda_k$. Fixing $(\Xi_1^k,\Xi_2^k,\dots,\Xi_k^k)$ is equivalent to fixing boundary conditions for the six-vertex model on the first~$k$ rows, and the values  of ${\Xi}_i^j$, $1\le i\le j \le k$ completely determine the configuration on these $k$ rows; thus, we must consider the Boltzmann measure in the rectangle $[1, N] \times [1, k]$ with the corresponding boundary conditions (which look like DWBC along the left, bottom, and right boundaries, and are determined by $(\Xi_1^k,\Xi_2^k,\dots,\Xi_k^k)$ along the top boundary). The weight of a configuration in this rectangle, corresponding to $\{\Xi_i^j\}_{1\le i\le j \le k}$ satisfying the property \eqref{eqn:event} can be written (taking into account that $a_1=a_2=1$ for weights \eqref{eqn:weights} we use and that there are no $b_2$ vertices under \eqref{eqn:event}) as
\begin{equation}
\label{eq_6v_Gibbs}
b_1^{-k(k+1)/2 + \sum_{i=1}^k\Xi_i^k} c_1^{k(k+1)/2}  c_2^{k(k-1)/2}
\end{equation}
and if \eqref{eqn:event} is not satisfied, the weight will be the above multiplied by a correcting prefactor, which stays bounded away from $0$ and $\infty$ as $N\to\infty$. Note that \eqref{eq_6v_Gibbs} depends only on $(\Xi_1^k,\Xi_2^k,\dots,\Xi_k^k)$, but not on $\Xi_i^j$ with $j<k$. The configurations not satisfying~\eqref{eqn:event} have~$O(1)$ weight.

We claim that the number of the configurations in the rectangle satisfying \eqref{eqn:event} is an order of magnitude larger than the number not satisfying \eqref{eqn:event}, which, from the discussion in the previous paragraph, implies that asymptotically only the configurations satisfying \eqref{eqn:event} matter, which in turn implies that the desired uniform Gibbs property is satisfied asymptotically.

In order, to prove the claim, note that in our asymptotic regime $\Xi_i^k - \Xi_{i-1}^k \propto \sqrt{N}$ and the number of configurations in \eqref{eqn:event} can be readily seen to be growing as~$\sqrt{N}^{k (k-1)/2}$. On the other hand, the number of configurations not satisfying \eqref{eqn:event} is $O\left(\sqrt{N}^{k (k-1)/2-1}\right)$, because there is one less degree of freedom.

\bigskip

We now turn to the first statement on the tightness as $N\to\infty$. The essential idea is that if the magnitude of an element~$\widetilde{\Xi}_i^{j_0}$ of the random vector is very large with a probability bounded away from $0$, then by the Gibbs property (which approximately induces the uniform measure, by the argument above) some point~$\widetilde{\Xi}_{i'}^{j_0-1}$ on the previous is also very large with a probability bounded away from $0$. Thus, one may proceed by induction on the maximum level~$k$. Note that the base case $k=1$ holds because~$ \widetilde{\Xi}_1^1 = \eta_1/\sqrt{\gamma (1-\gamma)} $, and we have already shown in Lemma \ref{lem:diag_conv} that~$\eta_1$ converges in distribution to a Gaussian random variable, implying the tightness for the random variables~$\eta_1 = \eta_1(N)$).
\end{proof}

Now we are in a position to prove Theorem~\ref{thm:GUE_fluct}.

\begin{proof}[Proof of Theorem~\ref{thm:GUE_fluct}]

Let us consider an arbitrary subsequential limit in distribution $\{\zeta_i^j\}_{1 \leq i \leq j \leq k}$ of the sequence of arrays~$ \{\widetilde{\Xi}_i^j(N)\}_{1 \leq i \leq j \leq k}$ of Lemma \ref{lem:diag_conv}. It suffices for us to show that~$\{\zeta_i^j\}_{1 \leq i \leq j \leq k}$ coincides with the GUE corners process. Indeed, with the tightness from Lemma~\ref{lem:tightness}, this would show that the sequence of random arrays~$\{\widetilde{\Xi}_i^j\}$ converges to the GUE corners process, which in turn implies the two statements in the theorem by the discussion before Lemma \ref{lem:diag_conv}.

Lemma~\ref{lem:tightness} implies that~$\{\zeta_i^j\}_{1 \leq i \leq j \leq k}$ satisfies the uniform Gibbs property. Lemma~\ref{lem:diag_conv} implies that if we define
 \begin{equation}\label{eqn:Mdiag}
 m_{j j} \defeq \sum_{i=1}^j \zeta_i^j  - \sum_{i=1}^{j-1} \zeta_i^{j-1}
 \end{equation}
 then~$(m_{j j})_{j=1}^k$ is a tuple of i.i.d.\ standard normal random variables. (The notation~$m_{j j}$ is motivated by the definition of the matrix~$M$ below.)

Next, we argue that if we have a random Gelfand--Tsetlin pattern $\{\zeta_i^j\}_{1\le i \le j \le k}$ which satisfies the uniform Gibbs property and the property that the random variables $m_{j j}$ as defined in \eqref{eqn:Mdiag} are i.i.d.\ mean $0$, variance $1$ Gaussians, then~$\{\zeta_i^j\}$ is distributed as the GUE corners process.  To prove this claim, we follow the strategy of \cite{gorin2014alternating}. We construct a unitarily invariant Hermitian random matrix as follows: first, sample its eigenvalues~$\nu_1,\dots, \nu_k$ according to the distribution of~$(\zeta_1^k,\dots, \zeta_k^k)$, and put the eigenvalues into a diagonal matrix~$\Lambda$; then independently sample a matrix~$U$ from the Haar measure on all $k\times k$ unitary matrices and construct the matrix
$$M = U \Lambda U^{-1}.$$

This gives a random Hermitian matrix~$M$. Define~$\{\nu_i^j\}_{1 \leq i \leq j \leq k}$ as the eigenvalues of its top-left $j\times j$ corners. We claim that
\begin{enumerate}[i)]
\item The distribution of~$\{\nu_i^j\}$ equals that of~$\{\zeta_i^j\}$.

\item The distribution of~$M$ is the same as the GUE; in particular,~$\{\nu_i^j\}$ is distributed as the GUE corners process.
\end{enumerate}

\begin{proof}[Proof of i):] This follows from the facts (a)~$(\nu_1^k, \dots,\nu_k^k) \stackrel{d}{=} (\zeta_1^k,\dots, \zeta_k^k)$ by definition, and (b) both Gelfand--Tsetlin patterns satisfy the uniform Gibbs property. The fact that given~$(\nu_1^k,\dots, \nu_k^k)$, the Gelfand--Tselin pattern $\{\nu_i^j\}_{1 \leq i \leq j \leq k-1}$ is uniformly distributed among all Gelfand-Tsetlin patterns with such top-row is a well-known property of the pushforward of the Haar measure under the map (with~$\Lambda$ fixed)
$$U \mapsto U \Lambda U^{-1},$$
 see \cite[Proposition 4.7]{Baryshnikov_GUE2001} or \cite{neretin2003rayleigh} for modern proofs or \cite[Section 9.3]{gel1950unitary} for earlier discussions.
\end{proof}

\begin{proof}[Proof of ii):] This follows by computing the Fourier transform of~$M$: Let~$Z$ be an arbitrary Hermitian matrix. We claim that
$$\mathbb{E}[\exp(\i \text{Tr}(M Z))] = \exp\left(\frac{1}{2} \sum_{i=1}^k z_i^2 \right),$$
where~$z_i$ are the eigenvalues of~$Z$. This follows from the unitary invariance of the distribution of~$M$. Indeed, let~$V$ be a (deterministic) unitary matrix such that
$$V Z V^{-1} = D \defeq
 \begin{pmatrix}
z_1 &  & 0 \\
 &   \ddots &  \\
 0&             & z_k
\end{pmatrix}.
$$
By the cyclic property of trace and since~$ V^{-1} M V \stackrel{d}{=} M$, we have
$$\mathbb{E}[\exp(\i \text{Tr}(M Z))] = \mathbb{E}[\exp(\i \text{Tr}(M D)] =  \exp\left(\frac{1}{2} \sum_{i=1}^k z_i^2 \right).$$
The last equality is a consequence of the fact that the diagonal entries of~$M$ are independent standard Gaussians by \eqref{eqn:Mdiag}. On the other hand, the exact same expression is the characteristic function of a Hermitian matrix sampled from the GUE, as the argument above can be used verbatim to compute on the GUE side. Two probability distributions on the space of $k\times k$ Hermitian matrices with the same Fourier transform must coincide. Thus, we have completed the proof of ii).
\end{proof}

Thus $M$ is indeed a GUE random matrix, and the eigenvalues of its leading principal submatrices, which are distributed as~$\{\zeta_i^j\}$, are the GUE corners process. This completes the proof of uniqueness of subsequential limits, which proves that~$\{\widetilde{\Xi}_i^j\}$ converges to the GUE corners process. As a consequence of this we obtain the two statements in Theorem \ref{thm:GUE_fluct}.
\end{proof}

\begin{remark}
The work of Olshanski and Vershik \cite{OlVer1996} contains a classification of probability measures on triangular arrays~$\{\zeta_{i}^j\}_{1 \leq i \leq j < \infty}$ whose finite dimensional marginals~$\{\zeta_i^j\}_{1 \leq i \leq j \leq k}$ satisfy the uniform Gibbs property. The classification theorem there is phrased in terms of infinite Hermitian matrices invariant under conjugation by unitary matrices. As we have just seen, the eigenvalues of corners of a random Hermitian matrix whose distribution is invariant under conjugations by the unitary group form a Gelfand-Tsetlin pattern satisfying the uniform Gibbs property, providing a link to our setting.

The classification of \cite{OlVer1996} says that the nontrivial ergodic measures on Hermitian matrices are distributions of  linear combinations of identity matrices, GUE random matrices, and several (perhaps, infinitely-many) independent rank $1$ matrices of the form $v v^*$, where $v$ is a column-vector with i.i.d.\ standard Gaussian components (the sums of the latter matrices are often called Wishart random matrices). In view of this classification and Lemma \ref{lem:tightness},  Theorem~\ref{thm:GUE_fluct} essentially says that in the limit all the Wishart terms do not appear, and we are only left with the GUE part.
\end{remark}

\section{Free boundary stochastic six-vertex model}

\label{Section_Stochastic_6v}

\subsection{A generalization of the Izergin--Korepin determinant}
\label{subsec:IK_generalization}

\begin{figure}
\centering
\includegraphics[scale=.97]{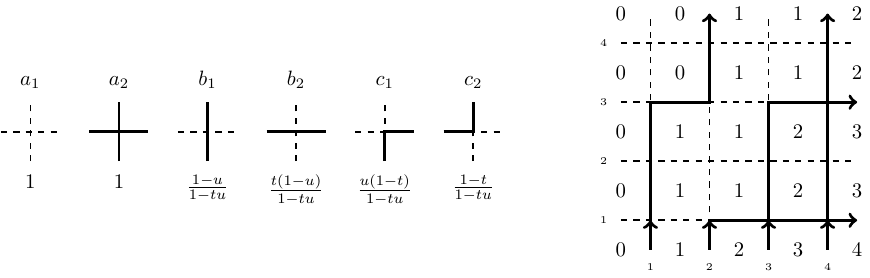}
\caption{A configuration of the stochastic six-vertex model in the~$4 \times 4$ square with the step initial condition and free exit data. Faces are labeled with their heights, for the height function of this configuration. $H_4$, the value of the height function at~$(4, 4)$, is equal to~$2$.}
\label{fig:step_H}
\end{figure}

In the previous sections we dealt with domain wall boundary conditions, which meant that the paths were entering through the bottom boundary (at each possible position) and exiting through the right boundary.

In this section we modify the boundary conditions and consider the six-vertex model with step initial condition and free exit data in $N\times N$ square. This means that paths enter along edges~$(i, 0) - (i, 1)$ for~$i =1 ,\dots, N$, and can exit the square anywhere; no paths enter from the left, as in Figure~\ref{fig:step_H}. We consider the stochastic weights of \eqref{eqn:weights} copied in Figure~\ref{fig:step_H} and require all of them to be positive, which can be achieved via either~$0 < u, t < 1$ or~$u, t > 1$.

Let us demonstrate how to sample the configuration in the~$N \times N$ square via a Markovian procedure, such that each configuration has probability equal to its weight. At the bottom left vertex~$(1,1)$ there is a single path entering from below. We sample the local configuration at~$(1,1)$ as follows: this path can either go up, with probability~$b_1$, or turn right, with probability~$1- b_1 = c_1$ and we make this choice by flipping an independent biased coin. In general, suppose we have sampled the configuration in vertices~$(i, j)$ in the square with~$i + j \leq k$. Then given the partial configuration, the allowed configurations at each vertex with~$i + j = k+1$ are independent of each other, so we sample them in parallel. We use the configuration we have already sampled to determine the local configuration at~$(i, j)$ as follows: if there is a single incoming vertical path (from below), then we flip a coin with probabilities~$b_1, c_1$ to determine if this path goes straight or turns right; if there is a single incoming horizontal path (from the left), then we flip a coin with probabilities~$b_2, c_2$ to determine if the path goes straight or turns up; if there are no paths or two paths, then deterministically the two outgoing edges of~$(i, j)$ will be empty or occupied, respectively. It is simple to check that sampling the configuration in the~$N \times N$ square with this procedure produces a configuration with probability exactly equal to its weight.

Let us emphasize that the possibility of the just described local sampling procedure for the configurations of the six-vertex model crucially depends on the boundary conditions, which are being kept free along the top and right boundaries. There is no way to make such a local construction for domain wall boundary conditions, no matter what vertex weights we choose. The observation about the existence of such a local stochastic algorithm goes back to \cite{GwaSpohn1992} for the model on the torus. More recently, on the plane such a construction was introduced in \cite{BCG6V} under the name the \emph{stochastic six-vertex model}.

An immediate corollary to the existence of the local sampling procedure is:
\begin{proposition}
For the stochastic six-vertex model in the~$N \times N$ square with step initial condition and free exit data, the partition function~$Z$ is equal to~$1$.
\end{proposition}
\begin{remark}
 The last statement remains true no matter whether $u$ and $t$ are fixed, or if they vary along the square in an arbitrary way and we use $t(i,j)$, $u(i,j)$ to sample the vertex at $(i,j)$.
\end{remark}

The~\emph{height function}, which is a function defined on the faces of the lattice, is the object commonly used to describe the asymptotic behavior of the six-vertex model. To define it, first one must fix its value at a single face. Then, it is defined by the local property that crossing a path in the down or right direction changes the height by~$+ 1$. See Figure~\ref{fig:step_H} for an example illustrating its definition. Theorems about its asymptotic behavior (for the stochastic six-vertex model) can be found in~\cite{BCG6V, Reshetikhin2018LimitSO, aggarwal2020limit, dimitrov2020two}. One of the main goals of the remainder of this note is to prove a special case of the main theorem of~\cite{BCG6V}: we will study the asymptotics of the height function at the point~$(N, N)$ as $N\to\infty$.

\begin{defn}
Let $H_N$ denote the value of the height function of the stochastic six-vertex model (with step initial condition and free exit data) at $(N, N)$. This is the number of paths which exit the $N \times N$ square at the top, rather than the right boundary.
\end{defn}

In our exposition, the asymptotic results will be based on the following curious identity, which can be found in~\cite[Proposition C.3]{aggarwal2021deformed}. Other closely related identities can be found in the earlier works~\cite{borodin2016stochastic_MM, BCG6V}.

\begin{theorem}
Consider the stochastic six-vertex model in the~$N \times N$ square with step initial condition and free exit data. Assume that the weights are as in Figure~\ref{fig:step_H} with fixed~$t$ and~$u$ depending on the vertex~$(i , j)$ via~$u = x_i y_j$.

Then for each $w \in \mathbb{C}$, we have
\begin{align}
\mathbb{E}\bigg( (1-w ) (1- w t) \cdots (1 - w t^{H_N-1}) \bigg)
= \frac{\prod_{i, j} (1 - x_i y_j)}{\prod_{i < j} (x_i - x_j)(y_i - y_j)}
\det\left[ \frac{1 - w - (t - w) x_i y_j}{(1-x_iy_j ) (1 - t x_i y_j)} \right]_{i, j=1}^N .\label{eqn:free_IK_det}
\end{align}
\label{thm:IKfree}
\end{theorem}

\begin{remark}
  By analytic continuation the equality above holds for any choice of $t$, $\{x_i\}_{i=1}^N$, $\{y_j\}_{j=1}^N$, where for general parameters we view the left hand side as an average with complex weights. Furthermore, Theorem~\ref{thm:IK} is a particular case of Theorem~\ref{thm:IKfree}. To see that we plug~$w = 1$ into~\eqref{eqn:free_IK_det}. The left-hand side turns into~$\text{Prob}(H_N = 0)$, which is the same as the partition function with DWBC. The right-hand side turns into the Izergin-Korepin determinant~\eqref{eqn:IZK}.
\end{remark}

\begin{proof}[Proof of Theorem \ref{thm:IKfree}]
We follow the same strategy of proof as for Theorem \ref{thm:IK}. The following four properties uniquely characterize the common value $R_N(\boldsymbol{x}; \boldsymbol{y};\, t, w)$ of both sides of \eqref{eqn:free_IK_det}:

\begin{enumerate}[(a)]
\item Symmetry in the variables~$\boldsymbol{x}=(x_1, \dots, x_N)$, and in the variables~$\boldsymbol{y}=(y_1,\dots, y_N)$.
\item The function $R_N(\boldsymbol{x}; \boldsymbol{y};\, t, w) \cdot \prod_{i, j} (1 - t x_i y_j)$ is a polynomial in~$x_1,\dots x_N$ and $y_1,\dots, y_N$, such that the degree of each variable is at most $N$.
\item If $x_N = \frac{1}{y_N}$, then
\begin{align*}
R_N(\boldsymbol{x}; \boldsymbol{y};\, t, w) &= R_{N-1}(x_1,\dots, x_{N-1}; y_1, \dots, y_{N-1};\, t, w)  .
\end{align*}

\item If $x_i = 0$ for all $i$, then $R_N(0,\dots,0; \boldsymbol{y};\, t,w) =  (1 - w)(1- w t) \cdots (1 - w t^{N-1})$.
\end{enumerate}
Properties (a)-(c) for the LHS and RHS of~\eqref{eqn:free_IK_det} are proven in a similar way to the corresponding properties in the proof of Theorem~\ref{thm:IK}. We only point out additional steps in the proof of part (a) for the LHS.
Suppose we are dragging the cross from left to right to show symmetry in $y_i$'s by repeated application of the Yang--Baxter equation, as in Figure~\ref{fig:graphical_arg}. We note that the observable under expectation in the LHS of \eqref{thm:IKfree} does not change as we drag the cross. Once the cross's first coordinate $i$ is larger than $N$, we can remove the cross: due to stochasticity of the model, the expectation does not depend on the configuration of paths in the cross and the possible weights of these paths (conditional on the two left inputs to the cross vertex) sum up to one.

%In more detail, suppose we are dragging the cross from left to right to show symmetry in $y_i$'s.  %For the terms on the LHS corresponding to a fixed choice of exit locations for paths, the observable~$H_N$ is constant. If either both of the rows with the cross have no path exiting or both have a path exiting, we may drag the cross through and then remove the cross from the right since it always has weight~$1$. This shows that the total contribution of terms with such a boundary condition is the same with~$y_i$ and~$y_{i+1}$ swapped. However, for boundary conditions such that one of the two rows with a cross has a path exiting, the cross contributes a weight to each configuration. We may account for this by summing over possible exit locations of the single path moving through the cross, which corresponds to summing over two different boundary conditions of the square at once, both with the same~$H_N$ value. This leads to a total contribution of~$1$ from the cross for each interior configuration by the stochasticity of the weights. Thus, the cross can indeed be removed.

The LHS satisfies property (d) because all of the paths are forced to go straight up if $x_i \equiv 0$, so that $H_N = N$ deterministically, from which the claim follows. Property (d) for the RHS will follow from Corollary \ref{cor:zero} and Proposition~\ref{prop:schur_form} later.

Now we argue by induction that \eqref{eqn:free_IK_det} follows from the fact that each side satisfies properties (a)-(d).

\begin{enumerate}[1.]
\item \textbf{Base Case: } In the base case $N = 1$ we just check the equality by direct calculation using the vertex weights and the definition of the observable. Both sides are equal to
$$\frac{1-w - (t-w) x_1 y_1}{1 - t x_1 y_1} = 1 - w \frac{1 - x_1 y_1}{1 - t x_1 y_1}.$$

\item \textbf{Induction Step: } By property (c) and symmetry, the value of~$R_N(\boldsymbol{x}; \boldsymbol{y} ;\, t, w) \cdot \prod_{i, j} (1 - t x_i y_j)$, viewed as a polynomial in $x_N$, is uniquely determined at $N$ points $x_N = 1/y_i$. Therefore, if $Z_N(\boldsymbol{x}; \boldsymbol{y};\, t, w)$ is another rational function satisfying the same properties, then~$Z_N - R_N$ vanishes at~$\frac{1}{y_1}, \cdots, \frac{1}{y_N}$, so that
$$\prod_{i, j} (1 - t x_i y_j) \left(Z_N(\boldsymbol{x}; \boldsymbol{y};\, t, w) - R_N(\boldsymbol{x}; \boldsymbol{y};\, t, w)\right) =\left( \cdots \right) \prod_{i=1}^N (x_N- \frac{1}{y_i})$$
where $\left( \cdots \right)$ is a polynomial independent of $x_N$, which we see by considering property (b) and the degree. However, by symmetry the same is true with $x_N$ replaced by $x_i$ for any $i = 1,\dots, N-1$. Therefore, we have
$$\prod_{i, j} (1 - t x_i y_j) \left(Z_N(\mathbf{x}; \mathbf{y};\, t, w) - R_N(\mathbf{x};\mathbf{y};\, t, w)\right) =\left( \cdots \right)' \prod_{i, j=1}^N (x_i- \frac{1}{y_j})$$
where $\left( \cdots \right)'$ must be independent of $(x_1,\dots, x_N)$ by degree considerations. Setting $x_i \equiv 0$, we see that this polynomial must be $0$ by property (d). \qedhere
\end{enumerate}
\end{proof}

\subsection{Asymptotics of the height function} \label{Section_height_function_assy}

Now we state a theorem about the fluctuations of the value~$H_N$ of the height function at~$(N, N)$ for the stochastic six-vertex model in the~$N \times N$ square, with step initial conditions, and free exit data, see also
Figure~\ref{fig:stochastic_sim} for a simulation.
\begin{theorem}[Borodin-Corwin-Gorin \cite{BCG6V}]\label{thm:S6Vfluct}
Assume that all $x_i \equiv 1$ and all $y_j \equiv u$ are equal, so that we use homogeneous stochastic weights with $0 < t < 1$, $0 < u < 1$. Then :
\begin{equation} \label{eq:Height_to_Tracy}
\lim_{N\to\infty}\frac{H_N - N \frac{1 - \sqrt{u}}{1 + \sqrt{u}}}{N^{1/3} \frac{(1-\sqrt{u})^{4/3}}{u^{1/3} (u^{-1/2}-u^{1/2})}} \stackrel{d}{=} -\xi_{GUE},
\end{equation}
where~$\xi_{\text{GUE}}$ is distributed according to~$F_2$, the Tracy-Widom GUE distribution.
\end{theorem}

\begin{remark} Theorem \ref{thm:S6Vfluct} can be extended to a similar asymptotic result for $H(N, M)$ (the height function at $(N, M)$, where $N$ does not have to be equal to $M$); only constants change, see \cite{BCG6V}.  Joint distributional convergence for two values of $M$ is also known, see \cite{dimitrov2020two}. However, general multipoint convergence over all possible pairs of $(N,M)$ has not yet been established at the time of writing this review.
\end{remark}

\begin{figure}
\centering
\includegraphics[width=0.67\linewidth, height=0.67\linewidth]{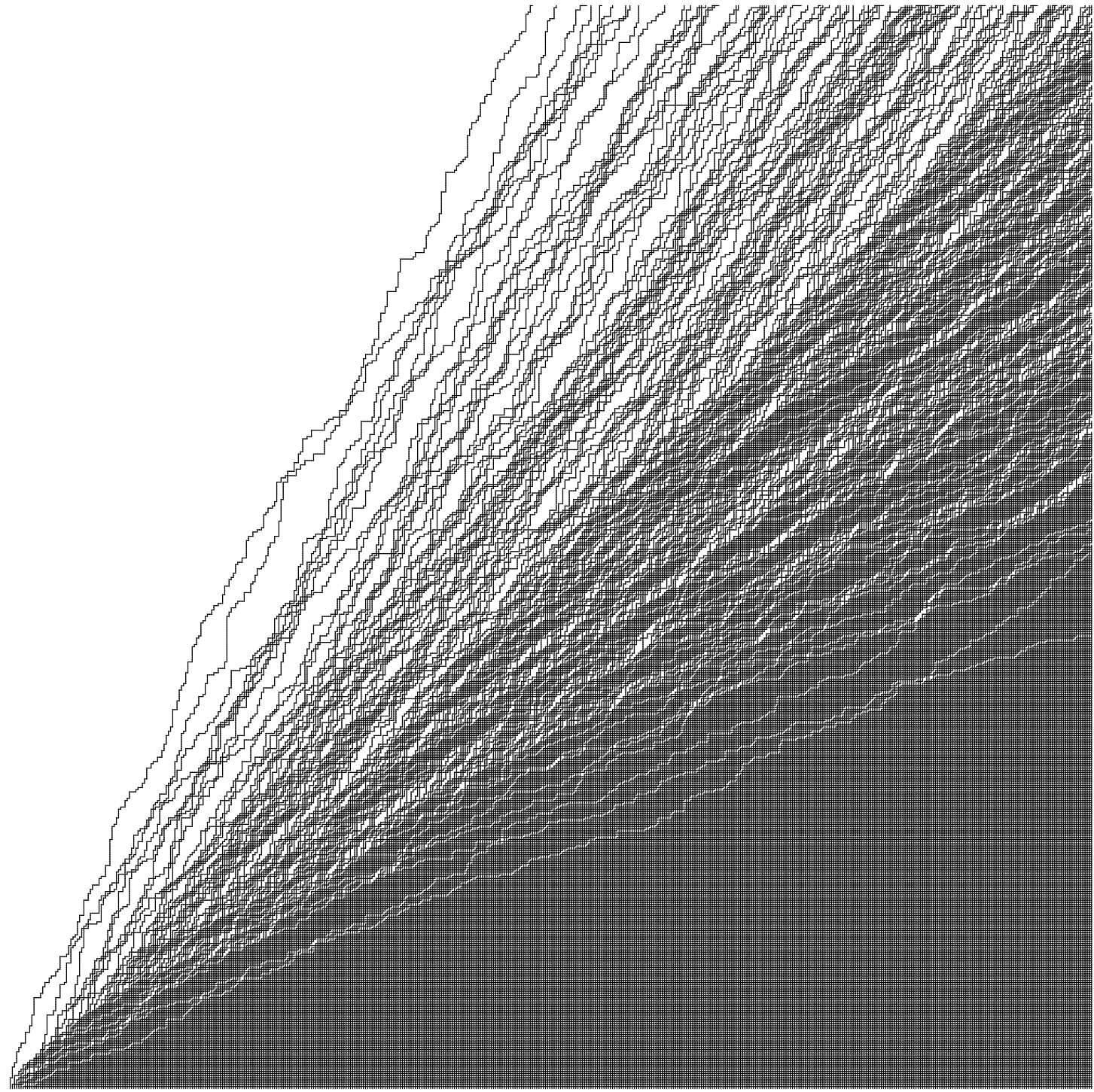}\\ \includegraphics[width=0.67\linewidth, height=0.67\linewidth]{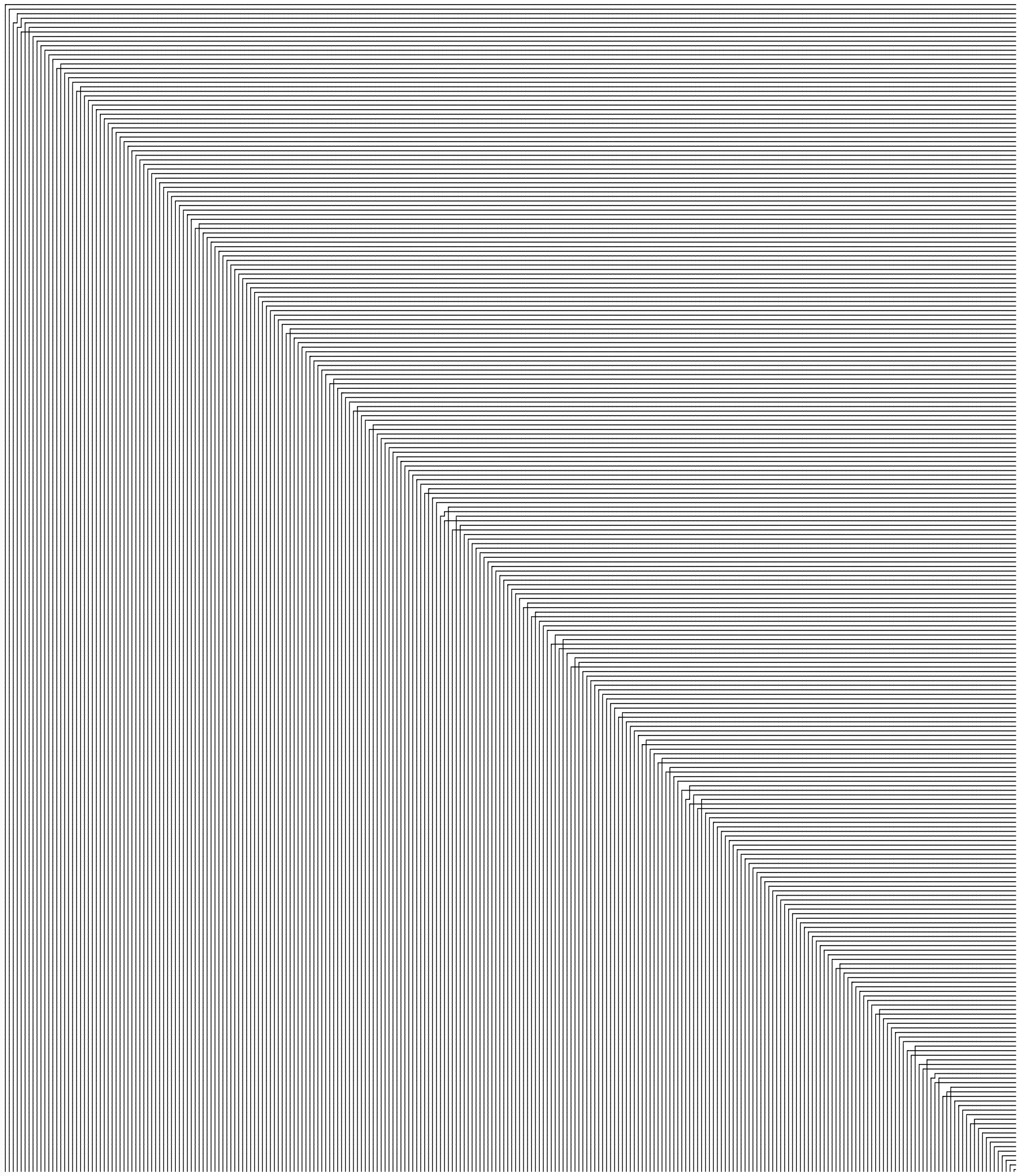}
\caption{Top: A six-vertex configuration, sampled using stochastic weights with~$b_1 = \frac{2}{3}, b_2 = \frac{1}{3}$ (obtained via $t = u = \frac{1}{2}$), with step initial condition and free exit data in the square for~$N = 500$. Bottom: $\Delta = 3/2$ ($a_1=a_2=1$, $b_1=b_2=3$, $c_1=c_2=1$), $N=256$ and domain wall boundary conditions. We thank Leonid Petrov and David Keating for these pictures.}
\label{fig:stochastic_sim}
\end{figure}

\begin{figure}
\centering
\includegraphics[width=0.63\linewidth, height=0.63\linewidth]{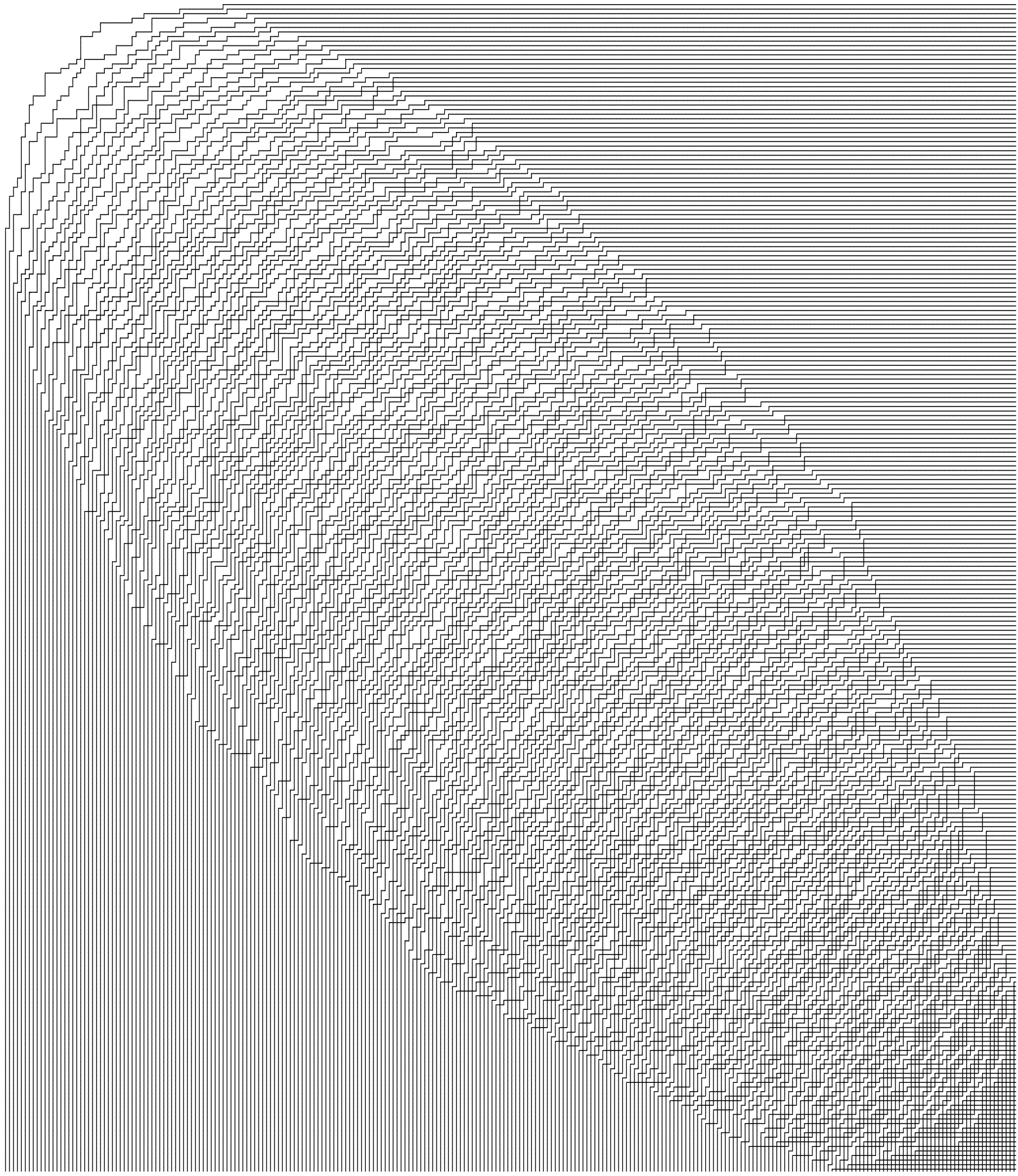} \\
\includegraphics[width=0.63\linewidth, height=0.63\linewidth]{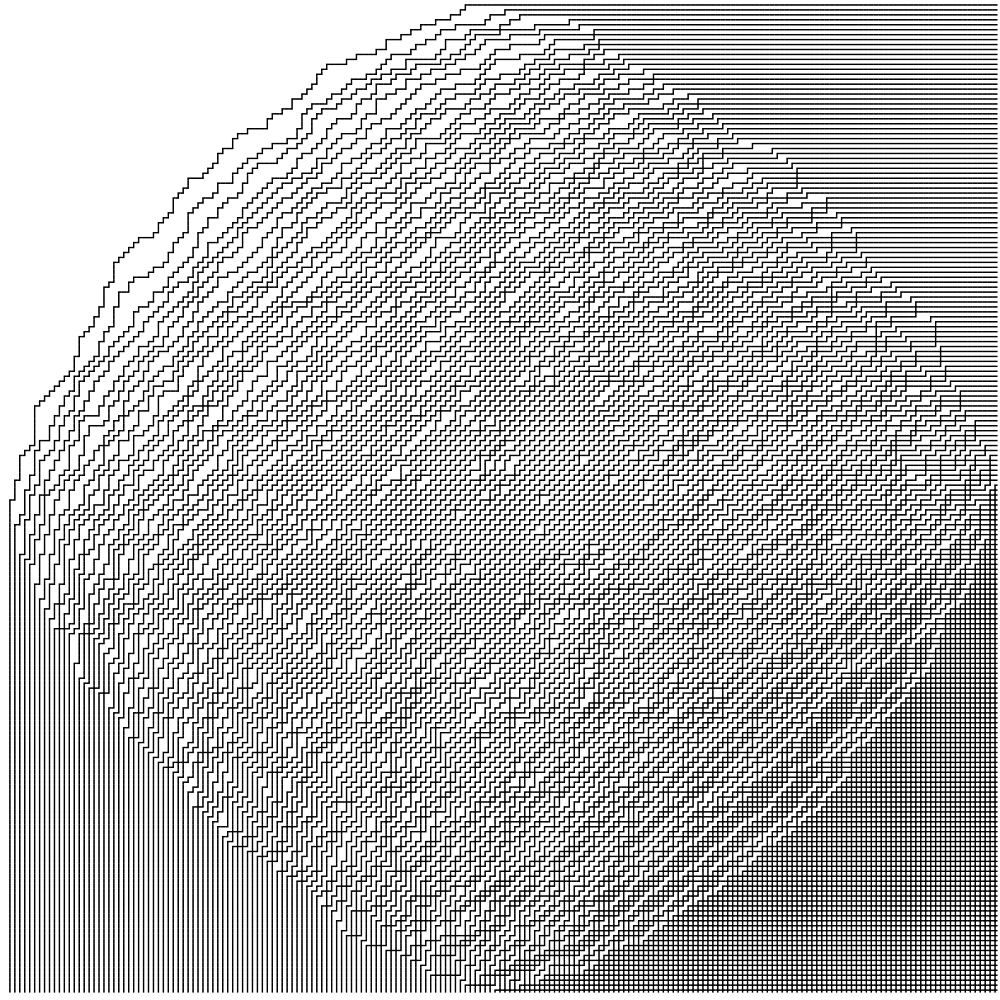}
\caption{
Top: A six-vertex configuration with DWBC, with~$\Delta = \frac{1}{4}$ ($a_1=a_2=1$, $b_1=b_2=c_1=c_2=2$) and~$N = 256$.
Bottom: $\Delta = -3$ ($a_1=a_1=b_1=b_2=1$, $c_1=c_2=\sqrt{8}$). Note that in a small region near the center of the square, the paths form a regular zig-zag pattern, with only occasional defects. This region is precisely the smooth or gaseous region. Note the difference between the gaseous region in the center and the \emph{frozen regions} near the corners, where the path configuration is completely regular and there are no defects. We thank Ananth Sridhar and David Keating for these pictures.}
\label{fig:non_stochastic_sim}
\end{figure}

We present a proof of Theorem \ref{thm:S6Vfluct} based on Theorem \ref{thm:IKfree} in the next section\footnote{The original proof in \cite{BCG6V} utilized a different set of ideas.}. Before doing so, let us put it into a wider context.

In words, \eqref{eq:Height_to_Tracy} says that the fluctuations of $H_N$ are of order $N^{1/3}$ size and have the Tracy-Widom scaling limit as $N\to\infty$. Both of these properties, namely the growth exponent of~$\frac{1}{3}$ and one point convergence to the Tracy-Widom distribution, are characteristic properties of the \emph{KPZ universality class}. For an introduction to KPZ universality, see the original reference~\cite{KPZ1986} about interface growth models in the KPZ class, and for example the survey~\cite{CorwinKPZ}. The six-vertex model with stochastic weights and free exit data can be viewed, via the Markovian sampling procedure, as a particle system on a line evolving in time or as an interface growth model (the height function along a line grows as we vary this line). From this viewpoint the model belongs to the KPZ class, as was first predicted in \cite{GwaSpohn1992}.
\smallskip

In the context of the lattice models of statistical mechanics (rather than interacting particle systems or interface growth models), the precise details of the setup of Theorem \ref{thm:S6Vfluct} turn out to be important for achieving \eqref{eq:Height_to_Tracy}. Namely, let us emphasize two features:
\begin{enumerate}
 \item Positive stochastic weights of Figure \ref{fig:step_H} imply that $\Delta=\frac{a_1 a_2 + b_1 b_2 - c_1 c_2}{2\sqrt{a_1 a_2 b_1 b_2}} > 1$.
 \item We use free exit data on top and right boundaries.
\end{enumerate}

\begin{figure}
  \centering
  \includegraphics[width=0.9\linewidth]{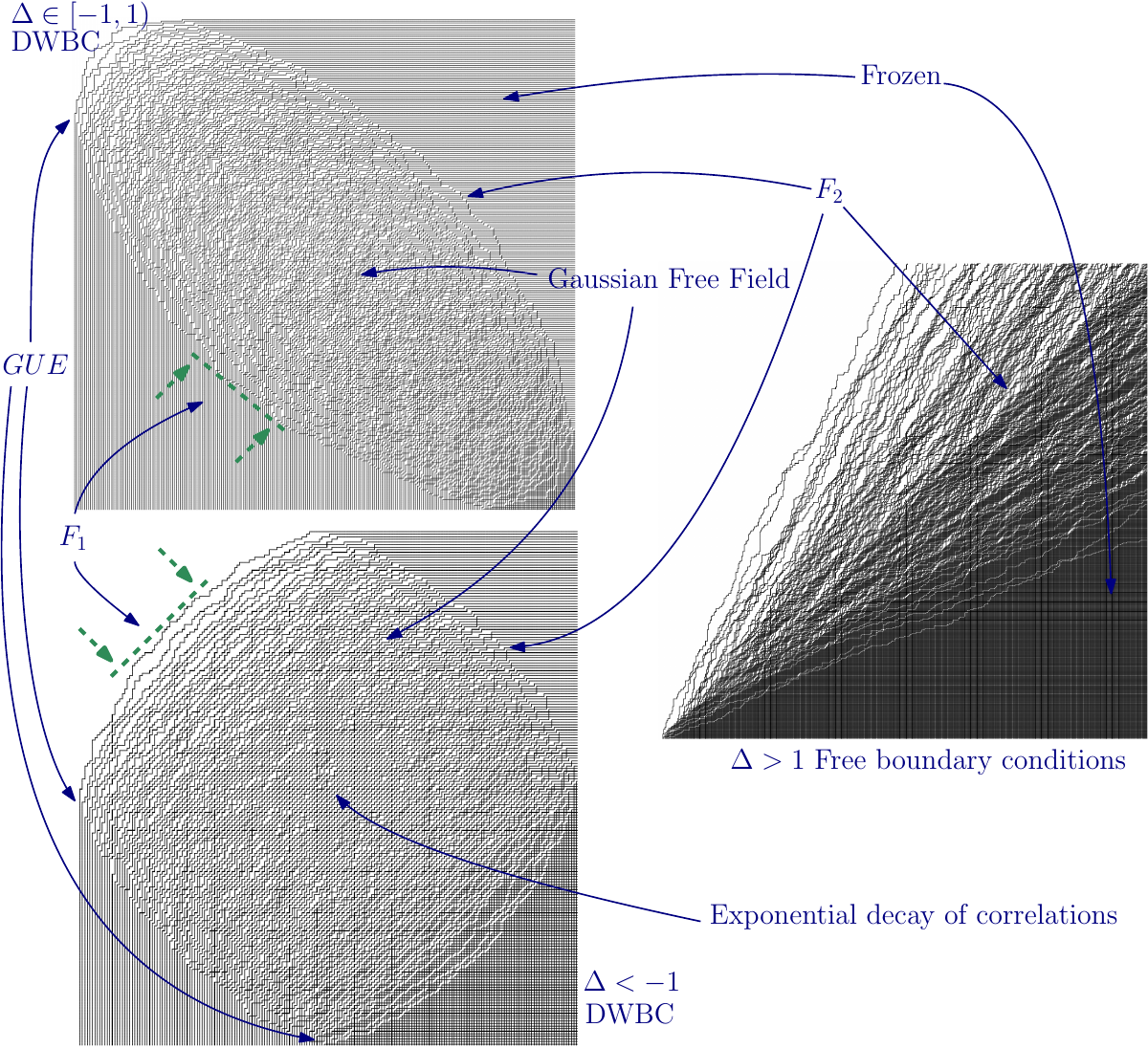}
  \caption{Various (mostly conjectural) asymptotic behaviors which can be observed
  in the six-vertex model. We have simulations for domain wall boundary conditions
  (left) with
   various choices of parameters
  and for free exit boundary conditions (right)
  with ``stochastic weights''
  (see Section \ref{subsec:IK_generalization}).
  We have attempted to illustrate the following:
  The GUE corners process is expected to
    generically describe local statistics at tangency points of the ``rough region'' to the boundary of the domain, in a way similar to Theorem \ref{thm:GUE_fluct}, and Theorems \ref{thm:JNeig} and \ref{thm:GP}.
   At the border of the rough region and frozen region away from tangency points, one should observe the Tracy--Widom distribution $F_2$, as in Theorem \ref{thm:circle}. As indicated on the right, with positive stochastic weights ($\Delta > 1$) and free exit data, $F_2$ is also observed in the bulk of the region as the limiting distribution function of the recentered and rescaled height function at a point; c.f. Theorem \ref{thm:S6Vfluct}. The green dotted line in the bottom left simulation indicates the extremal (in the northwest direction) fluctuation of the ``top path'', and the green dotted line in the top left describes a similar observable; both of these converge to $F_1$ after appropriate rescaling (c.f. Theorem \ref{thm:Ltriangle}). For $\Delta < 1$, in the rough region the correlations of height function fluctuations at the macroscopic scale should converge to those of a Gaussian free field. If $\Delta < -1$, a gaseous region may appear in the center of the domain, and the edge-edge covariances there decay exponentially in $N$.
  \label{Figure_teaser}}
  \end{figure}

It is unclear whether the height function can ever have $N^{1/3}$ fluctuations or a Tracy-Widom limit when $\Delta<1$ --- most probably, not:
\begin{itemize}
 \item For $\Delta=0$ the six-vertex model is equivalent to a random tilings model (see, e.g., \cite{ferrari2006domino}), for which the field of the fluctuations of the height function is expected to have very different asymptotic behavior leading to the Gaussian Free Field without any rescaling: see \cite[Lectures 11-12]{gorin2021lectures} for general discussion and \cite[Section 6]{chhita2015asymptotic}, \cite[Section 3.6]{bufetov2016fluctuations} for the Aztec diamond results, which are equivalent to the analysis of the six-vertex model with Domain Wall Boundary Conditions at $\Delta=0$. It is usually believed that the general $|\Delta|<1$ situation should be viewed as a small perturbation of $\Delta=0$ case with similar phenomenology. We also refer to Figure \ref{fig:non_stochastic_sim} for a simulation.
 \item For $\Delta<-1$, in addition to a region where fluctuations should be somewhat similar to $\Delta=0$ situation, a special new \emph{smooth} or~\emph{gaseous} region is expected (but not yet proven) to appear. The fluctuations there are expected to be even smaller and not related to the random matrix theory.  See Figure~\ref{fig:non_stochastic_sim} for simulations illustrating the appearance of a gaseous region for~$\Delta < -1$ and Domain Wall Boundary Conditions, and for more simulations see, e.g.,~\cite{keating2018random}.
 \item For the regimes~$\Delta < -1$ and~$-1 \leq \Delta < 1$, there are also several recent results about the height fluctuations in large domains with so-called~\emph{flat} boundary conditions; see for example \cite{duminil2020delocalization}, \cite{glazman2019transition}, \cite{glazman2023delocalisation}. The boundary conditions are called flat because with a slightly different convention for defining the height function, the boundary values have average slope~$0$. With our convention for the height function, these boundary conditions correspond to an average slope of~$\frac{1}{2}$ along each boundary segment. The results of these articles indicate that in this setting, the height fluctuations at a point are on average much smaller than the~$N^{1/3}$ of Theorem~\ref{thm:S6Vfluct}.
 In particular, in a domain of size~$N$
 they are on the order of~$\sqrt{\log N}$ if $\Delta \in [-1,1)$,
 and are $O(1)$ if $\Delta < -1$. We note, however,
 that if
 $\Delta < -1$ and we fix a boundary condition with
 constant \emph{nonzero} (or not equal to $1/2$ in our notations) average height slopes
 then it is expected that
 height fluctuations again have size~$\sqrt{\log N}$.
\end{itemize}

The free exit data along the top and left
boundary is equally important for \eqref{eq:Height_to_Tracy}.
We refer to the bottom panel of Figure \ref{fig:stochastic_sim} for
the stochastic six-vertex model with Domain Wall Boundary Conditions and $\Delta>1$: it is
clear that the fluctuations are finite and, therefore, there is
no hope for the appearance of the Tracy-Widom distribution. It might be possible
to recover $N^{1/3}$ fluctuations
by choosing more complicated
fixed deterministic boundary
conditions --- one needs to
approximate deterministically
the random exit data of the free
exit situation; yet, it is unclear
whether one can also preserve the
Tracy-Widom limit in this way. Generally
speaking, the dependence of the asymptotic
 behaviors on the boundary
 conditions seem to be richer
 in the~$\Delta>1$ situation, but
 rigorously known results are very limited.

Finally, we refer the reader to Figure \ref{Figure_teaser} for a schematic overview of the various asymptotic behaviors which can be observed in the six-vertex model.

\section{Proof of Theorem \ref{thm:S6Vfluct}}

\label{Section_Schur}

\subsection{Reduction to analysis of Schur measures}
The proof of Theorem~\ref{thm:S6Vfluct} is based on identifying the limiting fluctuations of $H_N$ with those of an observable of a random point process arising from a \emph{Schur measure}, and then using tools of integrable probability to analyze this point process. In order to develop the required machinery, we need some definitions and notation.

A \emph{partition}~$\lambda = (\lambda_1 \geq \lambda_2 \geq \lambda_3 \geq \cdots \geq 0)$ is a nonincreasing sequence of nonnegative integers (parts) such that~$\lambda_i \neq 0$ for finitely many $i$. We define the \emph{length} of a partition, denoted~$\ell(\lambda)$, as the number of nonzero parts~$\lambda_i$.

The \emph{Schur polynomials} $s_\lambda(x)$  are indexed by partitions.
\begin{defn}\label{def:schur_poly}
For a positive integer~$N$ and a partition~$\lambda$ with~$\ell(\lambda) \leq N$, the Schur polynomial~$s_\lambda(x_1,\dots, x_N)$ is defined by
$$s_\lambda(x_1,\dots, x_N) \defeq \frac{\det \left[ x_i^{\lambda_j + N - j} \right]}{\prod_{i < j} (x_i - x_j)}.$$
If $\lambda = (\lambda_1 \geq \cdots \geq \lambda_{N}  \geq \lambda_{N+1} \geq \cdots \geq \lambda_{\ell(\lambda)} > 0)$ has more than $N$ nonzero parts, then~$s_{\lambda}(x_1,\dots, x_N) = 0$.
\end{defn}
See~\cite{Macdonald1995} for a comprehensive review of symmetric function theory, including many facts about Schur polynomials.

The \emph{Schur measures} are a class of probability measures on partitions defined in terms of Schur functions.
\begin{defn}\label{def:SM}
Given~$2 N$ real numbers~$x = (x_1,\dots, x_N)$ and~$y= (y_1, \dots, y_N)$, with~$0 \leq x_i y_j < 1$ for all~$i, j = 1,\dots, N$, the corresponding Schur measure on partitions is defined by
\begin{equation}
P(\lambda) = \frac{1}{Z} s_{\lambda}(x) s_{\lambda}(y) , \qquad \ell(\lambda) \leq N  \label{eqn:Schur}
\end{equation}
where~$Z$ is a normalization constant.
\end{defn}
See the survey~\cite{BorodinGorinSPB12} and references therein for more background on Schur polynomials, Schur measures, and other related objects.

We denote expectation with respect to the Schur measure by $\mathbb{E}_{\text{Schur}}$. Before stating the proposition, we note that it follows from the \emph{Cauchy identity}, see, e.g.~\cite[Chapter I, Section 4]{Macdonald1995}, that the normalization constant for the Schur measure is given by
$$Z=\sum_{\lambda} s_{\lambda}(x) s_{\lambda}(y) = \prod_{i, j=1}^N (1- x_i y_j)^{-1}.$$
Now we state the proposition connecting Schur measures to the six-vertex model.

\begin{prop}\label{prop:schur_form}
Suppose~$0 \leq x_i y_j < 1$ for~$i, j = 1,\dots, N$. The value of~\eqref{eqn:free_IK_det} in Theorem~\ref{thm:IKfree} is also given by
\begin{multline}
\sum_{\lambda_1 \geq \cdots \geq \lambda_N \geq 0}
\prod_{i=1}^N\left(1 - w t^{\lambda_i + N - i}\right) \frac{s_{\lambda}(x_1,\dots, x_N) s_{\lambda}(y_1,\dots, y_N)}{\prod (1 - x_i y_j)^{-1}}
= \mathbb{E}\prod_{i=1}^N\left(1 - w t^{\lambda_i + N - i}\right)   \label{eqn:schur_expectation}
\end{multline}
where the expectation on the right is taken with respect to the Schur measure~\eqref{eqn:Schur}.
\end{prop}

\begin{proof}
Transforming the determinant in \eqref{eqn:free_IK_det}, we have
\begin{align*}
\det\left[ \frac{1 - w - (t-w) x_i y_j}{(1- x_i y_j)(1- t x_i y_j)} \right]_{i, j = 1}^N &=
\det \left[ \frac{1}{1- x_i y_j} - \frac{w}{1 - t x_i y_j} \right] \\
&= \det \left[\sum_{a=0}^\infty (1- w t^a) (x_i y_j)^a \right] \\
&= \sum_{0 \leq a_N < \dots < a_1} \left( \prod_{i=1}^N (1- w t^{a_i}) \right) \det[x_i^{a_j}] \det[ y_i^{a_j}] .
\end{align*}
The last step follows from the Cauchy-Binet formula for determinants of products of rectangular matrices. Now multiplying both sides by~$\frac{\prod_{i, j}(1- x_i y_j)}{\prod_{i < j} (x_i - x_j) (y_i - y_j)}$, identifying  strictly decreasing sequences~$a_1 > \cdots > a_N \geq 0$ with partitions $\lambda_1 \geq \lambda_2 \dots \geq \lambda_N \geq 0$ via $a_i = \lambda_i + N - i$, and then observing that $\det[x_i^{a_j}] $ is the numerator in the definition of the Schur polynomial, we convert the right hand side of~\eqref{eqn:free_IK_det} into~\eqref{eqn:schur_expectation}.
\end{proof}

\begin{remark}
A similar identity follows from~\cite[Lemma 3.3]{warnaar2008bisymmetric}. See also~\cite[Section 3, Chapter VI]{Macdonald1995}. Indeed,~\cite[Lemma 3.3]{warnaar2008bisymmetric} contains determinant formulas for the result of applying a certain generating function of \emph{Macdonald difference operators} at~$q =t $ (which we introduce in Section~\ref{subsec:contour_int}, see Equation~\eqref{eqn:Dq}) to the right hand side of the Cauchy identity. On the other hand, applying this generating function of operators to the left hand side of the Cauchy identity at~$q = t$ gives exactly~\eqref{eqn:schur_expectation}.
\end{remark}

Now we state a corollary which, together with the proposition above, completes the proof of Theorem \ref{thm:IKfree}.

\begin{cor}
If we plug in $x_1= \cdots = x_N =0$, only one term in~\eqref{eqn:schur_expectation} survives:
$$\lambda_1 = \cdots = \lambda_N = 0$$
and we get $\prod_{i=1}^N (1- w t^{N-i})$
\label{cor:zero}
\end{cor}
\begin{proof}
This follows from the fact that $s_{\lambda}(0)= \mathbf{1}_{\lambda = (0,\dots,0)}$, which, in turn, is implied by the observation that~$s_{\lambda}$ is a homogeneous polynomial of degree~$\lambda_1 + \cdots + \lambda_N$.
\end{proof}

The following theorem is an essential step in the proof of Theorem~\ref{thm:S6Vfluct}; it relates the asymptotic behavior of~$H_N $ to that of~$N - \ell(\lambda)$, for~$\lambda$ sampled from a particular sequence of Schur measures.

\begin{theorem}\label{thm:sufficient_cond}
Let $\lambda(N) = (\lambda_1(N) \geq \lambda_2(N) \geq \cdots \geq \lambda_N(N) \geq 0)$,~$N = 1,2,\dots$, be a sequence of random partitions and let~$\ell_N \defeq \ell(\lambda(N))$ be the length of~$\lambda(N)$. Let~$H_N$,~$N=1,2,\dots$, be a sequence of random variables taking values in non-negative integers. Fix~$0 < t < 1$ and assume that for all~$w \in \mathbb{C}$ and all~$N$ we have
\begin{align}
\mathbb{E}\bigl[ (1-w) (1- w t) \cdots (1- w t^{H_N-1}) \bigr]
=
\mathbb{E}\left[ \prod_{i=1}^N \left(1 - w t^{\lambda_i(N) + N - i}\right) \right] \label{eqn:equal_obs}.
\end{align}
If for some constants~$\alpha_N, \beta_N$ with~$\lim_{N\rightarrow \infty} \beta_N = \infty$, and for a random variable~$\xi$ with a continuous distribution function~$F_\xi(s)$, we have
\begin{equation}
\frac{N - \ell_N - \alpha_N}{\beta_N} \stackrel{d}{\rightarrow} \xi, \label{eqn:sc}
\end{equation}
then we also have
\begin{equation}
\frac{H_N - \alpha_N}{\beta_N} \stackrel{d}{\rightarrow} \xi. \label{eqn:sc2}
\end{equation}
\end{theorem}

\smallskip

This statement is from \cite{borodin2016stochastic_MM}; see Proposition 5.3, Example 5.5, and Corollary 5.11 there.

\begin{proof}[Sketch of the proof of Theorem \ref{thm:sufficient_cond}]
Consider any sequence of deterministic constants~$z_N$. For each fixed $N$, set~$w = -t^{-z_N}$ and divide both sides of~\eqref{eqn:equal_obs} by
$$\prod_{i \geq 0} \left(1 + t^{-z_N + i} \right).$$

This gives
 \begin{equation}\label{eqn:6V_Schur}
 \mathbb{E} \prod_{j \geq 0} \frac{1}{1 + t^{-z_N + H_N + j}} = \mathbb{E}
 \prod_{j \in \mathbb{Z}_{\geq 0} \setminus \{\lambda_i + N - i\}} \frac{1}{1+t^{-z_N + j}}
 \end{equation}
 where~$\lambda = \lambda(N)$.

Let us take $z_N \defeq \alpha_N + \beta_N s$ for an arbitrary $s\in\mathbb R$. The statement of the theorem would follow from \eqref{eqn:6V_Schur}, if we prove two claims:
\begin{align}
\label{eq_x2}
  \mathbb{E}
\prod_{j \in \mathbb{Z}_{\geq 0} \setminus \{\lambda_i + N - i\}} \frac{1}{1+t^{-z_N + j}} &=  \mathbb{P}\left(\frac{N - \ell_N - \alpha_N}{\beta_N} > s \right) + o(1),\qquad \text{ and}
\\
\label{eq_x3}
 \mathbb{E} \prod_{j \geq 0} \frac{1}{1 + t^{-z_N + H_N + j}} &= \mathbb P\left(\frac{H_N - \alpha_N}{\beta_N} > s \right) + o(1),
\end{align}
where $o(1)$ terms tend to $0$ as $N\to\infty$.

\smallskip

In order to prove \eqref{eq_x2} we take a large $M > 0$, to be specified later and notice that as a consequence of the distributional convergence of~$\frac{N - \ell_N - \alpha_N}{\beta_N}$ to a random variable with continuous distribution (and using $\lim_{N\to\infty}\beta_N=+\infty$), we have
\begin{equation}\label{eqn:spread}
\lim_{N\to\infty} \sup_{x \in \mathbb{R}} \mathbb{P}\bigl(- M \le  N - \ell_N - x \le  M \bigr)=0.
\end{equation}
Let $\mathcal A_>$ denote the event
$$\mathcal A_>=\bigl\{\lambda\mid N - \ell_N - \alpha_N-\beta_N s>M\bigr\}.
$$
 Noting that the smallest  point of the set $\mathbb{Z}_{\geq 0} \setminus \{\lambda_i + N - i \}$ is  $N - \ell_N$, we see that on $\mathcal A_>$ we have a two-sided bound
$$
 1\ge \prod_{j \in \mathbb{Z}_{\geq 0} \setminus \{\lambda_i + N - i\}} \frac{1}{1+t^{-z_N + j}}=\prod_{j \in \mathbb{Z}_{\geq 0} \setminus \{\lambda_i + N - i\}} \frac{1}{1+t^{-\alpha_N - \beta_N s+j }} \ge \prod_{i=1}^{\infty} \frac{1}{1+t^{i+M}},
$$
and the expression in the right-hand side is close to $1$ if $M$ is large. On the other hand, on the event
$$
\mathcal A_<=\bigl\{ \lambda\mid N - \ell_N - \alpha_N-\beta_N s<-M\bigr\},
$$
 we have another bound
$$
 0\le \prod_{j \in \mathbb{Z}_{\geq 0} \setminus \{\lambda_i + N - i\}} \frac{1}{1+t^{-z_N + j}}=\prod_{j \in \mathbb{Z}_{\geq 0} \setminus \{\lambda_i + N - i\}} \frac{1}{1+t^{-\alpha_N - \beta_N s + j}} \le \frac{1}{1+t^{-M}},
$$
and the expression in the right-hand side is close to $0$ is $M$ is large. Therefore, denoting also
$$
 \mathcal A_{\approx}=\bigl\{\lambda\mid -M\le  N - \ell_N - \alpha_N-\beta_N s\le M\bigr\},
$$
we conclude that
$$
 \left|  \mathbb{E}
\prod_{j \in \mathbb{Z}_{\geq 0} \setminus \{\lambda_i + N - i\}} \frac{1}{1+t^{-z_N + j}}- \mathbb P(\mathcal A_>)\right|\le \eps(M)+ \mathbb P(\mathcal A_{\approx}),
$$
where $\eps(M)$ tends to $0$ as $M\to\infty$ (uniformly in $N$). If $N$ is large, then $ \mathbb P(\mathcal A_>)$ approximates the right-hand side of \eqref{eq_x2} and $\mathbb P(\mathcal A_{\approx})$ is close $0$ because of \eqref{eqn:spread}. Hence, choosing first $M$ to be large, and then $N$ to be even larger, we obtain \eqref{eq_x2}.

\smallskip

The proof of \eqref{eq_x3} is very similar to the proof of \eqref{eq_x2}. The only new feature is that for proving an analogue of $\mathbb P(\mathcal A_{\approx})\to 0$ for $H_N$, we can no longer use an analogue of \eqref{eqn:spread} for $H_N$, because we have not yet proven that $H_N$ has any scaling limit. The remedy is to notice that the fact \eqref{eqn:spread} only relies on the monotonicity and continuity of the limiting distribution function which can be replaced by similar monotonicity and continuity in $t$ of both sides of \eqref{eqn:6V_Schur}, see \cite[Section 5]{borodin2016stochastic_MM} for some further details.
\end{proof}

In conclusion, we find that to prove Theorem \ref{thm:S6Vfluct} it suffices to prove the distributional convergence of~$N - \ell(\lambda)$ for~$\lambda$ sampled from the Schur measure of Definition~\ref{def:SM}. A popular approach for studying Schur measures going back to~\cite{okounkov2001infinite} produces a random point process (random subset of~$\mathbb{Z}$) from~$\lambda$, and then analyzes this point process by utilizing determinantal formulas for correlation functions, written in terms of contour integrals. In this way the distribution of~$N - \ell(\lambda)$ is expressed as a Fredholm determinant. Taking the limit, we arrive at the Fredholm determinant formula for the Tracy-Widom distribution~$F_2$. This approach is well-documented, see e.g.~\cite{BorodinGorinSPB12}, and we will not follow this path. Instead, we use the technology of \emph{Macdonald difference operators}.

\subsection{Contour integral formulas}
\label{subsec:contour_int}

The use of Macdonald difference operators to study asymptotics of random Young diagrams is a method which was first developed in~\cite{BorodinCorwin2011Macdonald}. These operators were utilized in~\cite{aggarwal2015correlation_schur} to study Schur measures. The machinery was then further developed and used to study edge asymptotics of eigenvalue distributions (and their discrete analogues) in~\cite{ahn2020airy}.

Define the operator acting on multivariate symmetric polynomials in~$x_1,x_2,\dots,x_N$:
\begin{equation}\label{eqn:Dq}
D_q \defeq \prod_{i < j} (x_i - x_j)^{-1} \left(\sum_{i=1}^N T_{q, i}\right) \prod_{i < j} (x_i - x_j) ,
\end{equation}
where~$T_{q, i}f(x_1,\dots, x_N) = f(x_1,\dots, x_{i-1}, q x_i, x_{i+1}, \dots, x_N)$, for a polynomial~$f$. This operator is a particular case of a \emph{Macdonald~$q$-difference operator} (see \cite[Chapter VI]{Macdonald1995}), with parameters~$q, t$ of Macdonald set as~$q = t$.

It follows from Definition~\ref{def:schur_poly} that
\begin{align}\label{eqn:schur_rel}
D_q s_{\lambda}(x_1,\dots, x_N) = \left(\sum_{i=1}^N q^{\lambda_i + N - i} \right) s_{\lambda}(x_1,\dots, x_N).
\end{align}
It is also helpful to use an equivalent form of~\eqref{eqn:Dq}:
\begin{equation}\label{eqn:Dq2}
D_q = \sum_{i=1}^N  \prod_{j\neq i} \frac{q x_i - x_j}{x_i - x_j} T_{q, i} .
\end{equation}
The following proposition uses this operator to extract contour integral formulas for observables of the Schur measure.

\begin{prop}\label{prop:integral_1}
Let $\lambda_1 \geq \cdots \geq \lambda_N \geq 0$ be distributed as the Schur measure
\begin{equation}
P(\lambda) = \left( \prod_{i, j =1}^N (1- x_i y_j) \right) s_{\lambda}(x_1, \dots, x_N) s_\lambda (y_1, \dots, y_N), \qquad |x_i y_j| < 1. \label{eqn:SM}
\end{equation}
Then for any distinct~$ 0 < q_1,\dots,q_k < 1$ we have
\begin{align*}
\mathbb{E}\left[ \prod_{m=1}^k \sum_{i=1}^N q_m^{\lambda_i + N - i } \right]
= & \frac{1}{(2 \pi \i)^k} \oint_{\Gamma_1} \cdots \oint_{\Gamma_k}   \prod_{1 \leq i  < j \leq k}  \frac{q_i z_i - q_j  z_j}{z_i - q_j z_j} \frac{z_i - z_j}{q_i z_i - z_j} \\
&\times  \left( \prod_{m=1}^k \prod_{j=1}^N \frac{q_m z_m - x_j}{z_m - x_j}\right) \left( \prod_{m=1}^k \prod_{j=1}^N\frac{1 - y_j z_m}{1 - q_m  y_j z_m} \right) \prod_{m=1}^k \frac{d z_m}{(q_m-1)z_m},
\end{align*}
where the contours of integration are as follows: Each $z_j$ contour~$\Gamma_j$ is a loop which contains $x_1,\dots, x_N$, does not contain $0$, and does not contain the pole at $z_i / q_j$ or $q_i z_i$ for any $i < j$ nor the pole at~$\frac{1}{ q_j  y_i}$ for any~$i=1,\dots,N$.
\end{prop}
\begin{remark}
If all~$x_i$ are equal, then we can take~$\Gamma_j$ to be small loops around the point~$x_1=\cdots=x_N$.
\end{remark}

\begin{proof}[Proof of Proposition \ref{prop:integral_1}]
The Cauchy identity, see \cite[Chapter I, Section 4]{Macdonald1995}, reads
$$\sum_{\lambda} s_{\lambda}(x_1,\dots, x_N) s_{\lambda}(y_1,\dots, y_N) = \prod_{i, j=1}^N (1 - x_i y_j)^{-1} .$$
We apply the operator~$D_q$ in the form~\eqref{eqn:Dq2} to both sides.
Using~\eqref{eqn:schur_rel}, the LHS becomes $\sum_{\lambda} \left( \sum_{i=1}^N q^{\lambda_i + N - i} \right) s_{\lambda}(x) s_{\lambda}(y)$. Transforming the RHS, we get
\begin{align}
\label{eq_x4}\sum_{\lambda} \left( \sum_{i=1}^N q^{\lambda_i + N - i} \right) s_{\lambda}(x) s_{\lambda}(y)&=\prod_{i, j=1}^N (1 - x_i y_j)^{-1} \sum_{i=1}^N \left[ \prod_{j\neq i} \frac{q x_i - x_j}{x_i - x_j} \right] \prod_{j=1}^N \frac{1 - x_i y_j}{1 - q x_i y_j} \\
&= \prod_{i, j=1}^N (1 - x_i y_j)^{-1} \frac{1}{2 \pi \i} \oint \prod_{i=1}^N \left[ \frac{q z - x_i}{z - x_i} \frac{ 1 - y_j z}{1 - q y_j z} \right] \frac{d z}{(q-1)z} \notag.
\end{align}
For the last equality we choose the contour of integration above to contain $x_1,\dots, x_N$ and exclude $0$ and $\frac{1}{q y_1}, \dots, \frac{1}{q y_N}$. Then we can see the equality above by summing up the residues at $z  = x_i$. Multiplying \eqref{eq_x4} by $\prod_{i,j}(1-x_i y_j)$, we get the result for $k = 1$.

For the~$k = 2$ result, we set $q = q_1$ in \eqref{eq_x4} and apply $D_{q_2}$ to both sides of the equation.
%\begin{multline}\label{eqn:q1}
%\sum_{\lambda}  \left( \sum_{i=1}^N q_1^{\lambda_i + N - i} \right) s_{\lambda}(x_1,\dots, x_N) s_{\lambda}(y_1,\dots, y_N) \\
% = \prod_{i, j} (1 - x_i y_j)^{-1} \frac{1}{2 \pi \i} \int \prod_{i=1}^N \left[ \frac{q_1 z - x_i}{z - x_i} \frac{ 1 - y_j z}{1 - q_1 y_j z} %\right] \frac{d z}{(q_1-1)z}.
%\end{multline}
Then the LHS again transforms according to the eigenrelation resulting in
\begin{equation}
\label{eq_x5}
\sum_{\lambda}  \left( \sum_{i=1}^N q_1^{\lambda_i + N - i} \right)  \left( \sum_{i=1}^N q_2^{\lambda_i + N - i} \right) s_{\lambda}(x) s_{\lambda}(y),
\end{equation}
and the right hand side becomes
\begin{multline}\label{eqn:2var}
 \frac{1}{2 \pi \i} \oint \left( \sum_{r=1}^N \prod_{l \neq r}  \frac{q_2 x_r - x_l}{x_r - x_l} T_{q_2, r} \left[ \prod_{i=1}^N \frac{q_1 z - x_i}{z - x_i} \prod_{j=1}^N (1 - x_i y_j)^{-1}\right] \right)
  \prod_{i=1}^N \frac{ 1 - y_i z}{1 - q_1 y_i z} \frac{d z}{(q_1-1)z}.
 \end{multline}
Note that we can write
$$\sum_{r=1}^N \prod_{l \neq r}  \frac{q_2 x_r - x_l}{x_r - x_l} T_{q_2, r} \left[ \prod_{i=1}^N \frac{q_1 z - x_i}{z - x_i} \prod_{j=1}^N (1 - x_i y_j)^{-1}\right] $$
as
\begin{multline}\label{eqn:integrand}
\left[ \prod_{i=1}^N \frac{q_1 z - x_i}{z - x_i} \prod_{j=1}^N (1 - x_i y_j)^{-1}\right] \frac{1}{2 \pi \i} \oint  \frac{q_1 z - q_2  w}{z - q_2 w} \frac{z - w}{q_1 z - w} \prod_{j=1}^N \frac{1 - w y_j}{1 - q_2 w y_j}   \\
  \times \left(\prod_{l =1}^N  \frac{q_2 w - x_l}{w - x_l} \right) \; \frac{d w}{(q_2-1)w},
\end{multline}
where the integration contour contains~$x_1,\dots,x_N$, and does not contain $0$, $\frac{1}{q_2 y_j}$ for any $j$, or $q_1 z$. We again see this by summing the residues. Plugging this back into the expression \eqref{eqn:2var} for the RHS, and then multiplying the result and the equal \eqref{eq_x5} by $\prod_{i, j} (1 - x_i y_j)$, we obtain the statement of the proposition for $k = 2$. The result for general $k$ is obtained by induction, and the general induction step proceeds in a similar fashion. \qedhere

\end{proof}

\begin{remark}
The mechanism we utilized to introduce contour integrals can be described as follows: If $f(z)$ is a analytic function, then to write $D_{q} \prod_{i=1}^N f(x_i)$ as a contour integral, one can integrate the function $\prod_{i=1}^N f(x_i) \cdot \frac{f(q z)}{f(z)}$ against $\prod_{i=1}^N \frac{q_1 z - x_i}{z - x_i} \frac{dz}{(q-1)z}$, if the contour of integration is chosen to contain $\{x_i\}$ and no other poles. More generally, if $S(x_1,\dots, x_N)$ is symmetric and analytic in the $x_i$, then in the case that $S$ admits a so-called \emph{supersymmetric lift}, one can construct contour integral formulas for $D_{q} S(x_1,\dots, x_N)$ in a similar way, see~\cite{ahn2020airy} for details.
\end{remark}

\begin{figure}
\centering
\includegraphics[width=0.8\linewidth]{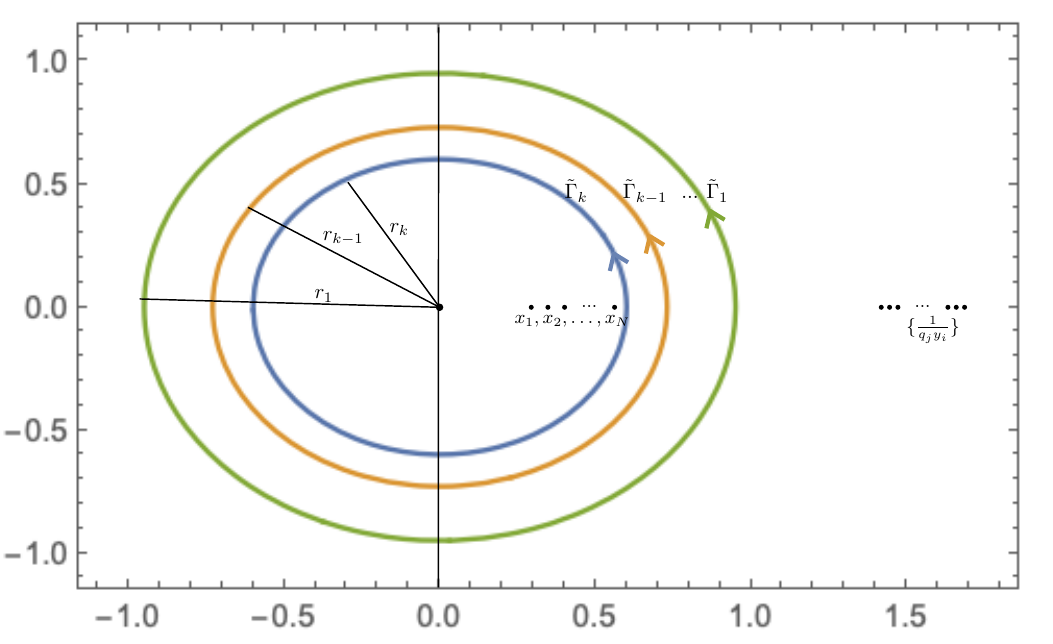}
\caption{A possible choice of contours in Corollary~\ref{cor:obs2}. Here~$\tilde{\Gamma}_j$ are circles with radii~$r_j$, which we must choose so that~$\frac{r_j}{r_i} < q_i $ for all~$i < j$ (recall in particular that each~$q_j < 1$). This can always be achieved provided that each~$q_j$ is close enough to~$1$.}
\label{fig:contours}
\end{figure}

\begin{cor}\label{cor:obs2}
Let $\lambda_1 \geq \cdots \geq \lambda_N \geq 0$ be distributed as a Schur measure

$$P(\lambda) = \prod_{i, j =1}^N (1- x_i y_j) s_{\lambda}(x_1, \dots, x_N) s_\lambda (y_1, \dots, y_N),
\qquad |x_i y_j| < 1.$$
Then for any $ 0 < q_1 , \dots,  q_k < 1$ we have

\begin{align*}
\mathbb{E}\left[\prod_{m=1}^k \sum_{j \in \mathbb{Z}_{\geq 0} \setminus \{\lambda_i + N - i \}} q_m^j \right]
= & \frac{1}{(2 \pi \i)^k} \int_{\tilde{\Gamma}_1} \cdots \int_{\tilde{\Gamma}_k}   \prod_{1 \leq i  < j \leq k}  \frac{q_i z_i - q_j  z_j}{z_i - q_j z_j} \frac{z_i - z_j}{q_i z_i - z_j} \\
&\times  \left( \prod_{m=1}^k \prod_{j=1}^N \frac{q_m z_m - x_j}{z_m - x_j}\right) \left( \prod_{m=1}^k \prod_{j=1}^N\frac{1 - y_j z_m}{1 - q_m  y_j z_m} \right) \prod_{m=1}^k \frac{d z_m}{(1-q_m)z_m}
\end{align*}
where each contour~$\tilde{\Gamma}_j$ contains~$\{0, x_1,\dots, x_N\}$ in its interior, and~$\tilde{\Gamma}_j$ does not contain any of the poles in~$z_j$ occurring at~$\frac{1}{q_j} z_i$ or at~$q_i z_i$ for~$i < j$, or at~$\frac{1}{q_j y_i}$,~$i=1,\dots, N$. See Figure~\ref{fig:contours} for an illustration of allowed contours in the case that each~$q_j$ is close enough to~$1$.
\end{cor}

\begin{proof}
We only prove the~$k = 1$ version, which says that
\begin{align*}
\mathbb{E}\left[ \sum_{j \in \mathbb{Z}_{\geq 0} \setminus \{\lambda_i + N - i \}} q^j \right]
= \frac{1}{2 \pi \i (1-q)} \oint_{\text{around } \{0, x_1,\dots, x_N\}} \prod_{j=1}^N \left( \frac{q z - x_j}{z - x_j} \frac{1 - y_j z}{1 - q  y_j z} \right) \frac{d z}{z},
\end{align*}
where the contour contains~$\{0, x_1,\dots, x_N\}$ and no other poles of the integrand. (For~$k > 1$, the argument is similar.)

We prove this identity by using the $k = 1$ version of Proposition~\ref{prop:integral_1} with both sides multiplied by $(-1)$, and deforming the $\Gamma_1$ contour through $0$ to get $\tilde{\Gamma}_1$, picking up the residue $\frac{1}{1-q} = 1 + q + q^2 + \cdots$ in the process. This residue is responsible for the (deterministic) change of the expression under the expectation.
\end{proof}

\subsection{Asymptotic analysis}

Note that the smallest point in the point configuration~$\mathbb{Z}_{\geq 0} \setminus  \{\lambda_i + N - i \}_{i=1}^N$ is~$N-\ell(\lambda)$, the observable which appeared in Theorem~\ref{thm:sufficient_cond}. In particular, we have
$$\sum_{j \in \mathbb{Z}_{\geq 0} \setminus  \{\lambda_i + N - i \}} q^j = q^{N-\ell(\lambda)} + \text{higher powers of } q.$$
Therefore, $q^{N-\ell(\lambda)}$ is the largest term in the sum, and if we scale $q \rightarrow 1$ with $N$, the asymptotic behavior in the observable of Corollary \ref{cor:obs2} encodes information about the distribution of ${N-\ell(\lambda)}$. In more detail, suppose the points in $\mathbb{Z}_{\geq 0} \setminus \{\lambda_i + N - i \}$ are $p_1 < p_2 < p_3 <\cdots$, and define random variables $a_i$ through
\begin{equation}
\label{eq_x6}
p_i = N \alpha - N^{1/3} a_i,
\end{equation}
where $\alpha$ is a deterministic constant. If we set $q = 1 - s N^{-1/3}$, then we get as $N\to\infty$
\begin{align}
\sum_{i \geq 1} q^{p_i} &= \sum_{i\geq 1} (1 - s N^{-1/3})^{N \alpha - N^{1/3} a_i }  \notag \\
&= (1 - s N^{-1/3})^{N \alpha } \sum_{i\geq 1}  e^{s a_i}  \left(1 + o(1)\right),  \label{eqn:edgescaling}
\end{align}
so that as  $\sum_{i \geq 1} q^{p_i}$ is converging $N\to\infty$ to a type of (random) Laplace transform for the limiting point process defined by $\{a_i\}_{i=1}^\infty$. Motivated by this discussion, we have the following theorem.

\begin{theorem}\label{thm:oneq_conv}
Fix a parameter~$ 0 < u < 1$ and let~$\lambda = (\lambda_1, \dots, \lambda_N)$ be a random partition distributed according to the Schur measure of Definition~\ref{def:SM} with~$x_1 = \cdots = x_N = 1$,~$y_1 = \cdots = y_N = u$. Let~$p_1 < p_2 < \cdots$ be the points of~$\mathbb{Z}_{\geq 0} \setminus \{\lambda_i + N - i \}_{i=1}^N$. Then for each~$ s  > 0$, setting~$q = 1 - s N^{-1/3}$, we have as~$N \rightarrow \infty$
$$(1 - s N^{-1/3})^{-\alpha N}  \mathbb{E} \sum_{\{p_i\}} q^{p_i} \rightarrow \frac{1}{ 2\sqrt{ \pi}} (s \sigma)^{-3/2} \exp\left( \frac{(s \sigma)^3}{12} \right),$$
where~$\displaystyle \alpha = \frac{1- \sqrt{u}}{1 + \sqrt{u}}$, $\displaystyle \sigma = u^{1/6} \frac{(1 - \sqrt{u})^{1/3}}{1 + \sqrt{u}}$.
\end{theorem}

\begin{remark}\label{rmk:AiryHint}
In the notation \eqref{eq_x6}, Theorem~\ref{thm:oneq_conv} leads to a prediction
$$\lim_{N \rightarrow \infty} \left\{\frac{a_i}{\sigma}\right\}_{i=1}^\infty \stackrel{d}{=} \{\mathfrak{a}_i\}_{i=1}^\infty,$$
where~$\mathfrak{a}_i$ satisfy
$$\mathbb{E} \sum_{i=1}^{\infty} \exp(s \mathfrak{a}_i) =  \frac{1}{ 2\sqrt{ \pi}} s^{-3/2} \exp \left( \frac{s^3}{12} \right) . $$
Theorem~\ref{thm:oneq_conv} fixes the scaling but is not yet sufficient to determine the distribution of the~$\{\mathfrak{a}_i\}_{i=1}^\infty$, from which we are mostly interested in~$\mathfrak{a}_1$ giving the desired limit of~$N - \ell(\lambda)$.
\end{remark}

\begin{proof}[Proof of Theorem \ref{thm:oneq_conv}]

 Our strategy is to follow the method of steepest decent (see, e.g., \cite{Erdelyi56, Copson65} for general discussions) in order to analyze the asymptotic behavior of the contour integral formula of Corollary~\ref{cor:obs2}, which at a high level consists of the following steps. We first write the integrand as~$\exp(N F_N(z))$, where~$F_N(z) = F_0(z)  + F_1(z) \delta + F_2(z) \delta^2 + \cdots$ for some~$\delta= \delta(N) \rightarrow 0$. Then we look for the critical point~$z_c$ of~$F_0$, and aim to deform the integration contour so that it passes through this point and goes along a level line of~$\text{Im}(F_0)$, so that the integrand does not oscillate and~$\text{Re}(F_0)$ strictly decreases as we move along the contour away from $z_c$. One can often argue that this can be done using only a few generic properties of $F_0$, instead of using its exact form; however, in our case we use the explicit formulas available to make the argument concrete. Once we have done this, we may restrict attention to a neighborhood of the critical point, as the rest of the integral is negligible. In order to obtain the exact coefficient of the leading order term, we will compute a Gaussian integral after an appropriate change of variables.

\begin{enumerate}[1.]

\item We start with the formula
\begin{align*}
\mathbb{E} \sum_{i\geq1} q^{p_i} = \frac{1}{2 \pi \i (1-q)} \int_{\gamma_{0,1}}  \left( \frac{q z - 1}{z - 1} \frac{1 - u z}{1 - q  u z} \right)^N \frac{d z}{z},
\end{align*}
where~$\gamma_{0, 1}$ above is a contour which contains the points $0, 1$ and not the point $\frac{1}{u}$; this comes from specializing $x_j \equiv 1, y_j \equiv u$ in Corollary \ref{cor:obs2}. Setting $q = 1 - s N^{-1/3}$ and Taylor expanding we see that
\begin{align*}
\log \left( \frac{q z - 1}{z - 1} \frac{1 - u z}{1 - q  u z} \right) &=
 \frac{s (1 - u) z }{(-1 + z) (-1 + u z)} N^{-1/3}  \\
&+\left(-\frac{s^2 z^2}{2 (-1 + z)^2} +
    \frac{s^2 u^2 z^2}{2 (-1 + u z)^2}\right) N^{-2/3} \\
     &+ \left(-\frac{s^3 z^3}{
  3 (-1 + z)^3} + \frac{s^3 u^3 z^3}{3 (-1 + u z)^3}\right) N^{-1} + O(N^{-4/3}),
\end{align*}
where the $O(N^{-4/3})$ error is uniform over $z$ in compact subsets of $\mathbb{C} \setminus \{1, \frac{1}{u}\}$.

We can therefore write the integrand as
\begin{align}\label{eqn:actions}
\exp\left(s N^{2/3} F_0(z) + \frac{s^2}{2} N^{1/3} F_1(z) + \frac{s^3}{3} F_2(z)+ O(N^{-1/3})\right) \frac{d z}{z},
\end{align}
where the meromorphic functions $F_0,F_1,F_2$ also implicitly depend on $u$ and are defined by
\begin{align}
F_0(z) & \defeq \frac{(1 - u) z }{(1 - z) (1 - u z)} , \label{eqn:F0}\\
F_1(z) & \defeq -\frac{z^2 }{(1 - z)^2 } + \frac{(u z)^2 }{(1 - u z)^2 } , \label{eqn:F1}\\
F_2(z) & \defeq \frac{z^3 }{(1 - z)^3 } - \frac{(u z)^3 }{(1 - u z)^3 } . \label{eqn:F2}
\end{align}

\item Now we find the critical points of~$F \defeq F_0$ and describe the deformation to the new contour~$\mathcal{C}$ of steepest descent.

We see that
$$F'(z) = \frac{ (1 - u)}{(1-z)^2 (1- u z)^2} (1 - u z^2)$$
so the solutions to
$$F'(z) = 0$$
are $z_{\pm} = \pm \frac{1}{\sqrt{u}}$. We choose to deform the contour to pass through the critical point~$z_c \defeq z_- = - \frac{1}{\sqrt{u}}$. The level lines of the imaginary part of $F$ consist of two branches passing through $z_c$; one is the real axis (as~$F$ is real analytic), and the other branch is a closed contour in $\mathbb{C}$ which passes through the point $z_c = z_- <0$ at an angle of $\pi/2$ with the real axis. As we traverse this branch clockwise, it will move through the upper half plane and intersect the real axis again at $z_+ >1$, and the piece of this branch in the lower half plane is the reflection of the piece in the upper half plane. Indeed, this description of the level curve~$\text{Im}(F(z)) =\text{Im}(F(z_c)) = 0$ becomes clear when we realize that the equation for this curve is
$$y (-1 - u^2 (x^2 + y^2) + u (1 + x^2 + y^2)) = 0 , \qquad z = x + \i y$$
which clearly has two irreducible components: $y = 0$ and a circle
$$\mathcal{C} \defeq \left\{|z| = \frac{1}{\sqrt{u}} \right\},$$
see Figure~\ref{fig:sc} for an illustration.

\begin{figure}
\centering
\includegraphics[width=0.6\linewidth]{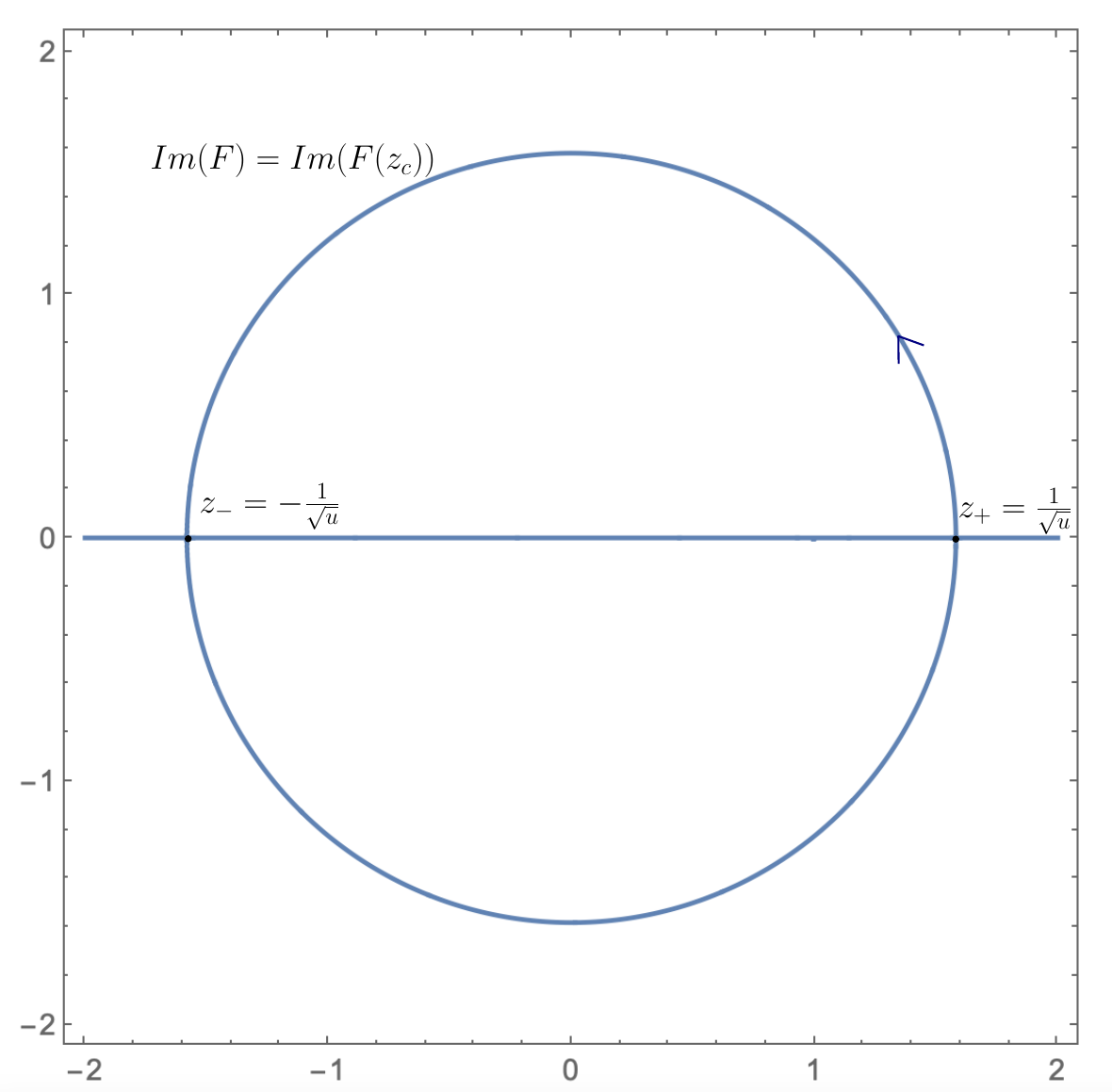}
\caption{Our steepest descent contour is the circle shown above, which is one branch of the level curve~$\text{Im}(F(z)) = \text{Im}(F(z_c))$. The other branch of the level curve is~$\{y=0\}$, which is also shown in the figure.}
\label{fig:sc}
\end{figure}

Since we have~$F''(z_c) > 0$, we know that as we move away from~$z_c$ along the circle (in either direction),~$\text{Re}(F(z))$ decreases. Indeed, near~$z_c$,~$F(z) \approx \frac{1}{2} F''(z_c) (z - z_c)^2 $, so~$F$ decreases as we move a small distance away from~$z_c$ in either direction along the circle, as~$z - z_c \approx \i \epsilon$ for small~$ \epsilon > 0$. To see that~$F$ continues to decrease until we reach~$z_+$, first recall that for an analytic function~$F$,~$\nabla \text{Re} F$ is orthogonal to~$\nabla \text{Im} F$. So~$\nabla \text{Re} F$ is always tangential to the circle, and for~$z - z_c$ small, its direction points towards~$z_c$. The only way~$\nabla \text{Re} F$ can change direction as we move along the contour is if (a) the contour passes through a point of non-analycity, or (b) at some point~$z'$,~$\nabla \text{Re} F (z') = 0$, which means~$F'(z') = 0$, i.e. we have reached another critical point. In our case, we see that the latter occurs; thus,~$z_c = z_-$ is the maximizing point for~$F$ along~$\mathcal{C}$, and~$z_+$ is a minimum.

We can clearly deform $\gamma_{0, 1}$ to $\mathcal{C}$ without passing through any poles of the integrand (note that $\frac{1}{u} > \frac{1}{\sqrt{u}}$, so we avoid the pole near $\frac{1}{u}$).

\item The contour integral will now be dominated by a small neighborhood of the critical point. Indeed, let~$\mathcal{C}_{\epsilon}$ denote the part of the contour~$\mathcal{C}$ which is contained in an~$\epsilon$ neighborhood of~$z_c= -\frac{1}{\sqrt{u}}$. If~$\epsilon$ is small enough, then for any~$z \in \mathcal{C} \setminus \mathcal{C}_{\epsilon}$ we have
\begin{align*}
\text{Re}(F(z) ) - \text{Re}(F(z_c)) < - \frac{1}{4} F''(z_c) \epsilon^2 .
\end{align*}
Indeed, this follows from the fact that $\text{Re} F$ decreases as we move from~$z_-$ to~$z_+$ (see the discussion in part 2. above), together with the fact that at the endpoints of~$\mathcal{C}_{\epsilon}$,
$$\text{Re}(F(z)  -F(z_c)) \approx \frac{1}{2} F''(z_c) (\epsilon \i)^2  = -\frac{1}{2} F''(z_c) \epsilon^2$$
where~$\approx$ denotes equality up to an order~$\epsilon^3$ error.

Taking~$\epsilon = N^{-\delta}$ for very small fixed~$\delta>0$, this implies that after we pull out a factor of $\exp(s N^{2/3} F(z_c))$ from the integrand, the integrand is dominated by
$$\exp\left(-\frac{1}{4} s F''(z_c) N^{2/3 -2 \delta} + O(N^{1/3})\right)$$
for all $z \in \mathcal C \setminus \mathcal C_{\epsilon}$, which decays fast as $N \rightarrow \infty$ (recall $s > 0$).

Now we make the change of integration variable~$z = -\frac{1}{\sqrt{u}} + \zeta N^{-1/3}$. By throwing away the part of the integral outside of~$\mathcal C_{\epsilon}$ and Taylor expanding each term in \eqref{eqn:actions} and noting that error terms are uniformly small, we then have
 \begin{align}\label{eqn:int}
\E \sum_{i=1}^\infty q^{p_i} &= \exp\left(s N^{2/3} F(z_c) +\frac{s^2}{2} N^{1/3} F_1(z_c) + \frac{s^3}{3} F_2(z_c) \right) \cdot \exp \left(- \frac{1}{4} s^3 \sigma^3  \right) \\
& \times \frac{-N^{-1/3}}{2 \pi \i ( s N^{-1/3})} \int_{- \i \epsilon N^{1/3}}^{\i \epsilon N^{1/3}} \exp\left( s u \sigma^3 \left(\zeta + s \frac{1}{
    2 \sqrt{u}}\right)^2 \right)   \frac{d \zeta}{-\frac{1}{\sqrt{u}} + \zeta N^{-1/3}}  \cdot \bigg(1 + o(1)\bigg). \notag
    \end{align}
    We claim that as $N \rightarrow \infty$, the expression above converges to
    \begin{align*}
   (1 - s N^{-1/3})^{\alpha N} \frac{1}{ 2\sqrt{ \pi}}(s \sigma)^{-3/2} \exp \left( \frac{(s \sigma)^3}{12}\right)
 \end{align*}
 with~$\alpha = \frac{1-\sqrt{u}}{1+\sqrt{u}}$ and~$\sigma = u^{1/6} \frac{(1 - \sqrt{u})^{1/3}}{1 + \sqrt{u}}$ as in the statement on the theorem.

First, we consider the prefactors in the first line of \eqref{eqn:int} and show that they are asymptotically given by
 \begin{align*}
&\exp\left(\alpha N \log(1 - s N^{-1/3}) + o(1)\right) \exp\left( \frac{(s \sigma)^3}{12}\right) \\
&= \exp\left(-\alpha s  N^{2/3} - \alpha \frac{s^2}{2} N^{1/3} - \alpha \frac{s^3}{3} + \frac{(s \sigma)^3}{12} + o(1)\right).
\end{align*}
Indeed, one can directly compute that $F_0(z_c) = F_1(z_c) = -\alpha$. This gives agreement of the logarithm of the prefactor in \eqref{eqn:int} with the $N^{2/3}$ and $N^{1/3}$ term in the exponential above. Then, we can also check directly that
$$\frac{s^3}{3} F_2(z_c) -\frac{1}{4} s^3 \sigma^3 = -\alpha \frac{s^3}{3} + \frac{(s \sigma)^3}{12}.$$

  Next we consider the second line of \eqref{eqn:int} and note that it simplifies to
$$\frac{\sqrt{u}}{2 \pi \i s} \int_{- \i \epsilon N^{1/3}}^{\i \epsilon N^{1/3}} \exp\left( s u \sigma^3 \left(\zeta + s \frac{1}{
    2 \sqrt{u}}\right)^2 \right)
  d \zeta  \cdot \bigg(1 + o(1)\bigg). $$
  Now we can extend the integration contour to~$\pm \i \infty$ at the cost of a small error due to the exponential decay of the integrand.  Computation of the Gaussian integral gives the remaining factor
$$\frac{1}{2 \sqrt{\pi}} (s \sigma)^{-3/2}. \qedhere $$
\end{enumerate}
\end{proof}

Theorem \ref{thm:oneq_conv} has a multi-parameter generalization, which (c.f.\ Remark~\ref{rmk:AiryHint}), leads to the distributional convergence of the entire sequence~$\{a_i\}_{i=1}^\infty$ as $N\to\infty$. This convergence will allow us to extract the distributional limit of the first point~$a_1$, which describes the asymptotic behavior of~$N - \ell(\lambda)$, and thus also (due to Theorem~\ref{thm:sufficient_cond}) that of the height function value~$H_N$.

The limiting object is the \emph{Airy point process}  $\{\mathfrak{a}_i\}_{i=1}^\infty$, which also arises in random matrix theory. The Airy point process is a~\emph{random point process} which has a number of possible definitions. Postponing more detailed discussions till Appendix~\ref{app:pt_proc}, we record three equivalent definitions for $\{\mathfrak{a}_i\}_{i=1}^\infty$:
\begin{enumerate}[1.]

\item This is the limiting point process of the eigenvalues near the edge of the spectrum of a large GUE random matrix. In particular, the law of $\mathfrak a_1$ is the Tracy-Widom GUE distribution.
\item This is a determinantal point process on~$\mathbb{R}$ with correlation kernel given by
     $$K_{\Airy}(x, y) = \int_0^\infty \Ai(x + a) \Ai(y +a ) da,$$
where $\Ai(x)$ is the \emph{Airy function}.
\item This is a point process satisfying for all $k=1,2,\dots$ and $s_1,s_2,\dots,s_k>0$:
\begin{multline} \label{eq_Airy_Laplace}
 \mathbb{E}\left[\prod_{m=1}^k \sum_i \exp(s_m \mathfrak{a}_i)\right]   = \frac{e^{\sum_{i=1}^k s_i^3/12}}{(2 \pi \i)^k} \int_{\i \; \mathbb{R} + v_1} \frac{d z_1}{s_1}\cdots \int_{\i \; \mathbb{R} + v_k} \frac{d z_k}{s_k}  \\
 \times e^{\sum_{i=1}^k s_i z_i^2}  \prod_{1 \leq i < j \leq k} \frac{z_j - \frac{s_j}{2} - z_i + \frac{s_i}{2}}{z_j + \frac{s_j}{2} - z_i + \frac{s_i}{2}} \frac{z_j + \frac{s_j}{2} - z_i - \frac{s_i}{2}}{z_j - \frac{s_j}{2} - z_i - \frac{s_i}{2}},
 \end{multline}
 where and~$v_1 < \cdots < v_k$ are chosen so that
\begin{equation}\label{eqn:v_conds}
v_j - \frac{s_j}{2} > v_i + \frac{s_i}{2}, \qquad 1 \leq i < j \leq k,
\end{equation}
and integration contours are oriented from~$-\infty$ to~$\infty$.
\end{enumerate}

\begin{theorem}\label{thm:corr_convergence}
Define the random point configuration~$\{p_i\} = \mathbb{Z}_{\geq 0} \setminus \{\lambda_i + N -i\}$, where~$\lambda$ is sampled from the Schur measure of Definition~\eqref{def:SM} with~$x_i \equiv 1$ and~$y_i \equiv u$. Then, defining ~$\{a_i\}$ by~$p_i = N \alpha - N^{1/3} a_i$, we have convergence in finite-dimensional distributions:
\begin{equation}
\label{eq_x8}
 \lim_{N\to\infty} \left\{\frac{a_i}{\sigma}\right\} \stackrel{d}{=}  \{\mathfrak{a}_i\},
\end{equation}
where~$\{\mathfrak{a}_i\}$ is the Airy point process, and $\alpha = \frac{1- \sqrt{u}}{1 + \sqrt{u}}$, $\sigma = u^{1/6} \frac{(1 - \sqrt{u})^{1/3}}{1 + \sqrt{u}}$.
\end{theorem}

\begin{proof}
Using \eqref{eq_Airy_Laplace} and postponing details on the topology of convergence till Proposition \ref{prop:sufficient_conv} in Appendix~\ref{app:pt_proc}, it suffices to prove that for any $k \geq 1$ and $s_1, \dots, s_k > 0$, we have
\begin{multline}
  \lim_{N\to\infty} \mathbb{E}\left( \prod_{m=1}^k (1 - s_m N^{-1/3})^{-\alpha N} \sum_{\{p_i\}} q_m^{p_i} \right) =
\frac{e^{\sum_{i=1}^k s_i^3/12}}{(2 \pi \i)^k} \int_{\i \; \mathbb{R} + v_1} \frac{d z_1}{s_1}\cdots \int_{\i \; \mathbb{R} + v_k} \frac{d z_k}{s_k}  \\
 \times e^{\sum_{i=1}^k s_i z_i^2}  \prod_{1 \leq i < j \leq k} \frac{z_j - \frac{s_j}{2} - z_i + \frac{s_i}{2}}{z_j + \frac{s_j}{2} - z_i + \frac{s_i}{2}} \frac{z_j + \frac{s_j}{2} - z_i - \frac{s_i}{2}}{z_j - \frac{s_j}{2} - z_i - \frac{s_i}{2}}, \label{eqn:res}
\end{multline}
where the contours are as in \eqref{eq_Airy_Laplace}.
Analogously to the proof of Theorem~\ref{thm:oneq_conv}, we start with the formula
\begin{align}
\mathbb{E}\left[ \prod_{m=1}^k \sum_{\{p_i\}} q_m^{p_i} \right] &= \frac{1}{(2 \pi \i)^k} \int_{C_1} \frac{d z_1}{(1-q_1)z_1} \cdots \int_{C_k} \frac{d z_k}{(1-q_k)z_k}  \notag  \\
&\times   \left( \prod_{m=1}^k \left(\frac{q_m z_m - 1}{z_m - 1} \right)^N\right) \left( \prod_{m=1}^k \left(\frac{1 - u z_m}{1 - q_m u z_m} \right)^N \right) \prod_{1 \leq i  < j \leq k}  \frac{q_i z_i - q_j  z_j}{z_i - q_j z_j} \frac{z_i - z_j}{q_i z_i - z_j} \label{eqn:int_secondline}
\end{align}
from specializing Corollary~\ref{cor:obs2}, and we perform a steepest descent analysis to take asymptotics. Note that this expression is valid for large enough~$N$ as long as we take~$C_1 ,\dots, C_k$ to be concentric circles centered at~$0$ with strictly decreasing radii~$\frac{1}{u} > r_1 > \cdots > r_k > 1$. The arguments involving deformation of contours and concentration of the integral around a neighborhood of the critical point are similar but with more details, so we will be brief, and mostly focus on the computations leading to the result. See also~\cite[Section 2.3]{ahn2020airy} for a very similar argument, with all error terms carefully accounted for.

Exactly as in step 1 of the proof of Theorem~\ref{thm:oneq_conv}, we write
$$\left(\frac{q_m z_m - 1}{z_m - 1} \frac{1 - u z_m}{1 - q_m u z_m}  \right)^N$$
as
\begin{align}\label{eqn:actions_multi}
\exp\left(s_m N^{2/3} F_0(z_m) + \frac{s_m^2}{2} N^{1/3} F_1(z_m) + \frac{s_m^3}{3} F_2(z_m)+ O(N^{-1/3})\right)
\end{align}
where~$F_0, F_1, F_2$ are defined in~\eqref{eqn:F0}-\eqref{eqn:F2}. Recall that the critical points of~$F_0$ occur at~$z_{\pm} = \pm \frac{1}{\sqrt{u}}$.

Now we deform the contours (without them crossing each other) to~$\tilde{\Gamma}_1 ,\dots, \tilde{\Gamma}_k$ so that our contours are circles of radii all very close to~$\frac{1}{\sqrt{u}}$. More precisely, we choose circles of radius~$r_m = \frac{1}{\sqrt{u}} + t_m N^{-1/3}$, where we demand that~$t_1 > \cdots > t_m$ satisfy
$$t_j < t_i - \frac{s_i}{\sqrt{u}}, \qquad 1 \leq i < j \leq k .$$
Indeed, with this condition one can easily verify that~$q_i z_i$ stays outside of the~$z_j$ contour for~$i < j$.

Along the contours~$\tilde{\Gamma}_m$, we can see by a similar argument to the one we gave before that for each integral, the only contribution comes from the region when each variable is in a small $N^{-\delta}$-neighborhood of the critical point~$z_c = -\frac{1}{\sqrt{u}}$ of~$F_0$.

Furthermore, the contours are well-approximated by vertical lines near~$z_c$, and we make the change of integration variable
$$z_m = z_c + \zeta_m N^{-1/3}.$$
Then, we replace~$F_i(z_m)$ by its second order Taylor expansion about the critical point~$z_c$ for each~$i=0,1,2$ and~$m=1,\dots, k$, and extend the contours to~$\pm \infty$.  Then the~$\zeta$ contours become
$$\Re(\zeta_m) = -t_m$$
and one can check that the expression~\eqref{eqn:int_secondline} in the integrand becomes
\begin{align}
\left( \prod_{m=1}^k(1-s_m N^{-1/3})^{\alpha N} \right) &
\exp\left( \sum_{m=1}^k s_m u \sigma^3 \left(\zeta_m + s_m \frac{1}{ 2 \sqrt{u}}\right)^2 \right) \label{eqn:int1} \\
&\times \prod_{1 \leq i < j \leq k} \frac{\zeta_i - \zeta_j}{s_i/\sqrt{u} + \zeta_i - \zeta_j} \frac{s_i/\sqrt{u} - s_j/\sqrt{u} + \zeta_i -
   \zeta_j}{-s_j/\sqrt{u} + \zeta_i - \zeta_j} \bigg(1 + o(1)\bigg). \notag
\end{align}

Further substituting~$\zeta_m \rightarrow \frac{1}{\sqrt{u}} \zeta_m - s_m \frac{1}{ 2 \sqrt{u}}$, our contours become~$\Re(\zeta_m) = v_m$ (from~$+\infty$ to~$-\infty$, which is opposite to the orientation of the contours in~\eqref{eqn:res}) where~$v_m$ satisfy
$$v_j - \frac{s_j}{2} > v_i + \frac{s_i}{2}, \qquad 1 \leq i < j \leq k,$$
and the display~\eqref{eqn:int1} becomes
\begin{align}
\left( \prod_{m=1}^k(1-s_m N^{-1/3})^{\alpha N} \right) &
\exp\left(  \sum_{m=1}^k s_m  \sigma^3 \zeta_m^2 \right) \label{eqn:int2} \\
&\times \prod_{1 \leq i < j \leq k} \frac{\zeta_i - \zeta_j - s_i/2 +s_j/2}{s_i/2 + s_j/2 + \zeta_i - \zeta_j} \cdot  \frac{s_i/2 - s_j/2 + \zeta_i -
   \zeta_j}{-s_j/2 - s_i/2 + \zeta_i - \zeta_j}. \notag
\end{align}
Now it remains to keep track of factors from the change of integration variable, minus signs from the fact that our vertical integration contours have the wrong orientation, and incorporate the pre-factors coming from the evaluation of~$F_0,F_1,F_2$ at~$z_c$ (c.f. the expression~\ref{eqn:int} and the discussion that follows in the proof of Theorem~\ref{thm:oneq_conv}). Once we do this, we obtain the result stated in~\eqref{eqn:res}.  \qedhere

\end{proof}

At this point we have collected all the ingredients for Theorem \ref{thm:S6Vfluct}.

\begin{proof}[Proof of Theorem~\ref{thm:S6Vfluct}]
By Theorem \ref{thm:corr_convergence} we have distributional convergence
\begin{equation}
\label{eq_x7}
 \lim_{N\to\infty} \frac{N - \ell(\lambda) - \alpha N}{\sigma N^{1/3}} \stackrel{d}= - \mathfrak{a}_1,
\end{equation}
where $\mathfrak a_1$ is the first point of the Airy point process. The distribution of $\mathfrak a_1$ is precisely the Tracy-Widom GUE distribution and matches the distribution of $\xi_{GUE}$ in \eqref{eq:Height_to_Tracy}, see Definition~\ref{def:TWd} in Appendix \ref{app:pt_proc}.

Applying Proposition \ref{prop:schur_form} and Theorem \ref{thm:sufficient_cond} with $\alpha_N=\alpha N$, $\beta_N=\sigma N^{1/3}$, \eqref{eq_x7} implies that
$$
 \lim_{N\to\infty} \frac{H_N - \alpha N}{\sigma N^{1/3}}
= \lim_{N\to\infty} \frac{H_N - N \frac{1 - \sqrt{u}}{1 + \sqrt{u}}}{N^{1/3} \frac{(1-\sqrt{u})^{4/3}}{u^{1/3} (u^{-1/2}-u^{1/2})}}
\stackrel{d}= - \mathfrak{a}_1.\qedhere
$$
\end{proof}

\appendix

\section{Airy point process and Tracy-Widom distribution}\label{app:pt_proc}

In this appendix we review the definition and some properties of the Airy point process (at~${\beta=2}$) and the Tracy-Widom distribution. Our exposition is based on the theory of determinantal point processes and we refer to~\cite{BorodinGorinSPB12} and references therein for more detailed discussions. Here we first introduce the basics of determinantal point processes, then specialize to the Airy case, and finally discuss the contour integral formulas for the Laplace transforms, which already appeared in our proofs in the previous section.

Consider a configuration space~$\mathcal{X}$, such as~$\mathbb{Z}$ or~$\R$. A point configuration~$X$ in~$\mathcal{X}$ is a subset of~$\mathcal{X}$ without accumulation points.

We equip the set of point configurations with a $\sigma$-algebra generated by functions~$N_A = N_A(X)$, which give the number of points of~$X$ in a compact subset~$A \subset \mathcal{X}$.
\begin{defn}
A \emph{random point process} is a probability measure on the set of point configurations equipped with the~$\sigma$-algebra described above.
\end{defn}

\emph{Correlation functions} can be used to fully describe a random point process (under mild technical growth conditions, which hold in our case, see \cite{Soshnikov2000}). An $n$-point correlation function $\rho_n(x_1,\dots, x_n)$ is a density with respect to a reference measure~$\mu$ on~$\mathcal{X}$, which is defined by the property that for all compactly supported bounded measurable functions $f$ on~$\mathcal{X}^n$, we have
\begin{align}\label{eqn:cor_def}
\int_{\mathcal{X}^n} f \rho_n \mu^{(\otimes n)}(dx)
= \mathbb{E}\bigg[ \sum_{x_{i_1},\dots, x_{i_n} \in \mathcal{X}} f(x_{i_1}, \dots, i_{n}) \bigg]
\end{align}
where the sum is $n$--tuples of distinct points $(x_{i_1}, \dots, x_{i_n})$ from $X$.

\begin{remark} Taking $f(x_1,\dots, x_n) = \prod_{i=1}^n \mathbf{1}_{A}(x_i)$ to be the product of indicator functions for a compact subset~$A \subset \mathcal{X}$, one may deduce from \eqref{eqn:cor_def} that
\begin{align*}
\mathbb{E}\bigl[ N_A (N_A - 1) \cdots (N_A - n + 1) \bigr] = \int_{A^n}   \rho_n(x_1,\dots, x_n) \mu^{(\otimes n)}(dx).
\end{align*}
We leave the details as an exercise for the interested reader.
\end{remark}

For~$\mu$ absolutely continuous with respect to Lebesgue measure, if we have distinct~$x_1,\dots, x_n$ then one can find (by using appropriate~$f$ in~\eqref{eqn:cor_def}) that
$$\rho(x_1,\dots, x_n) \mu([x_1, x_1+ dx_1]) \cdots \mu([x_n, x_n+ dx_n]) $$
is approximately the probability of finding a particle in each small interval
${[x_1,x_1+dx_1]}, \dots, {[x_n,x_n+dx_n]}$.

Suppose we have a random point process on~$\mathbb{R}$ with correlation functions~$\rho_k$, with~$\mu$ as the reference measure. Suppose in addition that with probability~$1$, in the random point configuration~$X$ there are only finitely many points in any interval~$(s, \infty)$. Then, enumerating the random set of points in~$X$ as~$x_1 > x_2 > x_3 > \cdots $, by the inclusion-exclusion principle and~\eqref{eqn:cor_def}, we have a formula for the distribution function of the rightmost point~$x_1$:
\begin{align}
\mathbb{P}(x_1 \leq s)=\mathbb{P}\left( N_{(s, \infty)}(X) = 0 \right) &=
1 - \sum_{k=1}^\infty (-1)^{k+1} \mathbb{E}\left[\sum_{i_1 < \cdots < i_k} \prod_{j=1}^k \mathbf{1}\{x_{i_j} \in (s, \infty)\} \right]  \notag\\
&= \sum_{k=0}^\infty \frac{(-1)^k}{k!} \int_{(s, \infty)} \cdots \int_{(s, \infty)} \rho_k(x_1,\dots, x_k) \prod_{i=1}^k \mu(d x_i) .
\label{eqn:last_cdf}
\end{align}

 Next, we describe an important type of random point process, whose correlation functions have a special form.

\begin{defn}
A \emph{determinantal point process} is a random point process whose correlation functions~$\rho_n$ are determinants of a fixed kernel: There exists a kernel~$K : \mathcal{X} \times \mathcal{X} \rightarrow \mathbb{R}$ such that
$$\rho_n(x_1,\dots, x_n) =\det\left( K(x_i, x_j \right)_{i, j=1}^n.$$
\end{defn}

The \emph{Airy point process}, which we define next, is an example of a determinantal point process.

\begin{defn}
The \emph{Airy function}~$\Ai(x)$ is defined by
$$\Ai(x) \defeq \frac{1}{\pi} \int_0^\infty \cos\left(\frac{t^3}{3} + x t \right) dt.$$
\end{defn}

The Airy function first was introduced by George Biddell Airy in the study of light near caustics~\cite{AiryONTI}. Up to a normalization, the Airy function is the solution of the \emph{Airy equation}
$$y'' - x y = 0,$$
subject to the condition that $y \rightarrow 0$ as $x \rightarrow \infty$.

\begin{defn}
The \emph{Airy kernel} is defined by
$$K_{\Airy}(x, y) = \int_0^\infty \Ai(x + a) \Ai(y +a ) da.$$
\end{defn}
\begin{remark} It is a nice exercise to show that the Airy kernel can also be written as
$$K_{\Airy}(x, y)=\frac{\Ai(x) \Ai'(y ) - \Ai(y) \Ai'(x )}{x - y}.  $$
\end{remark}
\begin{defn}
\label{def:Airy_PP}
The \emph{Airy point process} is the determinantal point process on~$\mathbb{R}$ whose $k$-point correlation function~$\rho_k^{\Airy}(x_1,\dots, x_k)$, with Lebesgue measure as the reference measure, is given by
$$\rho_k^{\Airy}(x_1,\dots, x_k) = \det\left(K_{\Airy}(x_i, x_j) \right)_{i, j=1}^k .$$
\end{defn}

From the standard estimates on the decay of the Airy function~$\Ai(x)$ as~$x \rightarrow + \infty$, see e.g.~\cite[Subsection 3.7.3]{AndersonGuionnetZeitouniBook}, one may deduce that with probability~$1$ there are only finitely many points in any interval~$(s, \infty)$. Therefore, it makes sense to speak of the rightmost, or largest, particle in the Airy point process. The next definition gives a name to the distribution of the largest particle.

\begin{defn}\label{def:TWd}
The \emph{Tracy-Widom GUE} distribution function is defined by
$$F_2(s) \defeq \sum_{k=0}^\infty \frac{(-1)^k}{k!} \int_{(s, \infty)} \cdots \int_{(s, \infty)} \rho_k^{\Airy}(x_1,\dots, x_k) \prod_{i=1}^k d x_i .$$
Comparing with~\eqref{eqn:last_cdf},~$F_2$ is the distribution function of the rightmost particle in the Airy point process.
\end{defn}
\begin{remark}
One can prove that for any~$s \in \mathbb{R}$, this series converges absolutely. Furthermore, plugging in the appropriate determinants for the correlation functions, the expression defining~$F_2(s)$ above becomes the Fredholm determinant expansion for
$$\det(\mathbf{1} - K_{\Airy})_{L_2(s, \infty)}.$$
\end{remark}

Definitions \ref{def:Airy_PP} and \ref{def:TWd} are standard ways to construct the ($\beta=2$) Airy point process and Tracy-Widom distribution, adopted by many random-matrix textbooks, see e.g., \cite[Section 24.2]{MehtaRMT}, \cite[Section 3.1.1]{AndersonGuionnetZeitouniBook},  \cite[Section 7.1.3]{Forrester-LogGas}. The motivation for this definition comes from the appearance of the Airy point process in the asymptotics of the largest eigenvalues of complex Hermitian random matrices. The same limiting object appears for many random matrix ensembles (this is a manifestation of the \emph{universality} phenomenon), but let us present the result for GUE, which extends \eqref{eqn:TWF2}.

\begin{theorem} \label{thm:GUE_to_Airy}
  Consider $N\times N$ Hermitian matrix $A = \frac{1}{2}\left(X + X^*\right)$, where~$X$ is a random matrix whose entries are i.i.d.\ complex Gaussian random variables~$\mathcal{N}(0,1) + \i \mathcal{N}(0, 1)$, with independent real and imaginary parts. Let $\lambda_1\le \lambda_2\le \dots\le \lambda_N$ be eigenvalues of $A$. Then, in the sense of convergence in finite-dimensional distributions:
\begin{equation}\label{eqn:Airy_as_limit}
 \lim_{N\to\infty} \left\{\frac{\lambda_{N+1-i} - 2 \sqrt{N}}{N^{-\frac{1}{6}}}\right\}_{i=1}^{\infty} \stackrel{d}{=} \{\mathfrak{a}_i\}_{i=1}^{\infty}.
\end{equation}
\end{theorem}
The original proof of \eqref{eqn:Airy_as_limit} going back to \cite{tracy1993level} uses the machinery of the determinantal point process, see the aforementioned textbooks for the details. We do not present a proof of Theorem \ref{thm:GUE_to_Airy} here. Let us, however, remark that the asymptotics \eqref{eqn:Airy_as_limit} can be alternatively obtained in exactly the same way as in our proof of Theorem \ref{thm:corr_convergence}: by writing moments of random Laplace transforms of the GUE as contour integrals and then making the steepest descent analysis in the result. In this way, if one wished, it could have been possible to completely avoid any mention of determinantal point processes and only establish the fact that asymptotics of the six-vertex model are governed by the same object as asymptotics of the GUE. In such an approach, Definition \ref{def:Airy_PP} of the Airy point process is replaced by \eqref{eq_Airy_Laplace}.

Our next task is to show that \eqref{eq_Airy_Laplace} is indeed equivalent to the standard definition used in random matrix textbooks. To show this, we first need to compute the Laplace transforms of the correlation functions.

\begin{proposition}[\cite{BG2016_Airy_moments}]\label{prop:airy_moments}
The Laplace transform of the k-point correlation function of the Airy process~$\rho_k^{\Airy}(x_1,\dots, x_k)$ is given by
\begin{multline}\label{eqn:airy_laplace}
 \int_{\mathbb{R}} \cdots \int_{\mathbb{R}} e^{s_1 x_1 + \cdots + s_k x_k} \rho_k^{\Airy}(x_1,\dots, x_k) \prod_{i=1}^k dx_i   = \frac{e^{\sum_{i=1}^k s_i^3/12}}{(2 \pi \i)^k} \int_{\i \; \mathbb{R}} \frac{d z_1}{s_1}\cdots \int_{\i \; \mathbb{R}} \frac{d z_k}{s_k}  \\
 \times e^{\sum_{i=1}^k s_i z_i^2}  \prod_{1 \leq i < j \leq k} \frac{z_j - \frac{s_j}{2} - z_i + \frac{s_i}{2}}{z_j + \frac{s_j}{2} - z_i + \frac{s_i}{2}} \frac{z_j + \frac{s_j}{2} - z_i - \frac{s_i}{2}}{z_j - \frac{s_j}{2} - z_i - \frac{s_i}{2}}.
 \end{multline}
  The integration contours above are oriented from $-\i \infty$ to $+\i \infty$.
\end{proposition}
\begin{proof}
The essential identity is
$$\int_{-\infty}^{+\infty} e^{x z} \Ai(z+a) \Ai(z+b) \; dz =
\frac{1}{2 \sqrt{\pi x}}\exp\left(\frac{x^3}{12} - \frac{a+b}{2} x- \frac{(a-b)^2}{4 x} \right)$$
which can be found in~\cite[Lemma 2.6]{okounkov2002generating}.

From here, the proof in~\cite{BG2016_Airy_moments} proceeds as follows. Using the definition of the Airy kernel one can compute directly that we have
\begin{multline}\label{eqn:cycle}
\mathcal{E}(s_1,\dots, s_n) \defeq \int_{\R^n} e^{c \cdot z} \prod_{i=1}^n K_{\Airy}(z_i, z_{i+1}) \; dz \\
= \frac{1}{2^n \pi^{n/2}} \frac{e^{\sum s_i^3/12}}{\prod_{i=1}^n \sqrt{s_i}}
\int_{d_1 \geq 0} \cdots \int_{d_n \geq 0} \exp \left(-\sum_{i=1}^n \frac{(d_i-d_{i+1})^2}{4 s_i} - \sum_{i=1}^n \frac{d_i + d_{i+1}}{2} s_i \right) \prod_{i=1}^n d d_i
\end{multline}
where we use the convention that~$z_{n+1} = z_1$ and~$d_{n+1} = d_1$. Using Gaussian integrals in new variables~$z_1,\dots, z_n$,
\begin{multline}\label{eqn:cycle2}
\mathcal{E}(s_1,\dots, s_n) \
= \frac{1}{2^n \pi^{n}} e^{\sum s_i^3/12}
\int_{d_1 \geq 0} \cdots \int_{d_n \geq 0} \\ \int_{z_1 \in \mathbb{R}} \cdots \int_{z_n \in \mathbb{R}}
 \exp \left(\sum_{i=1}^n(
-s_i z_i^2 + \i (z_i - z_{i+1}) d_i - (s_i + s_{i+1}) d_i/2)
 \right) \prod_{i=1}^n d d_i \prod_{i=1}^n d z_i.
\end{multline}
Note that to get to~\eqref{eqn:cycle2} above we also must first make the variable swap~$d_i \rightarrow d_{n-i+1}$ in the expression~\eqref{eqn:cycle}.

Now, one can evaluate~\eqref{eqn:cycle} as
\begin{equation}\label{eqn:cycle3}
\frac{e^{\sum x_i^3/12}}{(2 \pi)^n} \int_{z_1 \in \R} \cdots \int_{z_n \in \R} \exp(- \sum s_i z_i^2) \prod_{i=1}^n \frac{d z_i}{-\i (z_i - z_{i+1}) + \frac{s_i+s_{i+1}}{2}}.
\end{equation}

 Using the definition of~$\rho_k^{\Airy}(x_1,\dots, x_k)$ as a determinant, one can expand out the LHS of~\eqref{eqn:airy_laplace} as a sum over permutations. Then for each cycle in each permutation, one can use Equation~\eqref{eqn:cycle} to obtain a formula for the integrals over variables in the product~$\prod_{i \in \text{ cycle}}K_{\Airy}(x_i, x_{\sigma(i)})$, and then one can group the resulting terms back together to obtain
 $$
 \frac{e^{\sum x_i^3/12}}{(2 \pi)^n} \int_{z_1 \in \R} dz_1 \cdots \int_{z_n \in \R} dz_n \exp(- \sum s_i z_i^2) \det \left[ \frac{1}{-\i z_i  + \frac{s_i}{2} +(\i z_{j} + \frac{s_j}{2}) }\right]_{i,j=1}^n .
 $$
The determinant in the integrand can be simplified using the Cauchy determinant formula of Lemma \ref{Lem:Cauchy_det}, and the result of this can be identified with the RHS of~\eqref{eqn:airy_laplace}.
\end{proof}

Now we can prove \eqref{eq_Airy_Laplace}.

\begin{proposition}[\cite{ahn2020airy}]\label{prop:airy_observable}
Let~$\{\mathfrak{a}_i\}$ be the Airy point process of Definition \ref{def:Airy_PP}, and let~$s_1 ,\dots, s_k > 0$. For any~$v_1 < \cdots < v_k$ satisfying \footnote{Note that in \cite[Proposition B.1]{ahn2020airy}, the corresponding conditions (B.2) on $v_1,\dots, v_k$ have $i$ and $j$ swapped compared to \eqref{eqn:vconds}: we believe this to be a minor typo.}
\begin{equation}\label{eqn:vconds}
v_j - \frac{s_j}{2} > v_i + \frac{s_i}{2}, \qquad 1 \leq i < j \leq k,
\end{equation}
 we have
\begin{multline}\label{eqn:Laplace_observable}
 \mathbb{E}\left[\prod_{m=1}^k \sum_i \exp(s_m \mathfrak{a}_i)\right]   = \frac{e^{\sum_{i=1}^k s_i^3/12}}{(2 \pi \i)^k} \int_{\i \; \mathbb{R} + v_1} \frac{d z_1}{s_1}\cdots \int_{\i \; \mathbb{R} + v_k} \frac{d z_k}{s_k}  \\
 \times e^{\sum_{i=1}^k s_i z_i^2}  \prod_{1 \leq i < j \leq k} \frac{z_j - \frac{s_j}{2} - z_i + \frac{s_i}{2}}{z_j + \frac{s_j}{2} - z_i + \frac{s_i}{2}} \frac{z_j + \frac{s_j}{2} - z_i - \frac{s_i}{2}}{z_j - \frac{s_j}{2} - z_i - \frac{s_i}{2}}.
 \end{multline}
 The integration contours above are oriented from $-\i \infty$ to $+\i \infty$.
\end{proposition}
\begin{proof}
This appears as Proposition B.1 in~\cite{ahn2020airy}, and we outline the proof given there. We will proceed by induction on~$k$, with the~$k=1$ case matching~$k=1$ in Proposition~\ref{prop:airy_moments}.

The idea of the inductive step, which proves the statement for~$k>1$, is as follows. Note that the expression on the left hand side can be written as
\begin{equation}\label{eqn:k_statement}
\E \left[ \sum_{j_1,\dots, j_k} e^{s_1 \mathfrak{a}_{j_1} + \cdots + s_k\mathfrak{a}_{j_k}} \right].
\end{equation}
Note that the above equals
$$
\E \left[ \sum_{j_1 \neq j_2, j_3 \dots, j_k} e^{s_1 \mathfrak{a}_{j_1} + \cdots + s_k\mathfrak{a}_{j_k}} \right]
 + \E \left[ \sum_{j_1 = j_2, j_3 \dots, j_k} e^{(s_1+s_2) \mathfrak{a}_{j_1} + s_3 \mathfrak{a}_{j_3} \cdots + s_k\mathfrak{a}_{j_k}} \right].
$$
The second term above is of the form~\eqref{eqn:k_statement}, but with smaller~$k$. The first term can be further decomposed into terms where~$j_1, j_2, j_3$ are distinct, plus terms where~$j_3 = j_1$ or~$j_3 = j_2$. Thus, we begin to see that to compute~\eqref{eqn:k_statement}, we can use the inductive hypothesis, together with backward induction on the number of indices which are distinct, using Proposition~\ref{prop:airy_moments} as the base case for each~$k$.

Now we will implement this in detail. Throughout the rest of the proof we fix~$0 < v_1 < \cdots < v_k$ satisfying~\eqref{eqn:vconds}; there is no loss of generality in assuming positivity of the~$v_i$.

To prove the statement for~$k>1$ given the statement for all smaller~$k$, we proceed by induction on the index~$n = k - \ell+1$, where~$\ell$ is such that~$j_1, \dots, j_\ell$ are distinct, in order to prove that
\begin{multline}
\E \left[ \sum_{j_1, \dots j_\ell \text{ distinct}, j_{\ell+1} , \dots, j_k} e^{s_1 \mathfrak{a}_{j_1} + \cdots + s_k\mathfrak{a}_{j_k}} \right] = \\I(s_1,\dots, s_{\ell}; s_{\ell+1}, \dots, s_{k})
\defeq \frac{e^{\sum_{i=1}^k s_i^3/12}}{(2 \pi \i)^k} \int_{\i \; \mathbb{R} + v_k } \frac{d z_k}{s_k} \cdots \int_{\i \; \mathbb{R} + v_{\ell+1}} \frac{d z_{\ell+1}}{s_{\ell+1}} \int_{\i \; \mathbb{R} } \frac{d z_{\ell}}{s_{\ell}}   \cdots \int_{\i \; \mathbb{R} } \frac{d z_1}{s_1}  \\
 \times e^{\sum_{i=1}^k s_i z_i^2}  \prod_{1 \leq i < j \leq k} \frac{z_j - \frac{s_j}{2} - z_i + \frac{s_i}{2}}{z_j + \frac{s_j}{2} - z_i + \frac{s_i}{2}} \frac{z_j + \frac{s_j}{2} - z_i - \frac{s_i}{2}}{z_j - \frac{s_j}{2} - z_i - \frac{s_i}{2}}.
\end{multline}
For each fixed~$k$, as the base case~$n=1$, which means~$\ell = k$, we have
$$\E \left[ \sum_{j_1, \dots j_k \text{ distinct}} e^{s_1 \mathfrak{a}_{j_1} + \cdots + s_k\mathfrak{a}_{j_k}} \right]  = I( s_1,\dots s_k; )$$
by Proposition~\ref{prop:airy_moments}. For the induction step (for the index~$n$) it suffices to show that the values of~$I$ satisfy the recurrence
\begin{multline}\label{eqn:Irec}
I(s_1,\dots, s_\ell; s_{\ell+1}, \dots, s_k) = I(s_1,\dots, s_{\ell+1}; s_{\ell+2}, \dots, s_k) \\
+ \sum_{m = 1}^\ell I(s_1,\dots, s_m+s_{\ell+1} , \dots, s_{\ell}; s_{\ell+2}, \dots, s_k)
\end{multline}
since the expressions
$$M(s_1,\dots, s_{\ell}; s_{\ell+1},\dots, s_k) = \E \left[ \sum_{j_1, \dots j_\ell \text{ distinct}, j_{\ell+1} , \dots, j_k} e^{s_1 \mathfrak{a}_{j_1} + \cdots + s_k\mathfrak{a}_{j_k}} \right] $$
satisfy the exact same recurrence.

To show~\eqref{eqn:Irec}, we drag the~$z_{\ell+1}$ contour from~$\i \R + v_{\ell+1}$ to~$\i \R$, and keep track of residues we cross as we drag it. This gives us
%\begin{align*}
%I(s_1,\dots, s_{\ell}; s_{\ell+1}, \dots, s_{k})  &=  \frac{e^{\sum_{i=1}^k s_i^3/12}}{(2 \pi \i)^k} \int_{\i \; \mathbb{R} + v_k } \frac{d z_k}{s_k} \cdots \int_{\i \; \mathbb{R} } \frac{d z_{\ell+1}}{s_{\ell+1}} \int_{\i \; \mathbb{R} } \frac{d z_{\ell}}{s_{\ell}}   \cdots \int_{\i \; \mathbb{R} } \frac{d z_1}{s_1}  \\
%& \times e^{\sum_{i=1}^k s_i z_i^2}  \prod_{1 \leq i < j \leq k} \frac{z_j - \frac{s_j}{2} - z_i + \frac{s_i}{2}}{z_j + \frac{s_j}{2} - z_i + \frac{s_i}{2}} \frac{z_j + \frac{s_j}{2} - z_i - \frac{s_i}{2}}{z_j - \frac{s_j}{2} - z_i - \frac{s_i}{2}} \\
%&+\frac{1}{s_{k-\ell}} \sum_{m = k-\ell+1}^k \frac{e^{\sum_{i=1}^k s_i^3/12}}{(2 \pi \i)^k}\int_{\i \; \mathbb{R} } \frac{d z_k}{s_k} \cdots \int_{\i \; \mathbb{R}} \frac{d z_{k-\ell+1}}{s_{k-\ell+1}} \\
%& \int_{\i \; \mathbb{R} + v_{k-\ell-1}} \frac{d z_{k-\ell-1}}{s_{k-\ell-1}}   \cdots \int_{\i \; \mathbb{R} + v_1} \frac{d z_1}{s_1}   \\
%& \times \text{Res}_{z_{k-\ell} - \frac{s_{k-\ell}}{2} = z_m + \frac{s_m}{2}} e^{\sum_{i=1}^k s_i z_i^2}  \prod_{1 \leq i < j \leq k} \frac{z_j - \frac{s_j}{2} - z_i + \frac{s_i}{2}}{z_j + \frac{s_j}{2} - z_i + \frac{s_i}{2}} \frac{z_j + \frac{s_j}{2} - z_i - \frac{s_i}{2}}{z_j - \frac{s_j}{2} - z_i - \frac{s_i}{2}} .
% \end{align*}
\begin{align*}
I(s_1,\dots, s_{\ell}; s_{\ell+1}, \dots, s_{k})  &=  \frac{e^{\sum_{i=1}^k s_i^3/12}}{(2 \pi \i)^k} \int_{\i \; \mathbb{R} + v_k } \frac{d z_k}{s_k} \cdots \int_{\i \; \mathbb{R}+ v_{\ell+2}  } \frac{d z_{\ell+2}}{s_{\ell+2}} \int_{\i \; \mathbb{R} } \frac{d z_{\ell+1}}{s_{\ell+1}}   \cdots \int_{\i \; \mathbb{R} } \frac{d z_1}{s_1}  \\
& \times e^{\sum_{i=1}^k s_i z_i^2}  \prod_{1 \leq i < j \leq k} \frac{z_j - \frac{s_j}{2} - z_i + \frac{s_i}{2}}{z_j + \frac{s_j}{2} - z_i + \frac{s_i}{2}} \frac{z_j + \frac{s_j}{2} - z_i - \frac{s_i}{2}}{z_j - \frac{s_j}{2} - z_i - \frac{s_i}{2}} \\
&+\frac{1}{s_{\ell+1}} \sum_{m =1}^\ell \frac{e^{\sum_{i=1}^k s_i^3/12}}{(2 \pi \i)^{k-1}}\int_{\i \; \mathbb{R} + v_k } \frac{d z_k}{s_k} \cdots \int_{\i \; \mathbb{R} + v_{\ell+2}} \frac{d z_{\ell+2}}{s_{\ell+2}} \\
& \int_{\i \; \mathbb{R} } \frac{d z_{\ell}}{s_{\ell}}   \cdots \int_{\i \; \mathbb{R} } \frac{d z_1}{s_1}   \\
& \times \text{Res}_{z_{\ell+1} - \frac{s_{\ell+1}}{2} = z_m + \frac{s_m}{2}} e^{\sum_{i=1}^k s_i z_i^2}  \prod_{1 \leq i < j \leq k} \frac{z_j - \frac{s_j}{2} - z_i + \frac{s_i}{2}}{z_j + \frac{s_j}{2} - z_i + \frac{s_i}{2}} \frac{z_j + \frac{s_j}{2} - z_i - \frac{s_i}{2}}{z_j - \frac{s_j}{2} - z_i - \frac{s_i}{2}} .
 \end{align*}
 The residue at~$z_{\ell+1}  = z_m +  \frac{s_{\ell+1}}{2} + \frac{s_m}{2}$ can be seen to be, after some simplification,
 \begin{multline}
 e^{ s_{\ell+1} ( \frac{s_{\ell+1}}{2} + z_m + \frac{s_m}{2})^2 + \sum_{i\neq \ell+1} s_i z_i^2} \frac{s_{\ell+1} s_m}{s_{\ell+1}+s_m}  \prod_{(i, j) \in S_1} \frac{z_j - \frac{s_j}{2} - z_i + \frac{s_i}{2}}{z_j + \frac{s_j}{2} - z_i + \frac{s_i}{2}} \frac{z_j + \frac{s_j}{2} - z_i - \frac{s_i}{2}}{z_j - \frac{s_j}{2} - z_i - \frac{s_i}{2}}\\
\prod_{j \in S_2} \frac{z_j - \frac{s_j}{2} - z_m - \frac{s_m}{2}}{z_j + \frac{s_j}{2} - z_m - \frac{s_m}{2}} \frac{ z_j + \frac{s_j}{2}  - z_m - s_{\ell+1} - \frac{s_m}{2} }{z_j - \frac{s_j}{2} - z_m - s_{\ell+1} - \frac{s_m}{2}}.
 \end{multline}
 In the first product above,~$S_1$ is the set of pairs~$i < j$ such that neither~$i$ nor~$j$ equals~$\ell+1$, and in the second product~$S_2 = \{1,\dots, k\} \setminus \{\ell+1, m\}$.

Next, cancelling some terms shared between the two products and simplifying yields
 \begin{multline}
 e^{ s_{\ell+1} ( \frac{s_{\ell+1}}{2} + z_m + \frac{s_m}{2})^2 + \sum_{i\neq \ell+1} s_i z_i^2} \frac{s_{\ell+1} s_m}{s_{\ell+1}+s_m}  \prod_{(i, j) \in \tilde{S}_1} \frac{z_j - \frac{s_j}{2} - z_i + \frac{s_i}{2}}{z_j + \frac{s_j}{2} - z_i + \frac{s_i}{2}} \frac{z_j + \frac{s_j}{2} - z_i - \frac{s_i}{2}}{z_j - \frac{s_j}{2} - z_i - \frac{s_i}{2}}\\
\prod_{j \in S_2} \frac{z_j - \frac{s_j}{2} - z_m + \frac{s_m}{2}}{z_j + \frac{s_j}{2} - z_m + \frac{s_m}{2}} \frac{ z_j + \frac{s_j}{2}  - z_m - s_{\ell+1} - \frac{s_m}{2} }{z_j - \frac{s_j}{2} - z_m - s_{\ell+1} - \frac{s_m}{2}}.
 \end{multline}
 where $\tilde{S}_1$ is the set of $(i, j)$ such that neither $i$ nor $j$ is in the set $\{\ell+1, m\}$.

 After making the variable change~$z_m \rightarrow z_m-\frac{s_{\ell+1}}{2}$ and doing some simplification and re-arranging of prefactors, we obtain exactly the integrand of the~$m^{th}$ term of the summation in \eqref{eqn:Irec} above. The only difference is that our integration contour is now~$\Re(z_m) = \frac{s_{\ell+1}}{2}$ instead of~$\Re(z_m) = 0$. One can check that if we choose our original $v_{\ell+1} < \cdots <v_{k}$ far enough apart, then we can drag the~$z_m$ contour back to~$\Re(z_m) = 0$ without crossing any additional poles.

 Thus we have shown the recursion~\eqref{eqn:Irec}, and we are done. \qedhere

\end{proof}

The last ingredient which was implicitly used in the proof of Theorem \ref{thm:corr_convergence} is the topological question of equivalence of different modes of convergence. In that proof, we are establishing the convergence of moments of Laplace transforms \eqref{eqn:res}. Why does this imply more standard convergence of the finite-dimensional distributions \eqref{eq_x8}? For a justification we need a separate argument.\footnote{One needs certain care in producing this argument. An unpleasant feature is that the random variable $ \sum_{i \geq 1} \exp(s \mathfrak{a}_i )$ has fast-growing moments, so that the associated moments problem is not uniquely determined for any fixed $s>0$. However, our access to all $s>0$ simultaneously still allows to reconstruct the point process $\{\mathfrak{a}_i\}_{i=1}^\infty$.}   The arguments of this type were developed previously in the random-matrix literature: they were first employed by Soshnikov to study eigenvalues of random matrices near the edge of the spectrum in \cite{soshnikov1999universality}, see also more recent work~\cite{Sodin_2014}.

%Suppose that~$\rho_k = \rho_k^{\Airy}$, and that we would like to prove distributional convergence of the first particle. First, we must show (a) the convergence of each term
%$$\int_{x_1, \dots, x_k  > s} \rho_k^N( d x_1,\dots, d x_k) \rightarrow
%\int_{x_1, \dots, x_k  > s} \rho_k^{\Airy}( d x_1,\dots, d x_k).$$
%
%
%Then, we must argue (b) that we can swap the limit with the series and as such conclude the convergence
%$$P_N(a_1 \leq s)  \rightarrow P(\mathfrak{a}_1 \leq s).$$
%
%
%One way to do both of these steps is to give tail bounds on~$\rho_k^N$ which hold uniformly in~$N$. If these are good enough, then one can show that the limit and series can be commuted. To point out one example, this argument is carried out in~\cite[Section 3]{Johansson2005arctic} to study fluctuations of random domino tilings near the so-called ``arctic boundary'' (c.f. the discussion after Theorem~\ref{thm:Ltriangle} in the introduction). See in particular the proof of Lemma 3.2 there.
%
%In the setting of the Schur measure we consider in the text, it is also possible to proceed in this way. It is a known fact that for~$\lambda$ sampled from the Schur measure, the point process made up of points~$\{\lambda_i + N - i\}_{i=1}^N$, and thus also its complement~$\mathbb{Z}_{\geq 0} \setminus \{\lambda_i + N - i\}$, is determinantal, and the kernel defining the correlation functions has explicit contour integral formulas. Using these formulas, one could give the necessary bounds.

The following proposition is what we need to complete the proof of Theorem~\ref{thm:corr_convergence}.

\begin{proposition}\label{prop:sufficient_conv}
Let~$\{a_1 > a_2 > \cdots\}$,~$a_i = a_i(N) \in \mathbb{R}$, be a sequence of point processes, and suppose that for any~$k \geq 1$ and~$s_1,\dots, s_k > 0$ we have the convergence
\begin{equation}\label{eqn:lap_conv}
\mathbb{E}\left[\prod_{m=1}^k  \sum_{i \geq 1} ( 1-s_m N^{-1/3})^{- a_i N^{1/3}} \right] \rightarrow  \mathbb{E}\left[\prod_{m=1}^k \sum_{i \geq 1} \exp(s_m \mathfrak{a}_i )\right]
\end{equation}
for some random point process~$(\mathfrak{a}_1 > \mathfrak{a}_2 > \cdots)$. Suppose in addition that the limiting point process~$\{\mathfrak{a}_i\}_{i=1}^\infty$ is uniquely determined by its correlation functions.\footnote{This is true for the Airy Point process, see \cite{Soshnikov2000}.} Then for each~$k \geq 1$,
$$(a_1, \dots, a_k) \stackrel{d}{\rightarrow } (\mathfrak{a}_1, \dots, \mathfrak{a}_k ) .$$
\end{proposition}

%\begin{proof}[Sketch of proof]
%We follow the argument given in~\cite[Section 6]{gorin2018stochastic}. The condition~\eqref{eqn:lap_conv} implies the distributional convergence
%$$\sum_{i \geq 1} \exp(c a_i) \rightarrow \sum_{i \geq 1} \exp(c \mathfrak{a}_i )$$
%jointly for all~$c \in \R_{> 0} \cap \mathbb{Q}$.
%Thus, we may assume, by Skorokhod's representation theorem, that almost surely~$\sum_{i \geq 1} \exp(c a_i) \rightarrow \sum_{i \geq 1} \exp(c \mathfrak{a}_i )$ for all~$c \in \R_{> 0} \cap \mathbb{Q}$. This in turn implies the almost sure convergence~$a_i \rightarrow \mathfrak{a}_i$. For details on the proof of this final claim, see~\cite[Section 5]{Sodin_2014}.
%\end{proof}

\begin{proof}[Sketch of the proof]
We outline the main steps of the argument, and we refer the reader to~\cite[Section 5]{soshnikov1999universality} and~\cite[Sections 5, 6]{Sodin_2014} for further details. The first step is to observe that the law of the random point process, and thus, in particular, each finite dimensional marginal~$(a_1,\dots, a_k)$, is determined by the joint laws of random variables
$$N_{[a,b]} \defeq \# \{a_i \in [a, b]\} . $$
Hence, we must prove convergence in distribution of these random variables.

A standard proof of distributional convergence proceeds in two steps: proving tightness, and identification of the limit. Tightness follows immediately from~\eqref{eqn:lap_conv}: it is implied by the boundedness of the left hand side for large enough~$N$. For the identification, we reconstruct the correlation functions of~$\{\mathfrak{a}_i\}_{i=1}^\infty$ from the expectations $\mathbb{E}\left[\prod_{m=1}^k \sum_{i \geq 1} \exp(s_m \mathfrak{a}_i )\right]$ on the right hand side of~\eqref{eqn:lap_conv} (through a reversal of the arguments in Propositions \ref{prop:airy_moments} and \ref{prop:airy_observable} interconnecting the correlation functions and expected Laplace transforms). Then, we use a general theorem which says that correlation functions uniquely determine particle counts, see, e.g., \cite{Soshnikov2000}.
\end{proof}

\bibliographystyle{alpha}
\bibliography{bib}

\end{document}